\documentclass[reqno,a4paper]{amsart}
\usepackage{amssymb,stmaryrd}
\usepackage{amsfonts}
\usepackage{amstext}
\usepackage{algorithmic}
\usepackage{algorithm}
\usepackage{color}
\usepackage{graphicx}
\usepackage{epstopdf}
\usepackage[all]{xy}
\usepackage{MnSymbol}
\numberwithin{equation}{section}
\newtheorem{theorem}{Theorem}[section]
\newtheorem{lemma}[theorem]{Lemma}
\newtheorem{proposition}[theorem]{Proposition}
\newtheorem{corollary}[theorem]{Corollary}

\theoremstyle{definition}
\newtheorem{definition}[theorem]{Definition}
\newtheorem{example}[theorem]{Example}

\theoremstyle{remark}
\newtheorem{remark}[theorem]{\bf{Remark}}

\newcommand{\R}{{\mathbb{R}}}

\newcommand{\C}{{\mathbb{C}}}

\newcommand{\<}{{\langle}}

\renewcommand{\>}{{\rangle}}

\newcommand{\CD}{{\mathcal{D}}}

\newcommand{\cg}{{\mathfrak g}}
\newcommand{\ch}{{\mathfrak{h}}}

\newcommand{\tens}{\otimes}

\newcommand{\id}{{\rm id}}
\newcommand{\extd}{{\rm d}}
\newcommand{\del}{{\partial}}
\newcommand{\eps}{\epsilon}
\renewcommand{\imath}{\mathrm{i}}

\newcommand{\rh}{\rho(\ch)}

\allowdisplaybreaks
\usepackage{scalerel,stackengine}
\stackMath
\newcommand\reallywidehat[1]{%
\savestack{\tmpbox}{\stretchto{%
  \scaleto{%
    \scalerel*[\widthof{\ensuremath{#1}}]{\kern-.6pt\bigwedge\kern-.6pt}%
    {\rule[-\textheight/2]{1ex}{\textheight}}
  }{\textheight}%
}{0.5ex}}%
\stackon[1pt]{#1}{\tmpbox}%
}

\begin{document}

\title{Generally covariant quantum mechanics}



\author{Edwin Beggs and Shahn Majid}\thanks{
{\it Authors to whom correspondence should be addressed:} e.j.beggs@swansea.ac.uk  and s.majid@qmul.ac.uk}
\date{Nov 2025 Ver 1.3}
\address{Department of Mathematics, Bay Campus, Swansea University, SA1 8EN, UK; Queen Mary University of London, 
School of Mathematical Sciences, Mile End Rd, London E1 4NS, UK}


\begin{abstract} We obtain generally covariant operator-valued geodesic equations on a pseudo-Riemannian manifold $M$ as part of the construction of quantum geodesics on the algebra $\CD(M)$ of differential operators.  Geodesic motion arises here as an associativity condition for a certain form of first order differential calculus on this algebra in the presence of curvature. The corresponding Schr\"odinger picture has wave functions on spacetime and proper time evolution by the Klein-Gordon operator, with stationary modes being  solutions of the Klein-Gordon equation. As an application, we describe gravatom solutions of the Klein-Gordon equations around a Schwarzschild black hole, i.e. gravitationally bound states which far from the event horizon resemble atomic states with the black hole in the role of the nucleus. The spatial eigenfunctions exhibit probability density banding as for higher orbital modes of an ordinary atom, but of a fractal nature approaching the horizon.  \end{abstract}

\maketitle

\section{Introduction}

Noncommutative geometry\cite{Con,BegMa}, and in particular the theory of quantum geodesics  \cite{Beg:geo,BegMa:cur, BegMa:geogra, LiuMa}, was recently applied in \cite{BegMa:geo} to the Heisenberg algebra $A$ in quantum mechanics viewed as a noncommutative phase space. This work equipped  $A$ with a certain carefully chosen exterior algebra of differential forms $\Omega_A$ defined by a choice of Hamiltonian, and a certain generalised quantum metric in $\Omega^1\tens_A\Omega^1$ such that quantum geodesic flow with parameter $t$ recovers the standard Schr\"odinger equation. A generalised quantum metric here means we  assume neither symmetry nor nondeneracy, and indeed the one in \cite{BegMa:geo} was antisymmetric in the $\hbar\to 0$ limit (so more like a  symplectic or contact structure)  and also had a kernel (due to having one dimension more in the calculus) which encodes the Hamilton-Jacobi equations of motion. Differential forms mean that as well as the Heisenberg algebra generators $x^\mu, p_\nu$, we also have their differentials $\extd x^\mu,\extd p_\nu$ (and an unexpected additonal 1-form $\theta'$) and all the tools of noncommutative differential geometry.  The work included also the relativistic case where now the geodesic time parameter is an external proper time, and extended it to an electromagnetic Heisenberg algebra applicable to spacetime with a background $U(1)$-gauge field. We will say later what we mean by a quantum geodesic, but in practical terms it amounts to a flow generated by the (minimally coupled) Klein-Gordon operator, and in the electromagnetic case we saw how the Lorentz force law appears naturally at this  level. Although we did not claim new physics from such constructions, they provided a novel perspective using new tools. 

In the present sequel, we now aim to extend these ideas to the generally covariant setting where the Heisenberg algebra is, we propose, replaced by the algebra $\CD(M)$ of differential operators on a smooth manifold $M$ and we aim to obtain geodesic motion rather than the Lorentz force, as a flow on this algebra. For our purposes, we consider $\CD(M)$  as the algebra generated by functions $f$ and vector fields $X$ with cross-relations $[X,f]=\lambda\, X(f)$, where $\lambda=-\imath\hbar$ for the application we have in mind, hence looks like the usual Heisenberg algebra in any local coordinates. We will explain in the preliminaries, see Section~\ref{sec:back}, how this is equivalent to the usual definition of $\CD(M)$.  Recent interest in this algebra lies in the fact that it forms a Hopf algebroid, see \cite{Xu} where further deformed versions of it were introduced. It is also used in algebraic geometry and homological approaches to physics\cite{Ver}. Importantly for us, even though it is noncommutative, this algebra is defined globally and when taken in the above form, we can work with it using the same tools of tensor calculus as familiar in General Relativity (GR), allowing us to build structures on it in a globally defined manner when $M$ is equipped with a Riemannian or pseudo-Riemannian metric $g$. Thus, although, as is standard practice in GR, we will frequently write down expressions in a local coordinate chart with coordinates $x^\mu$ and coordinate vector fields $\del_\mu$, this is mainly for convenience and we do not assume a global coordinate chart. Indeed, $x^\mu, \del_\mu$ are not precisely elements of $\CD(M)$ but merely locally-defined representatives and the relevant equations can be rewritten in terms of global quantities. Although standard in GR, such methods are not standard in quantum mechanics, which normally goes via the theory of linear operators on Hilbert spaces. The two are connected via a canonical {\em Schr\"odinger representation} $\rho:\CD(M)\to {\rm Lin}(L^2(M))$ in which functions act by multiplication and vector fields act by $\lambda$ times differentiation along the vector field, as defined on the smooth functions of compact support $C_c^\infty(M)\subset L^2(M)$. By ${\rm Lin}(L^2(M))$ we mean possibly unbounded operators with associated domains, but the ones arising in our case all have domains containing at least $C_c^\infty(M)$. Moreover, we work with complexified functions and vector fields, then $\CD(M)$ is  a $*$-algebra and the $*$-operation is required to map to the adjoint of the associated operator in the sense $(\phi,\rho(a)\psi)= (\rho(a^*)\phi,\psi)$ for all $\phi,\psi\in C^\infty_c(M)$, where $a\in \CD(M)$ and  $(\ ,\ )$ denotes the $L^2(M)$ inner product.  The basic idea of the paper is to transfer quantum mechanical ideas from the operator side in the image of $\rho$ to the smooth geometric side of $\CD(M)$ where we can aim to use tools of (noncommutative) differential geometry. We view this as motivation for our  definitions on $\CD(M)$, with the further study of their connection to the operator side a topic for further work, as it would need rather different tools of functional analysis.

In the case where $M$ is space and time is external, the picture we have in mind does appear to be in line with the time-independent case of what is usually called quantum mechanics on curved spacetime in the sense of \cite{Witt}, except that this is usually done as a Schr\"odinger equation whereas we lift the evolution to the $\CD(M)$ level. But when $M$ is spacetime, the physical interpretation of our results is less clear for two reasons, both of which we start to look at in  examples and applications towards the end of the paper. The first is that `wave functions' $\psi$ on spacetime are not something usually considered in physics. The second is that whereas the concept of one particle moving along a geodesic entails a proper time $s$ as geodesic parameter, we are effectively extending this to some kind of collective proper time with respect to which evolution take places. Hence, our quantum mechanics-like evolution with respect to $s$ is not actual quantum mechanics in any conventional sense, although it has parallels. What has sometimes also been referred to as quantum mechanics in curved spacetime is fields obeying the minimally coupled Klein-Gordon or Dirac equation viewed, where possible, as quantum-mechanics-like with respect to a global time coordinate. See for example \cite{Don} for an early work on this view around a Schwarzschild black hole. There are issues in the Dirac case about what is left of coordinate invariance, see \cite{Arm} and related works. In our case, by focusing on stationary states under $s$ we arrive at a similar point of considering solutions of the Klein-Gordon equation (but not for a particle of fixed mass) which is then comparable to these works. Finally, this circle of ideas should not be confused with quantum field theory on curved spacetime, which is quite well understood using operator algebra methods at least when $M$ is globally hyperbolic (the main issue in general is lack of a unique vacuum state), see \cite{HolWal, Kay} for relatively recent reviews. We do, however, obtain classical geodesic motion in the $\hbar\to 0$ commutative limit, and hence, with the above caveats, the present work fills a certain quantum-mechanics level gap between classical geodesics on $M$ as in GR, and quantum field theory on spacetime $M$ using operator algebra methods.

We turn now to the principal results of the paper and where they can be found. In order not to overpromise, let us say right away that the key thing we find, in Section~\ref{sec:jac} is a curvature obstruction to an associative differential exterior algebra $\Omega_{\CD(M)}$ on $\CD(M)$ of the type needed. This obstruction, which appears in a breakdown of Jacobi identities at order $\lambda^2$, is in line with curvature obstructions in \cite{BegMa:poi} in a different context. It is also at this order that we see the appearance of the Ricci tensor in our resulting commutation relations
\begin{align} \label{Xxi}
 [X,\widehat\xi ] =\lambda\,(\widehat{ \nabla_X \xi  })  - {\lambda\over m} \theta'\, (g^{\mu\nu}\,\xi_\mu\, \nabla_\nu X) 
-{\lambda^2\over 2m}\, \theta'\big(X^\rho\, \xi_\mu\, g^{\mu\nu}\, R_{\nu\rho} + g^{\mu\nu}\, X^\rho{}_{;\nu}\, \xi_{\mu;\rho}
\big).
\end{align}
Here, $X$ is a vector field, $\xi$ is a 1-form on $M$,  $\hat\xi$ is its image as a 1-form in $\Omega^1_{\CD(M)}$, $\nabla$ is the Levi-Civita connection, also indicated by a semicolon ;  and $R_{\mu\nu}$ is the Ricci curvature. The parameter $m$ will play the role of a particle mass and the element $\theta'$ is a central 1-form on $\CD(M)$ as in \cite{BegMa:geo, Ma:alm} which will be understood as a proper time interval. The way that these commutators emerge is that we ask that the Schr\"odinger representation extends to a representation of the whole exterior algebra
\begin{equation}\label{rho} \rho: \Omega_{\CD(M)}\to {\rm Lin}(L^2(M))\end{equation}
still typically with unbounded operators. Since the composition of operators (where their domains allow) is necessarily associative, $\Omega_{\CD(M)}$ would also have to be associative if $\rho$ were to be injective. However, this is not the case and instead all the terms that lead to nonassociativity of the calculus are in the kernel of $\rho$. On the other hand, searching for $\rho$ to the extent possible does lead to a full set of commutation relations, of which (\ref{Xxi}) is one, between elements of $\CD(M)$ and their differentials in $\Omega^1_{\CD(M)}$. This step by step derivation of a generally covariant set of commutation relations is achieved in Section~\ref{sec:calc}.  The method used is the same as in \cite{BegMa:geo}, namely to write down a natural (covariant) form of Hamiltonian $\ch$ and ask that evolution given by commutator with this matches up with what we expect for the quantum geodesic flow. The Jacobi identity issue is then covered in Section~\ref{sec:jac}. 

To complete the noncommutative geometry behind the quantum geodesic requires a technical component of a suitable `bimodule connection' $\nabla$ with respect to which our geodesic velocity field $\mathfrak{X}$ on $\CD(M)$ in Proposition~\ref{propX} should be autoparallel. This appears to be rather complicated and is deferred to further work. It is expected to exist at least to order $\lambda^2$ as a globalisation of the one in \cite{BegMa:geo} for the flat case, but meanwhile the geodesic flow itself {\em is} still defined in line with $\mathfrak{X}$ found to this order in Section~\ref{sec:geo}. Section~\ref{sec:jac} also identifies natural elements of $\Omega^1_{\CD(M)}$ that are in the kernel (and likely span it over the algebra, although we do not prove this). In the relativistic case without external potential $V$, and working to order $\lambda^2$ and (for convenience) with local coordinates, these are
\begin{align}\label{kerrho1}
&\extd x^\mu - {\theta'\over m}\left(g^{\mu\nu} p_\nu - {\lambda\over 2} \Gamma^\mu\right),\\ \label{kerrho2}
 &\extd p_\mu- {\theta'\over m} \left(\Gamma^\nu{}_{\mu\sigma}g^{\sigma\rho}(p_\nu p_\rho-\lambda\Gamma^\tau{}_{\nu\rho}p_\tau)+{\lambda\over 2}g^{\alpha\beta} \, \Gamma^\nu{}_{\beta\alpha,\mu}\,p_\nu - V_{,\mu}\right), 
 \end{align}
where $p_\mu= \del_\mu$ as a local vector field when viewed locally in $\CD(M)$ and mapping to $\lambda{\del\over\del x^\mu}$ in the Schr\"odinger representation, and $\Gamma^\mu=\Gamma^\mu {}_{\nu\rho}g^{\nu\rho}$ is a contraction of the Christoffel symbols. Therefore, if we set (\ref{kerrho1})-(\ref{kerrho2}) to zero in order to kill the nonassociativity in the calculus, and if we interpret $\theta'=\extd s$  as `proper time' $s$ then we can interpret  (\ref{kerrho1}) as definition of $p_\mu$ in terms of $\extd x^\mu\over\extd s$, in which case (\ref{kerrho2}) becomes
\begin{equation}\label{opgeo} {\extd^2x^\mu\over\extd s^2}=-\Gamma^\mu{}_{\nu\rho}{\extd x^\nu\over\extd s}{\extd x^\rho\over\extd s}+{\lambda\over 2 m}C^\mu{}_\nu {\extd x^\nu\over\extd s}+O(\lambda^2)\end{equation}
(the order $\lambda^2$ term can also be computed), where 
\[ C^{\mu\nu}=-g^{\alpha\beta}(g^{\mu\gamma}\Gamma^\nu{}_{\gamma \alpha,\beta}+g^{\nu\gamma}\Gamma^\mu{}_{\gamma \alpha,\beta})+  g^{\mu j}\Gamma^\nu{}_{;\beta}-g^{\nu \beta}\Gamma^\mu{}_{;\beta}+\Gamma^{\alpha\beta\mu}\Gamma^{\nu}{}_{\alpha\beta}- \Gamma^{\alpha\beta\nu}\Gamma^{\mu}{}_{\alpha\beta},\]
see Proposition~\ref{propC}. The combination of derivatives here is different from that in the curvature, and indeed $C^{\mu\nu}$ does not transform as a tensor (it could possibly be better understood in terms of jets, although we have not done this). Moreover, (\ref{opgeo}) becomes an operator equation in the `Heisenberg picture' when viewed in the Schr\"odinger representation, where these relations hold. The equations (\ref{opgeo}) are coordinate invariant and can be computed in any coordinates, but the separate terms in isolation do not transform simply, both because of  $\Gamma^\mu{}_{\nu\rho}$ and $C^{\mu\nu}$, and because the $\extd x^\mu\over\extd s$ do not commute with functions.  In the non-relativistic version where $M$ is space and $\theta'=\extd t$ for an external time $s=t$, and with an external potential $V$ in the Hamiltonian, we similarly recover noncommutative versions of Hamilton-Jacobi equations of motion on the curved space with order $\lambda$ corrections. Also note that real vector fields are not invariant under $*$ due to a divergence correction (see Section~\ref{orf}) but can be corrected so as to be invariant. Doing this to the $p_i$ associated to local vector fields $\del_i$ and applying the representation $\rho$ then recovers the hermitian momenta in (5.31) of \cite{Witt} in the case where $M$ is space.

From (\ref{opgeo}), we can see as promised that geodesic motion as in conventional GR is contained in our algebraic set-up at zeroth order. Indeed, for $\lambda=0$,  (\ref{kerrho1})-(\ref{kerrho2}) are a standard cotangent bundle approach to geodesics flows as used, for example, in \cite{Cha}. The difference is that we `quantise' this picture by providing order $\lambda$ corrections needed for a coordinate-invariant `Heisenberg picture' on $\CD(M)$ as the global version of the Heisenberg algebra. We shall see that order $\lambda$  is also relevant to the Schr\"odinger representation and Klein-Gordon operator on `wave functions'. We also explained that while the  differential calculus on $\CD(M)$ is nonassociative at order $\lambda^2$, the equations setting (\ref{kerrho1})-(\ref{kerrho2}), i.e. the geodesic equations, are required to kill this associativity obstruction in presence of generic curvature. This is a new `anomaly cancellation' derivation of classical   geodesic motion (rather different from the principle of least action), and is a conceptual outcome of the paper even at the zeroth order. As explained, the physical meaning of the order $\lambda$ level is less clear not least due to the interpretation of the time parameter $s$. On the other hand, some kind of proper time parameter like this  would seem to be unavoidable it we want to `quantise' geodesic motion.

At a more practical level,  Section~\ref{sec:ex} computes the main elements of the formalism for some  important special cases: (a) the flat case but now in any coordinate system due to our geometric approach (here the differential calculus is strictly associative as usual), (b) the case of a compact Lie group such as $SU(2)=S^3$ computed in a left-invariant basis and (c) a Schwarzschild black hole background with its usual coordinates.  This provides a sanity check on the general results in earlier sections. Section~\ref{sec:app} considers applications of the formalism, focussing on the case where $M$ is spacetime and without an external potential. This section can  be understood directly from (\ref{kerrho1})-(\ref{kerrho2}) as derived in the preceding sections of the paper. We look at these operator geodesic equations and an Ehrenfest theorem for their expectation values. On the other hand, our `Heisenberg picture' flow on  $\CD(M)$ has a corresponding `Schr\"odinger picture' evolution on `wave functions' $\psi$ but now on spacetime with the Klein-Gordon operator in place of the spatial Laplacian, and with respect to this external geodesic time $s$, as a less conventional outcome of our approach. However, when the spacetime admits a  time-like Killing vector, we can restrict as for flat space in \cite{BegMa:geo} to (non-normalisable) modes of a fixed frequency $e^{-\imath\omega t}$ with respect to the preferred time direction. On such modes, the Klein-Gordon flow reduces to what we call `pseudo-quantum mechanics', which resembles ordinary quantum mechanics for wave functions defined on space  but still has evolution with respect to geodesic time $s$. Using this formalism around a Schwarzschild black hole, we look in Section~\ref{sec:dir} at an initial Gaussian bump wave function and see in detail how  this gets absorbed by the black hole through the emergence of modes created at the horizon that eventually replace it. At least in  examples of the type we looked at, the classical entropy of the probability density $\rho=|\psi|^2$ increases throughout this process.

Finally, Section~\ref{sec:gravatom} constructs exact stationary states for pseudo-quantum mechanics around a black hole, i.e. non-normalisable modes for the Klein-Gordon flow of the form
\[ \psi(s,t,x)=e^{-\imath {E_{KG} \over\hbar}s} \phi(t,x),\quad  \phi(t,x)=e^{-\imath\omega t}\psi_{E}(x)\]
for spatial eigenfunctions $\psi_E(x)$ which resemble those of a hydrogen atom of energy $E$ far from the event horizon. Here $\phi(t,x)$ is an exact solution of the Klein-Gordon equation of square-mass proportional to   $E_{KG}$. Even though the Klein-Gordon equation is 2nd order in $t$ rather than 1st order as for the usual Schr\"odinger equation, this is irrelevant for stationary modes provided we specify, say, negative frequencies by $\omega\ge 0$ as here. Considering the ordinary Klein-Gordon equation as an extension of actual quantum mechanics with respect to $t$ is not new, see for example \cite{Don} around a black hole, but we arrive at it differently and with  somewhat more detailed results.  The spectrum of the gravatom that we describe, in the sense of gravitationally bound quantum states\cite{gravatom}, is not quantised due to an open boundary at the horizon, but the radial wave functions are not unlike higher orbital  modes of a hydrogen atom, albeit with a fractal banding in probability density, i.e. crossing zero infinitely often approaching the horizon. 

Section~\ref{sec:pre} provides some preliminary elements of noncommutative geometry of the particular `quantum Riemannian' flavour \cite{BegMa} that we use (as opposed to Connes earlier approach \cite{Con} coming out of operator algebras). This grew out of the  {\em quantum spacetime hypothesis} that spacetime is better modelled as noncommutative due to quantum gravity effects \cite{DFR,Ma:pla,Hoo,MaRue},  as used in works such as \cite{BegMa:gra,Ma:sq, Ma:haw,LirMa}. Likewise, quantum geodesics \cite{Beg:geo,BegMa:geo,BegMa:cur, BegMa:geogra, LiuMa}  have been introduced as a way of formulating geodesics on such quantum spacetimes. For the intuitive picture here, the reader  should imagine a dust of particles each moving on geodesics and  then replace the flow of a density $\rho$ of such particles by the flow of a wave function $\psi$ such that $\rho=|\psi|^2$ as in quantum mechanics. At the density level, there are also similarities with optimal transport\cite{LotVel} and there could be applications to relativistic fluid dynamics as in \cite{Olt}, but when we work with complex wave functions $\psi$ the theory acquires a very different and more quantum-mechanics like character. This is not our topic in the present paper, however, where we rather apply the formalism to $\CD(M)$ and then transfer the flow of its elements to $L^2(M)$. The paper concludes in Section~\ref{sec:con} with some remarks about directions for further work. We note that a literature search since the preprint version  turned up \cite{Fan} where a probabilistic view with respect to an external time $s$ was also proposed, and \cite{Kon} where applying noncommutative geometry of some form to quantise geodesics was also proposed, albeit different from our approach.  

We work in units of $c=1$ and signature $-+++$ in the spacetime case. In what follows, we will  distinguish between the real coordinate vector fields $\del_\mu$ as locally defined elements of $\CD(M)$ and their image $p_\mu=\rho(\del_\mu)$ as the corresponding local momentum operators. In the classical limit $\lambda=0$ of this algebra, the $\del_\mu$ map to $p_\mu$ as the real classical locally-defined momentum of a single particle moving on a geodesic as explained above. This gives continuity with usual notations in physics, extended to $\CD(M)$  via noncommutative geometry. 

\section{Preliminaries}\label{sec:pre}

Here we recap some basic preliminaries from conventional Riemannian geometry in the notations we need, and elements of noncommutative geometry from \cite{BegMa,BegMa:geo}. 

\subsection{Background and notation}\label{sec:back}
In the general theory, we will write $\partial_a$ for a local-coordinate vector field on the manifold, whereas $\frac{\del}{\del x^a}$ will be a partial derivative as an operator when we later consider vector fields acting as $\lambda \frac{\del}{\del x^a}$ on wave functions (in the  Schr\"odinger representation $\rho$). This imaginary constant  $\lambda$ obeys $\lambda^*=-\lambda$ and in quantum mechanics has value $-\mathrm{i}\,\hbar$. We take it to be `small' in that we count orders of $\lambda$ and take lower orders to be more significant. 

By working to order $\lambda^2$ we mean  discarding $\lambda^3$ in geometric constructions on the manifold $M$. Vector fields here will typically be denoted $X,Y,Z$ and functions typically $f,h$ etc. and will be taken to have order zero. The real parameter $m$ has dimensions of mass, and we will similarly not count its order or make assumptions on its size. We take $g^{ab}$ to be a (possibly Lorentzian) Riemannian metric on the connected manifold $M$, and $\nabla$ to be its Levi-Civita connection with Christoffel symbols $\Gamma^a{}_{bc}$. Unless otherwise stated we assume that the vector fields $X,Y,Z$ and functions $f,h$ are real, though  $\CD(M)$ below will be taken as a complex algebra with a $*$-operation that picks out the real geometry as invariant under it. 

We will use a semicolon to denote covariant differentiation of tensors, e.g.
\[
H^a{}_{b;c} = H^a{}_{b,c} + H^d{}_{b} \, \Gamma^a{}_{dc} - H^a{}_{d} \, \Gamma^d{}_{bc} 
\]
where comma denotes partial differentiation. We repeat the semicolon for successive covariant differentiation, including previous derivative indices. For example the differential of $f_{,a}=f_{;a}$ is
\[
f_{,a;b} =f_{;a;b} = f_{,a,b} - f_{,c} \,\Gamma^c{}_{ab}\ .
\]
The curvature on 1-forms and vector fields is
\[
([\nabla_a,\nabla_b] \xi)_a=-R^d{}_{cab}\, \xi_d,\quad ([\nabla_a,\nabla_b] X)^d=R^d{}_{cab}\, X^c
\]
in the case of a coordinate basis where $[\del_a,\del_b]=0$. More generally, as the Levi-Civita connection is torsion free, we can write the Lie bracket of vector fields as 
\begin{align} \label{bt6v}
[Y,X]_{\rm Lie}=\nabla_Y X-\nabla_X Y.
\end{align}

We will also have recourse to the standard measure of integration
\[
\int f(x^1,\dots,x^n)\,\sqrt{|\det(g)|}\, \extd x^1\dots\extd x^n
\]
on a coordinate patch, where $g$ is the matrix $g_{ab}$ for the metric in the coordinate basis.  Finally, for our purposes, we consider a complex version of the algebra of differential operators.

\begin{definition}\label{DM} Let  $\CD(M)$ be the algebra generated by complex valued smooth functions $C^\infty(M)$ and complex smooth vector fields, with  commutation relations
\begin{align} \label{btcv}
[Y,X]=\lambda [Y,X]_{\rm Lie},\quad [X,f]=\lambda\,X(f)\ ,\quad [f,g]=0.
\end{align}
We also set
\begin{align} \label{btcy}
f.X=fX,
\end{align}
where $fX$ denotes the vector field given by multiplying a function and a vector field to get a vector field as usual. Moreover, $\CD(M)$ acts on $C^\infty(M)$ (and on $C^\infty_c(M))$ by $\rho(f)g=fg$ and $\rho(X)=\lambda X(g)$ for all vector fields $X$ and functions $f,g$. \end{definition}

Here $\lambda$ at this point could be any nonzero parameter. It simply scales our generators in such a way as to have a commutative algebra as $\lambda\to 0$, in line with our view of the algebra as a quantisation of a classical geometry. There is also a $\star$ operation when $\lambda$ is imaginary,  which we will discuss in Section~\ref{orf} and which is the case we will need.  Although we  take this as a definition of $\CD(M)$, we now outline its equivalence with the usual description of 
the algebra of finite-degree differential operators on a smooth connected manifold $M$. Here, differential operators of degree $\le n$ are linear maps $C^\infty(M)\to C^\infty(M)$ which in every coordinate chart on an open set $U$ of any atlas of the manifold can be written as a sum of terms of the form
\begin{equation}\label{diffop} 
v_{i_1,\dots,i_k}\, {\del\over\del x^{i_1}}\cdots {\del\over\del x^{i_k}}.
\end{equation}
Here, $x^{i_m}$ are local coordinates, $v_{i_1,\dots,i_k}\in C^\infty(U)$  and $0\le k\le n$. Where charts overlap, the different coordinate descriptions will be related via the chain rule since the operator is the same. We first observe that  such operators can be described in global terms as the action (in the usual way, without $\lambda$) of sums of terms of the form
\begin{equation}\label{polyvec} f_k X_1\cdots X_k\end{equation}
for smooth functions $f_k$, smooth vector fields $X_i$ and $0\le k\le n$. In one direction, this is obvious as any (in our case, complex) vector field can be expanded in local coordinates and all the coefficients moved to the left via the product  rule. Conversely, every differential operator $D$ in the usual sense can be obtained this way, see e.g. \cite[Thm 9.62]{Nes}. To see this, one can multiply $D$ by 1 as a partition of unity for a finite atlas of $M$ to convert the local form of type  (\ref{diffop}) into a global one of type (\ref{polyvec}). It means that the usual algebra of finite-degree differential operators can be identified with the image of $\rho$ in Definition~\ref{DM}. Next, we note that  both $\CD(M)$ and the usual algebra of differential operators are filtered algebras according to the degree and have `associated graded' algebras in which one works modulo smaller degree. For $\CD(M)$, this is ${\rm Sym}(M)$, the algebra of symmetric tensors (i.e., symmetric tensor products over $C^\infty(M)$ of the space ${\rm Vect}(M)$ of vector fields) because the commutators are of lower degree. But this is also the associated graded algebra of the usual algebra of differential operators via the symbol of the differential operator. The induced map at the associated graded level is therefore an isomorphism to its image and by homological algebra, so is $\rho$ itself (i.e., one can see this by looking at the top degree). Hence, $\CD(M)$ as we defined it is isomorphic to the usual definition of finite-degree differential operators.

This completes the background from classical geometry. For noncommutative geometry, we use an approach that works over an algebra $A$, in our case a $*$-algebra working over $\C$ (namely, we take $A=\CD(M)$). A `differential calculus' means to specify an $A$-$A$-bimodule $\Omega^1$ and a map $\extd: A\to \Omega^1$ that obeys the Leibniz rule and where every element of $\Omega^1$ is a finite sum of terms $a\extd b$ for $a,b\in A$. In principle this should be extended to an `exterior algebra' $(\Omega,\extd)$ of all differential forms, but there is always a `maximal prolongation' way to do this by applying $\extd$ to the degree 1 relations. In the $*$-algebra case we require that this extends to a $*$ operation on $\Omega$ as a graded-involution that commutes with $\extd$.  A left bimodule connection \cite{DVM,Mou,BegMa} on $\Omega^1$ (or similarly for some other bimodule) is a pair of maps
\[ \nabla:\Omega^1\to \Omega^1\tens_A\Omega^1,\quad \sigma:
\Omega^1\tens_A\Omega^1\to \Omega^1\tens_A\Omega^1\]
obeying the Leibniz rules
\[  \nabla(a\xi)=\extd a\tens \xi+ a\nabla(\xi),\quad \nabla(\xi a)=\sigma(\xi\tens \extd a)+ \nabla(\xi)a \] 
for all $\xi\in\Omega^1$ and with $a\in A$ acting on the right hand side on the nearest copy of $\Omega^1$. The map $\sigma$ if it exists is determined by $\nabla$, i.e. is not additional data but a restriction in $\nabla$. One can apply a right module map `right vector field' $\Omega^1\to A$ to the left factor to turn $\nabla$ into something more like a covariant derivative. One also has a right handed version of these conditions, a right bimodule connection. The goal of the paper from a mathematical perspective is to find as best we can such a natural differential calculus on $A=\CD(M)$. 

\subsection{The Schr\"odinger representation and quantum geodesics flows}
We consider the Hilbert space $\mathcal{H}=L^2(M)$ of square integrable functions on $M$, using the standard measure. 
The algebra $\CD(M)$  acts on $L^2(M)$ in a representation $\rho:\CD(M)\to {\rm Lin}(L^2(M))$ as possibly unbounded operators by
\[
\rho(f)(\psi)=f\,\psi,\quad \rho(X)(\psi)=\lambda \, X^a\,\frac{\partial \psi}{\partial x^a},
\]
for $\psi\in L^2(M)$, $f\in C^\infty(M)$ and a vector field $X$. We use the coordinate formula for the standard differentiation of a function in the direction of a vector field.
We use $\rho$ explicitly to avoid confusion with powers of $\lambda$. We extend this to time dependent wave function $\psi(s)\in L^2(M)$ for some external `time' parameter $s$ i.e.\
 $\psi\in E=L^2(M) \bar\tens C^\infty(\R)$, where the overline is to remind us that we do not mean here the algebraic tensor product. Although one could make a completed tensor product to make $E$ a Hilbert module, what we mean more precisely by this in the present context is smooth $L^2(M)$-valued functions on $\R$. The tensor notation is rather more convenient for the description of the algebraic side of the constructions, so we retain this as a notational device. We next  fix a 
 Hermitian operator  $\rh$ acting on a suitable domain of $L^2(M)$ as our Hamiltonian and presented as the image in the Schr\"odinger representation of an element $\ch\in \CD(M)$.
 
 We now recap how this data, familiar from quantum mechanics (but we will also apply it to $M$ spacetime) relates to quantum geodesics flows on an algebra $A$. We recall \cite{BegMa} that a 
{\em right  $A-B$ bimodule connection} means an $A-B$ bimodule $E$ (so one can multiply elements of $E$ by elements of $A$ from the left and of $B$ from the right) and linear maps
\[ \nabla_E: E\to E\tens_B\Omega^1_B,\quad \sigma_E:\Omega^1_A\tens_A E\to E\tens_{B}\Omega^1_B\]
such that the Leibniz rules
\[  \nabla_E(e b)=e\tens \extd b+ \nabla_E(e)b,\quad \nabla_E(a e)=\sigma_E(\extd a\tens e)+ a\nabla_E(e) \] 
hold for all $e\in E, a\in A, b\in B$. This is a `polarised' version of a right $A-A$ bimodule connection on $E$. In our case, $A=\CD(M)$ with a differential calculus $\Omega^1_A$ to be determined and  $B=C^\infty(\R)$ is the geodesic time parameter $s$ coordinate algebra with its classical differential calculus and  $E=L^2(M)\bar\tens C^\infty(\R)$  (or $C^\infty(\R, C^\infty_c(M))$ prior to completion of $C_c^\infty(M)$ to $L^2(M)$). As in \cite{BegMa:geo}, we make a right $A-B$ bimodule  connection
\begin{equation}\label{nablaE}
\nabla_E(\psi)=(\dot\psi + \lambda^{-1}\rh\,\psi)\tens \extd s
\end{equation}
at least for a suitable domain (such as  $\psi\in C^\infty(\R, C^\infty_c(M)))$, where dot denotes partial derivative with respect to $s$. The quantum geodesic flow of $\psi\in E$ is given by $\nabla_E\psi=0$, i.e. a version of Schr\"odinger's equation for the observer of the quantum geodesic. We also have 
\begin{align} \label{vib}
\sigma_E(\extd a\tens\psi)=\nabla_{E}(\rho(a)\,\psi)-\rho(a)\, \nabla_E(\psi)=\rho(\mathfrak{X}(\extd a))\psi\tens \extd s,
\end{align}
where $\mathfrak{X}:\Omega^1_A\to A$ is the geodesic velocity vector also to be determined and again for example $\psi\in C^\infty(\R, C_c(M)))$.  As in \cite{BegMa:geo}, the composite $\rho\circ\mathfrak{X}$ is determined by (\ref{nablaE}) as 
\begin{align} \label{ghu}
\rho(\extd a):=\rho(\mathfrak{X}(\extd a))=\lambda^{-1}[\rh,\rho(a)]
\end{align}
and amounts to an extension of the Schr\"odinger representation of $\extd a$ on $L^2(M)$, for $a\in \CD(M)$.
 We will focus on  Hamiltonian $ \rh$ defined by the Laplacian and an optional external real potential $V$, 
\[
\rh\, \psi={\lambda^2\over 2m}\Delta\psi  + V\,\psi,\quad \Delta\psi=g^{ab}\, \psi_{,a;b} =g^{ab}\, \psi_{,a,b}  - \Gamma^c  \psi_{,c}
\]
or equivalently by the element 
\[ \ch={1\over 2m} (g^{ab}\del_a\del_b-\lambda\Gamma^c\del_c)+V\in \CD(M),\]
where $\Gamma^c:=g^{ab}\Gamma^c{}_{ab}$.

All of this depends on defining the differential calculus on $\CD(M)$, at least to degree 1, for the notion of a connection to make sense. After that the main part of the details for a quantum geodesic in the above case amounts to extending the Schr\"odinger representation as in (\ref{rho}). This is our main focus in the paper, with a little more about the underlying noncommutative geometry in Section~\ref{sec:geo}. 

\subsection{The star operation} \label{orf}
In particular,  we use the Schr\"odinger representation to motivate the definition of a  $*$-operation on $\CD(M)$ as follows,  for $\lambda$ imaginary. For a function $f$ on the manifold $M$, we let $f^*$ be simply the complex conjugate of $f$. For a real vector field $X$ we set
 \[ X^*:=X+\lambda\, \mathrm{div}(X),\]
 where we use the  divergence defined by the connection, $\mathrm{div}(X)=X^a{}_{;a}$. This is needed for $*$ to correspond to the adjoint operator  in the representation in the sense explained in the introduction (i.e. for inner products with elements from $C^\infty_c(M)\subset L^2(M)$).  The same principle applies to products of vector fields and makes $\CD(M)$ into a $*$-algebra with $*$ corresponding to adjoints in the sense discussed. Moreover, the operators in  $\CD(M)$ leave the dense subset $C_c^\infty(M)$ of $L^2(M)$ invariant, as do their images under $*$. (This means that we have an example of an $O^*$ algebra in the sense of \cite{Sch}.) For example in degree 2:

\begin{lemma}
Let the operator $T$ be defined by $T(\psi)=\lambda^2\,M^{ij}\, \psi_{,i;j}$ where $M^{ij}$ is a matrix of real functions. Then
\[
T^*(\psi)=T(\psi) + \lambda\, M^{ij}{}_{;i}\,\psi_{,j} + \lambda^2\, M^{ij}{}_{;j}\,\psi_{,i} +M^{ij}{}_{;i;j}\psi  \ .
\]
\end{lemma}
\begin{proof} We prove this for $M^{ij}=X^i\, Y^j$ and then use linear combinations for general $M^{ij}$. First, for vector fields 
$X=X^i\,\del_i$ and $Y=Y^j\,\del_j$
\[
T(\psi)=(Y\,X-\lambda\nabla_Y X) \psi. 
\]
We then use  $(Y\,X-\lambda\nabla_Y X)^*=X^*Y^*+\lambda(\nabla_Y X)^*$. Moreover, $(\phi,\rho(XY)\psi)=(\phi,\rho(X)\rho(Y)\psi)=(\rho(X^*)\phi,\rho(Y)\psi)=(\rho(Y^*)\rho(X^*)\phi,\psi)=(\rho((XY)^*))\phi,\psi)$ holds automatically for all $\phi,\psi\in C^\infty_c(M)$. \end{proof}

Also note that in the case of a single vector field $X$, we can subtract half the divergence correction. Then $X+{\lambda\over 2}{\rm div}(X)$ is invariant under $*$. In the case of a local coordinate vector field $\del_i$, the self-adjoint version in the above sense is 
\[ \del_i + {\lambda\over 2} \Gamma^j{}_{ij}.\]
But in the case where $M$ is space, $\Gamma^j{}_{ij}={1\over 2}{\del\over\del x^i}\ln \sqrt{\det(g)}$ so that the image of the self-adjoint version under $\rho$ is $-\imath\hbar({\del\over\del x^i}+ {1\over 4}{\del\over\del x^i}\ln\sqrt{\det(g)})$, which agrees with the proposed momenta (5.31) in \cite{Witt}, except that in that context spacetime is globally hyperbolic and the spatial metric on each time-slice can vary in the time $t$.

\section{Differential calculus on $\CD(M)$} \label{sec:calc}

In this section, we construct a natural differential calculus on $\CD(M)$ to the extent possible such as to obtain Schr\"odinger's equation as a quantum geodesic flow, i.e. the method used in \cite{BegMa:geo} in the flat spacetime case. The idea there, and here, is to work backwards to arrive at what we propose as set of relations for $\Omega^1_{\CD(M)}$ up to and including order $\lambda^2$ and which define a first order calculus to order $\lambda$. Higher differential forms will not be needed but their relations are in principle implied by application of $\extd$ to give an exterior algebra again to order $\lambda$. One can formally declare that $\lambda^2=0$ and then consider this as an actual calculus c.f. \cite{BegMa:poi}, although we do not do so here. We make no claim as to the axioms obeyed at order $\lambda^2$  due to failure of Jacobi identities in Section~\ref{sec:jac} at that order.

\subsection{Centrally extended one forms on $M$} \label{buio}

We start with the differentials for $f\in C^\infty(M)\subset \CD(M)$.  For the chosen form of $\ch$, we calculate
\[
m\, [\rh,\rho(f)] ={\lambda^2\over 2} \big(g^{ij}\, f_{,i,j}  - g^{ij}\, f_{,k} \,\Gamma^k{}_{ij} \big)+  \lambda^2\, g^{ij}\, f_{,j} \,\tfrac{\partial}{\partial x^i}\ ,\quad 
\]
and hence from (\ref{ghu}), we have
\begin{align*}
m\, \rho(\extd f)= \lambda\, \big({1\over 2}
g^{ij}\, (f_{,i})_{;j}  + g^{ij}\, f_{,j} \,\tfrac{\partial}{\partial x^i}
\big)\ ,\quad
m\, [\rho(\extd h),\rho(f)] = \,\lambda\,g^{ij}\, h_{,j} \,f_{,i}
\end{align*}
for all $f,h\in C^\infty(M)$. As $[\extd h,f]$ should be a 1-form on $\CD(M)$, we adjoin an extra 1-form $\theta'\in \Omega^1_{\CD(M)}$ which commutes with elements of $\CD(M)$ and obeys
\[
 \sigma_E(\theta'\tens\psi) =\psi\tens \extd s.
\]
Then we set
\begin{align*}
m\, [\extd h,f] = \,\lambda\,g^{ij}\, h_{,j} \,f_{,i}\,\theta',\quad \rho(\theta')=1
\end{align*}
which then has the right image under $\rho$.

We still have to be careful about defining a product between 1-forms and functions as part of $\Omega_{\CD(M)}$, rather than just a commutation relation, which we do symmetrically.  
Thus, for such a product $\bullet$ on the calculus which is consistent with the representation, we look at more general 1-forms than $\extd f$
while being careful about this lack of commutation. For $\eta_p\in C^\infty(M)$ we set $\widehat\eta=\frac12(\eta_p\bullet\extd x^p+\extd x^p\bullet \eta_p)$. Then
\begin{align} \label{bas}
2m\,\rho(\widehat\eta) =  2m\eta_p \rho(\extd x^p) + m [ \rho(\extd x^p), \eta_p] =\lambda\, \big(
g^{ij}\, \eta_{i;j}  + 2\,g^{ij}\, \eta_{j} \,\tfrac{\partial}{\partial x^i}\big)\ .
\end{align}
We can now define the centrally extended 1-forms
$\widehat \Omega^1(M)$ to consist of $\widehat \eta+f\,\theta'$, where $\eta\in\Omega^1(M)$ and $f\in C^\infty(M)$. The product is given by 
\begin{align} \label{repute}
f\bullet \widehat \eta = \widehat{f\,\eta} - {\lambda\over 2m} g^{ij}\, f_{,i}\,\eta_j \,\theta'\ ,\quad 
\widehat\eta\bullet  f = \widehat{f\,\eta} + {\lambda\over 2m} g^{ij}\, f_{,i}\,\eta_j\,\theta'\ ,
\end{align}
where $f\,\eta$ is the usual classical product of a function $f$ and 1-form $\eta$. This gives a commutator which is consistent with the formula above,
\begin{align} \label{firw}
m\, [\widehat \eta,f] = \lambda\,g^{ij}\, \eta_j \,f_{,i}\,\theta'.
\end{align}
This also means that $\theta'$ is in the algebra generated by functions and their differentials provided $\lambda\ne 0$, as for the flat case in \cite{BegMa:geo}. To see this, note that since $\Omega^1(M)$ is finitely generated and projective as an $A=C^\infty(M)$ module, there is is a `dual basis' ${\rm coev}=\sum_\alpha X_\alpha\tens_A \xi_\alpha$ with vector fields $X_\alpha$ and 1-forms $\xi_\alpha$ obeying $X_\alpha(\xi_\beta)=\delta_{\alpha,\beta}$. Writing $\xi_\alpha=\sum_{I}c_{\alpha I}\extd f_{\alpha I}$ for some functions $c_{\alpha I}$ and $f_{\alpha I}$ and letting $\eta_{\alpha I}$  be the 1-form corresponding to the vector field $X_\alpha c_{\alpha I}$ via the metric isomorphism $\#$, we have ${\rm coev}=\sum_{\alpha,I} (\eta_{\alpha I})^\#\tens_A \extd  f_{\alpha I}$. Applying  (\ref{firw}) to pairs $(\eta_{\alpha I}, f_{\alpha I})$ and summing over $(\alpha,I)$ gives the evaluation of $\rm coev$, which as constant, in front of $\theta'$. This expresses the latter in the required form. Moreover, having specified (\ref{firw}), we will omit writing the bullet product on the left in  (\ref{repute}) and consider the right bullet product there as defined by this commutator. The differential on $\CD(M)$ is given by setting $\extd f=\widehat{\extd f}$. This has a standard central extension form as in \cite[Prop.~1.22]{BegMa} except that we have chosen to do the product symmetrically. 

We observe that the 1-form
\begin{align}\label{keromega1}
m\,\widehat{\xi} -  g^{ij}\,\xi_j\, \theta'\,\partial_i -{ \lambda\over 2} g^{ij}\,\xi_{i;j}\,\theta'
\end{align}
is in the kernel of the Schr\"odinger representation $\rho$ for all $\xi\in\Omega^1_M$. We regard $\del_i$ (locally) as  an element of  $\CD(M)$ or 0-form. In particular, the elements
\begin{align*}
m\,\widehat{\extd x^i} -  g^{ij}\, \theta'\,\partial_j + {\lambda\over 2}  g^{pq}\,   \Gamma^i{}_{pq} \,\theta'
\end{align*}
are in the kernel to order $\lambda^2$. This means that the representation alone cannot uniquely determine the relations of $\Omega_{\CD(M)}$ even in degree 1; we need additional information.

\subsection{Commutator of differentials of functions and vector fields} \label{vhc}
We next  find the commutator $[\widehat\xi  ,X] $ for $\xi\in\Omega^1(M)$ and a vector field $X$. First we apply the representation and calculate
\begin{align} \label{bhg}
m\, [\rho(\widehat\xi & ),\rho(X)] =  -{ \lambda^2\over 2} X^a\,(g^{ij}\, \xi_{i;j} )_{,a}   + \lambda^2\,g^{ip}\, \big(g^{bj}\, \xi_{j}\, X^q{}_{;b}\, g_{pq}  -   \xi_{p;a} \, X^a\big)\,\partial_i  \cr
&=  - {\lambda^2\over 2} X^a\,g^{ij}\, (\xi_{i;j} )_{;a}   + \lambda^2\,g^{ip}\, \big(g^{bj}\, \xi_{j}\, X^q{}_{;b}\, g_{pq}  
-   \xi_{p;a} \, X^a\big)\,\partial_i     \ .
\end{align}
In addition, the relation (\ref{btcy}) gives
\begin{align} \label{hjs}
m\, [ f X,\widehat\xi]-m\, f\, [X,\widehat\xi] = m\, [f,\widehat\xi]\, X =-\lambda\,g^{ij}\xi_i\,f_{,j}\theta'\, X.
\end{align}
Using (\ref{bas}), we have
\begin{align*} 
m\,\rho(\widehat{ \nabla_X \xi  })  ={\lambda\over 2} 
g^{ij}\, (X^a\,\xi_{i;a})_{;j}  + \lambda g^{ip}\, X^a\,\xi_{p;a} \,\tfrac{\partial}{\partial x^i}\ ,
\end{align*}
and from this we propose the following to satisfy both (\ref{bhg}) and (\ref{hjs}):
\begin{align} \label{orp45}
m\, [X,\widehat\xi &] = m\, \lambda\,(\widehat{ \nabla_X \xi  })  - \lambda\,\theta'\, (g^{ij}\,\xi_i\, \nabla_j X) 
-{\lambda^2\over 2} \theta'\big(X^a\, \xi_p\, g^{pq}\, R_{qa} + g^{ij}\, X^a{}_{;j}\, \xi_{i;a}
\big)
\end{align}

\begin{proposition} \label{geid7}
The commutation relation in (\ref{orp45}) preserves the star operation.
\end{proposition}
\begin{proof} 
For real $X$ and $\xi$ we apply $*$ to this to find, on the assumption that the commutators are respected by the star operation,
\begin{align} \label{orp2}
m\, [\widehat\xi ,X+\lambda\,\mathrm{div}(X)] &=  -m\, \lambda\,(\widehat{ \nabla_X \xi  })  + \lambda\,\theta'\, (\nabla_j X)\, g^{ij}\,\xi_i
 + \lambda^2\,\theta'\, g^{ij}\,\xi_i\, \mathrm{div}(\nabla_j X)   \cr
&\quad -{\lambda^2\over 2} \theta'\big(X^a\, \xi_p\, g^{pq}\, R_{qa} + g^{ij}\, X^a{}_{;j}\, \xi_{i;a}
\big).
\end{align}
So we require to show
\begin{align} \label{orp3}
m\, [\widehat\xi ,\lambda\,\mathrm{div}(X)] &=  \lambda\,\theta'\, [\nabla_j X, g^{ij}\,\xi_i]  + \lambda^2\,\theta'\, g^{ij}\,\xi_i\, \mathrm{div}(\nabla_j X)   \cr
&\quad
- \lambda^2\, \theta'\big(X^a\, \xi_p\, g^{pq}\, R_{qa} +  g^{ij}\, X^a{}_{;j}\, \xi_{i;a}
\big)
\end{align}
and this is equivalent to
\begin{align} \label{orp5}
g^{ij}\,\xi_i\, \mathrm{div}(X)_{,j}&=  -   X^a{}_{;j}\, g^{ik}   \, \Gamma^j{}_{ak}    \,\xi_i + g^{ik}\,\xi_i\, \mathrm{div}(\nabla_k X)   -X^a\, \xi_p\, g^{pq}\, R_{qa} \cr
&= g^{ik}\,\xi_i\, X^a{}_{;k;a}   -X^a\, \xi_p\, g^{pq}\, g^{ik}\, R_{iqka}  \cr
&= g^{ik}\,\xi_i\, \big(X^a{}_{;k;a}   -X^b\,  g^{nm}\, R_{nkmb}  \big) \cr
&= g^{ik}\,\xi_i\, \big(X^a{}_{;k;a}   + X^b\,  g^{na}\, R_{nbka}  \big), 
\end{align}
which holds as required.\end{proof}

\subsection{Commutator of functions and differentials of vector fields}
From our previous calculations we have an immediate result to order $\lambda^2$

\begin{proposition} \label{geid}
We have
\begin{align*}
m\, [\extd X,f &] = m\, \lambda\,(\reallywidehat{X^a{}_{;i } \, f_{,a} \,\extd x^i})  + \lambda\,\theta'\, (g^{ij}\, f_{,i}\, \nabla_j X) 
+   {\lambda^2\over 2} \theta'\big(X^a\, f_{,p}\, g^{pq}\, R_{qa} + g^{ij}\, X^a{}_{;j}\, f_{,i;a}   
\big)
\end{align*}
and this preserves the star operation.
\end{proposition}
\begin{proof} We use Section~\ref{vhc} and differentiating the relation $[X,f]=\lambda\, X^i\,  f_{,i}$.
To check the star property we need to show that, for real $f,X$
\begin{align} \label{jdi}
m([\extd X,f] + [\extd X,f]^*) =\lambda\,m\,[f,\extd \, \mathrm{div}(X)]. 
\end{align}
The LHS of (\ref{jdi}) is
\begin{align*}
&\lambda\,\theta'\,[g^{ij}\,f_{,i} ,\nabla_j X] - \,\lambda^2\,\theta'\, g^{ij}\,f_{,i} \, \mathrm{div}(\nabla_j X) 
+  \lambda^2\, \theta'\big(X^a\, f_{,p}\, g^{pq}\, R_{qa} + g^{ij}\, X^a{}_{;j}\, f_{,i;a}   \big) \cr
&=-\lambda^2\,\theta'\, X^a{}_{;j}\, (g^{ij}\,f_{,i})_{,a}
- \lambda^2\,\theta'\, g^{ij}\,f_{,i} \, \mathrm{div}(\nabla_j X) 
+  \lambda^2\, \theta'\big(X^a\, f_{,p}\, g^{pq}\, R_{qa} + g^{ij}\, X^a{}_{;j}\, f_{,i;a}   \big) \cr
&=-\lambda^2\,\theta'\, X^a{}_{;j}\, (-g^{ik}\,f_{,i}\,\Gamma^j{}_{ka}  +    g^{ij}\,f_{,i;a}   )
- \,\lambda^2\,\theta'\, g^{ij}\,f_{,i} \, \mathrm{div}(\nabla_j X) 
+   \lambda^2\, \theta'\big(X^a\, f_{,p}\, g^{pq}\, R_{qa} + g^{ij}\, X^a{}_{;j}\, f_{,i;a}   \big) \cr
&=
-  \lambda^2\,\theta'\, g^{ij}\,f_{,i} \,  X^a{}_{;j;a}
+  \lambda^2\, \theta'   X^a\, f_{,i}\, g^{ij}\, R_{ja}   \cr
&=
- \lambda^2\,\theta'\, g^{ij}\,f_{,i} \,  X^a{}_{;a;j} ,
\end{align*}
which is equal to the RHS as required. \end{proof}

\subsection{The form of commutator of vector fields and their differentials}
\begin{proposition} \label{gup} The commutation relations (\ref{btcv}) and (\ref{bt6v}) imply  
commutation relations for $\extd X$ of the form
\[
[    Y  , \extd X  ] =\lambda\, \extd(\nabla_Y X)+\lambda\, P(X,Y)
\]
where $P(X,Y)=P(Y,X)$. Assuming associativity to order $\lambda$, the 
relation $f.X=(fX)$ implies to order $\lambda$
\begin{align*}
  \lambda\, P(f X,Y) - \lambda\,  f P(X,Y)
  &=  - \lambda\,(\reallywidehat{ Y^a{}_{;i } \, f_{,a} \,\extd x^i  })  \, X -  \lambda m^{-1}\,\theta'\, (g^{ij}\, f_{,i}\, \nabla_j Y) \, X
   -\lambda\,  \extd f\, \nabla_X Y 
\end{align*}
\end{proposition}
\begin{proof} We have $[Y,X]=\lambda(\nabla_Y X-\nabla_X Y)$, and applying the derivation $\extd$ gives
\[
[Y,\extd X] - \lambda\, \extd(\nabla_Y X) = [X,\extd Y] - \lambda\, \extd(\nabla_X Y) 
\]
and we label this $\lambda\, P(X,Y)$. 
Next, $\extd( f X)=\extd f.X+f \extd X$ and then, assuming associativity to order $\lambda$ in what follows
\begin{align*}
[Y,\extd(fX)] &= [Y,\extd\,f]\,X+  \extd f\, [Y,X]+[Y,f]\,\extd X+f [Y,\extd X]+O(\lambda^2)
\end{align*}
which gives
\begin{align*}
\lambda\, \extd(\nabla_Y (fX))+\lambda\, P(fX,Y)&= [Y,\extd f]\,X+  \extd f\, [Y,X]+[Y,f]\,\extd X+\lambda\,f \extd(\nabla_Y X)+\lambda\, f P(X,Y)
\end{align*}
Now
\begin{align*}
\lambda\, \extd(\nabla_Y (fX))  =\lambda\, \extd\big(Y(\extd f)\, X+ f,\nabla_Y (X) \big)  
\end{align*}
so we get
\begin{align*}
  \lambda\, P(f X,Y) - \lambda\, fP(X,Y)&=\big( [Y,\extd f]\, - \lambda\, \extd(Y(\extd f))\big)\, X -\lambda\,  \extd f\, \nabla_X Y \cr
  &= -[\extd Y, f]\,  X -\lambda\,  \extd f\, \nabla_X Y  \ . 
\end{align*}
giving the answer. \end{proof}

\begin{proposition} \label{real5}
The reality condition $[    Y  , \extd X  ]^*=-[    Y^*  , \extd X^*  ]$ for real vector fields $X,Y$, assuming that $\extd(X^*)=(\extd X)^*$
and using $X^*=X+\lambda\,\mathrm{div}(X)$ for real $X$, is that for real $X,Y$ (we name the expression $N(X,Y)$ to use it later)
\begin{align*}
N(X,Y)& = P(X,Y)   -   P(X,Y)^* =  \lambda\,m^{-1}\,\theta'\, g^{ij}\, \mathrm{div}(Y)_{,i}\, \nabla_j X 
 + \lambda\,m^{-1}\,\theta'\, g^{ij}\,\mathrm{div}(X)_{,i}\, \nabla_j Y \cr
 &\quad + \lambda\big( \extd\big( X^q\, Y^p\, R_{pq} + Y^p{}_{;q} X^q{}_{;p}\big) +\extd x^i\big(
 Y^p{}_{;i} \mathrm{div}(X)_{;p} +  X^p{}_{;i} \mathrm{div}(Y)_{;p} 
 \big)\big)
\end{align*}
\end{proposition}
\begin{proof} We have, to order $\lambda^2$
\begin{align*}
\big(\lambda\, \extd(\nabla_Y X)+\lambda\, P(X,Y)\big)^* = -[    Y+\lambda\,\mathrm{div}(Y)  , \extd X +\lambda\,\extd\,\mathrm{div}(X)  ]
\end{align*}
which gives
\begin{align*}
-\lambda\, (\extd(\nabla_Y X))^*-\lambda\, P(X,Y)^* =-\lambda\, \extd(\nabla_Y X)-\lambda\, P(X,Y)
 -[    \lambda\,\mathrm{div}(Y)  , \extd X  ]
  -[    Y  ,\lambda\,\extd\,\mathrm{div}(X)  ]
\end{align*}
which gives, to order $\lambda$
\begin{align*}
P(X,Y)   -   P(X,Y)^* = \lambda\, \extd\,\mathrm{div}(\nabla_Y X)  
 -[  \mathrm{div}(Y)  , \extd X  ]
  -[    Y  ,\extd\,\mathrm{div}(X)  ]
\end{align*}
and using Proposition~\ref{geid} and (\ref{orp45}) we have, to order $\lambda$
\begin{align*}
P(X,Y)   -   P(X,Y)^* &= \lambda\, \extd\,\mathrm{div}(\nabla_Y X)  
+  \lambda\, X^a{}_{;i } \, \mathrm{div}(Y)_{,a} \,\extd x^i + \lambda\,m^{-1}\,\theta'\, g^{ij}\, \mathrm{div}(Y)_{,i}\, \nabla_j X \cr
& \quad -  \lambda\,\nabla_Y( \extd\,\mathrm{div}(X) )  + \lambda\,m^{-1}\,\theta'\, g^{ij}\,\mathrm{div}(X)_{,i}\, \nabla_j Y
\end{align*}
and then use standard differential geometry calculations.
\end{proof}

\subsection{Schr\"odinger representation of the differential of a vector field}

\begin{proposition} \label{hdkk}
The representation of $\extd X$ for a vector field $X$ is
\[
m \rho(\extd X)(\psi)=
  \lambda^2\,g^{ij}\, X^a{}_{;i}\, (\psi_{,a})_{;j}  
+{\lambda^2\over 2}((\Delta X)^a + X^a\, g^{ij}\, R_{ja})\psi_{,i}
  -  m\,V_{,a}\,X^a\, \psi \ ,
\]
where $R_{qr}$ is the Ricci tensor and $\Delta$ is the Laplace-Beltrami operator on vector fields. This corresponds to
\[
m\,\extd X -  \theta'\left( \,g^{ij}\, X^a{}_{;i}\, (\del_a\del_j- \lambda\Gamma^k{}_{aj} \, \del_k)
+{\lambda\over 2}(\Delta X  +  X^a\, g^{ij}\, R_{ja}\,\del_i)
  -  m X(V)\right)
\]
being in the kernel of the Schr\"odinger representation to order $\lambda^2$.
\end{proposition}
\begin{proof} From (\ref{ghu}) 
\begin{align*}
2m\, \rho&(\extd X)(\psi)=2m\, \lambda^{-1}[ \rh,\rho(X)]\psi \cr
&= \lambda^2\,g^{ij}\, ((X^a\, \psi_{,a})_{,i})_{;j}    +  2m\,V\,X^a\, \psi_{,a}
- X^a\,\partial_a(\lambda^2\,g^{ij}\, (\psi_{,i})_{;j}    +  2m\,V\,\psi) \cr
&= \lambda^2\,g^{ij}\, X^a\, (((\psi_{,a})_{;i})_{;j}  - ((\psi_{,i})_{;j})_{;a}) + 2\lambda^2\,g^{ij}\, X^a{}_{;i}\, (\psi_{,a})_{;j}  
+\lambda^2\,g^{ij}\, (X^a{}_{;i})_{;j} \, \psi_{,a}  \cr
 &\quad  -  2m\,V_{,a}\,X^a\, \psi \cr
 &= \lambda^2\,g^{ij}\, X^a\, (((\psi_{,i})_{;a})_{;j}  - ((\psi_{,i})_{;j})_{;a}) + 2\lambda^2\,g^{ij}\, X^a{}_{;i}\, (\psi_{,a})_{;j}  
+\lambda^2\,g^{ij}\, (X^a{}_{;i})_{;j} \, \psi_{,a}  \cr
 &\quad  -  2m\,V_{,a}\,X^a\, \psi \cr
  &=  2\lambda^2\,g^{ij}\, X^a{}_{;i}\, (\psi_{,a})_{;j}  
+\lambda^2\,g^{ij}\, \big( (X^a{}_{,i})_{;j}  -  X^r\, R^a{}_{ijr}\big)\,\psi_{,a}   -  2m\,V_{,a}\,X^a\, \psi \ ,
\end{align*}
giving the answer.      \end{proof}

In particular, the elements 
\[ \extd \del_i- {\theta'\over m} \left(\Gamma^j{}_{i d}g^{d c}(\del_j \del_c-\lambda\Gamma^e{}_{jc}\del_e)+{\lambda\over 2}\Delta(\del_i)+ {\lambda\over 2} R_{ji}g^{jc}\del_c-m V_{,i}\right)\]
are in the kernel to order $\lambda^2$.

\begin{proposition} \label{bukk}
 \begin{align*} 
m\, \lambda^{-2}\rho&(P(X,Y))\psi  + m\, \lambda^{-2}Y^b\, X^a\, V_{,a;b}\,\psi=
\cr
&=  {1\over 2}\big(  Y^b\,X^c\, (   - g^{aq}\, g^{ij}(R_{qijc;b} + R_{qcib;j})   ) -   g^{cq}\, R_{qb}\, (X^b\, Y^a{}_{;c}  + Y^b\, X^a{}_{;c}) \cr
  &\quad -2 g^{ij} Y^b{}_{;i}\, X^a{}_{;b;j}  -2\,g^{ij}\, X^c{}_{;i}\, Y^a{}_{;c;j}   -  (\nabla_{\Delta(X)}Y)^a - (\nabla_{\Delta(Y)}X)^a    \big) \,\psi_{,a}  \cr
  &\quad -       g^{ij}\, X^a{}_{;i}\, Y^b{}_{;a}\,\psi_{,b;j} -   g^{ij}\, Y^b{}_{;i}\, X^a{}_{;b}\, \psi_{,a;j} 
-       g^{ij}\, X^a{}_{;i}\, Y^b{}_{;j}\,\psi_{,b;a}     +  g^{ij}\, Y^b \, X^c R_{ecbi} \,g^{ae}  \psi_{,a;j}  \ .
\end{align*}
\end{proposition}
\begin{proof} By definition of $P(X,Y)$,
 \begin{align*}
\lambda\, \rho&(P(X,Y))\psi =\rho([    Y  , \extd X  ])\psi -\lambda\, \rho(\extd(\nabla_Y X))\psi   \cr
 &=\lambda\, Y^a \,  \partial_a\,\rho( \extd X  )\psi - \lambda\, \rho(\extd X  ) (Y^a\,\psi_{,a})-\lambda\, \rho(\extd(\nabla_Y X))\psi
\end{align*}
and using Proposition~\ref{hdkk} we have
 \begin{align*} 
&2m\, \rho(P(X,Y))\psi =
2m\, Y^a \,  \partial_a\,\rho( \extd X  )\psi - 2m\, \rho(\extd X  ) (Y^a\,\psi_{,a})-2m\, \rho(\extd(\nabla_Y X))\psi \cr
&= Y^b \,  \partial_b\big( 2\lambda^2\,g^{ij}\, X^a{}_{;i}\, (\psi_{,a})_{;j}  
+\lambda^2\,\Delta(X)^a \,\psi_{,a} + \lambda^2\, X^r\, g^{aq}\, R_{qr}\, \psi_{,a}
  -  2m\,V_{,a}\,X^a\, \psi\big) \cr
  &\quad -\big(2\lambda^2\,g^{ij}\, X^a{}_{;i}\, ((Y^b\,\psi_{,b})_{,a})_{;j}  
+\lambda^2\,\Delta(X)^a \,(Y^b\,\psi_{,b})_{,a} + \lambda^2\, X^r\, g^{aq}\, R_{qr}\, (Y^b\,\psi_{,b})_{,a}
  -  2m\,V_{,a}\,X^a\, Y^b\,\psi_{,b}\big)   \cr
  &\quad - \big( 2\lambda^2\,g^{ij}\, (\nabla_Y X)^a{}_{;i}\, (\psi_{,a})_{;j}  
+\lambda^2\,\Delta(\nabla_Y X)^a \,\psi_{,a} + \lambda^2\, (\nabla_Y X)^r\, g^{aq}\, R_{qr}\, \psi_{,a}
  -  2m\,V_{,a}\,(\nabla_Y X)^a\, \psi \big),
\end{align*}
which we simplify as
 \begin{align*} 
2m\, \lambda^{-2}\big(\rho&(P(X,Y))  +Y^b\, X^a\, V_{,a;b}\big)\psi=
\cr
&= Y^b \,  \partial_b\big( 2\,g^{ij}\, X^a{}_{;i}\, (\psi_{,a})_{;j}  
+\,\Delta(X)^a \,\psi_{,a} + \, X^r\, g^{aq}\, R_{qr}\, \psi_{,a}  \big) \cr
  &\quad -\big(2\,g^{ij}\, X^a{}_{;i}\, ((Y^b\,\psi_{,b})_{,a})_{;j}  
+\,\Delta(X)^a \,(Y^b\,\psi_{,b})_{,a} + \, X^r\, g^{aq}\, R_{qr}\, (Y^b\,\psi_{,b})_{,a}   \big)   \cr
  &\quad - \big( 2\,g^{ij}\, (\nabla_Y X)^a{}_{;i}\, (\psi_{,a})_{;j}  
+\,\Delta(\nabla_Y X)^a \,\psi_{,a} + \, (\nabla_Y X)^r\, g^{aq}\, R_{qr}\, \psi_{,a}   \big)   \cr
&= Y^b \,  \partial_b\big( 2\,g^{ij}\, X^a{}_{;i}\, (\psi_{,a})_{;j}  
+\,\Delta(X)^a \,\psi_{,a}\big)
 +  Y^b\,X^r\, g^{aq}\, R_{qr;b}\, \psi_{,a}  +  Y^b\,X^r\, g^{aq}\, R_{qr}\, \psi_{,a;b}   \cr
  &\quad -\big(2\,g^{ij}\, X^a{}_{;i}\, ((Y^b\,\psi_{,b})_{,a})_{;j}  
+\,\Delta(X)^a \,(Y^b\,\psi_{,b})_{,a} + \, X^r\, g^{aq}\, R_{qr}\, (Y^b\,\psi_{,b})_{,a}   \big)   \cr
  &\quad - \big( 2\,g^{ij}\, (\nabla_Y X)^a{}_{;i}\, (\psi_{,a})_{;j}  
+\,\Delta(\nabla_Y X)^a \,\psi_{,a}  \big) \cr
&= Y^b \,  \partial_b\big( 2\,g^{ij}\, X^a{}_{;i}\, (\psi_{,a})_{;j}  
+\,\Delta(X)^a \,\psi_{,a}\big)
 +  Y^b\,X^r\, g^{aq}\, R_{qr;b}\, \psi_{,a}    \cr
  &\quad -\big(2\,g^{ij}\, X^a{}_{;i}\, ((Y^b\,\psi_{,b})_{,a})_{;j}  
+\,\Delta(X)^a \,Y^b\,\psi_{,b;a}    +\,\Delta(X)^a \,Y^b{}_{;a}\,\psi_{,b}    + \, X^r\, g^{aq}\, R_{qr}\, Y^b{}_{;a}\,\psi_{,b}   \big)   \cr
  &\quad - \big( 2\,g^{ij}\, (\nabla_Y X)^a{}_{;i}\, (\psi_{,a})_{;j}  
+\,\Delta(\nabla_Y X)^a \,\psi_{,a}  \big)   \cr
&= Y^b \,  \big( 2\,g^{ij}\, X^a{}_{;i}\, \psi_{,a;j;b}  + 2\,g^{ij}\, X^a{}_{;i;b}\, (\psi_{,a})_{;j}  
+\,\Delta(X)^a{}_{;b} \,\psi_{,a}\big)
 +  Y^b\,X^r\, g^{aq}\, R_{qr;b}\, \psi_{,a}    \cr
  &\quad -\big(2\,g^{ij}\, X^a{}_{;i}\, ((Y^b\,\psi_{,b})_{,a})_{;j}  
   +\,\Delta(X)^a \,Y^b{}_{;a}\,\psi_{,b}    + \, X^r\, g^{aq}\, R_{qr}\, Y^b{}_{;a}\,\psi_{,b}   \big)   \cr
  &\quad - \big( 2\,g^{ij}\, (\nabla_Y X)^a{}_{;i}\, (\psi_{,a})_{;j}  
+\,\Delta(\nabla_Y X)^a \,\psi_{,a}  \big) \cr
&=  2\,g^{ij}\, Y^b \, X^a{}_{;i}\, \psi_{,a;j;b}  + 2\,g^{ij}\, Y^b \, X^a{}_{;i;b}\, (\psi_{,a})_{;j}  
+ Y^b \, \Delta(X)^a{}_{;b} \,\psi_{,a}
 +  Y^b\,X^r\, g^{aq}\, R_{qr;b}\, \psi_{,a}    \cr
  &\quad -2\,g^{ij}\, X^a{}_{;i}\, ((Y^b\,\psi_{,b})_{,a})_{;j}  
   - \Delta(X)^a \,Y^b{}_{;a}\,\psi_{,b}    - X^r\, g^{aq}\, R_{qr}\, Y^b{}_{;a}\,\psi_{,b}      \cr
  &\quad -  2\,g^{ij}\, (\nabla_Y X)^a{}_{;i}\, (\psi_{,a})_{;j}  
- \Delta(\nabla_Y X)^a \,\psi_{,a}  \cr
&=  2\,g^{ij}\, Y^b \, X^a{}_{;i}\, (\psi_{,a;j;b} -\psi_{,b;a;j} ) + 2\,g^{ij}\, Y^b \, X^a{}_{;i;b}\, \psi_{,a;j}  
+ Y^b \, \Delta(X)^a{}_{;b} \,\psi_{,a}
 +  Y^b\,X^r\, g^{aq}\, R_{qr;b}\, \psi_{,a}    \cr
  &\quad -2\,g^{ij}\, X^a{}_{;i}\, Y^b{}_{;a;j}\,\psi_{,b}  -2\,g^{ij}\, X^a{}_{;i}\, Y^b{}_{;a}\,\psi_{,b;j}
-2\,g^{ij}\, X^a{}_{;i}\, Y^b{}_{;j}\,\psi_{,b;a}  \cr
&\quad   - \Delta(X)^a \,Y^b{}_{;a}\,\psi_{,b}    - X^r\, g^{aq}\, R_{qr}\, Y^b{}_{;a}\,\psi_{,b}  -  2\,g^{ij}\, (Y^b\, X^a{}_{;b})_{;i}\, \psi_{,a;j}  
- \Delta(\nabla_Y X)^a \,\psi_{,a}  \cr
&=  2\,g^{ij}\, Y^b \, X^a{}_{;i}\, (\psi_{,a;j;b} -\psi_{,a;b;j} ) + 2\,g^{ij}\, Y^b \, X^c R^a{}_{cbi}   \psi_{,a;j}  
+ Y^b \, \Delta(X)^a{}_{;b} \,\psi_{,a}
 +  Y^b\,X^r\, g^{aq}\, R_{qr;b}\, \psi_{,a}    \cr
  &\quad -2\,g^{ij}\, X^a{}_{;i}\, Y^b{}_{;a;j}\,\psi_{,b}  -     2\,g^{ij}\, X^a{}_{;i}\, Y^b{}_{;a}\,\psi_{,b;j}
-       2\,g^{ij}\, X^a{}_{;i}\, Y^b{}_{;j}\,\psi_{,b;a}   \cr
&\quad   - \Delta(X)^a \,Y^b{}_{;a}\,\psi_{,b}    -   X^r\, g^{aq}\, R_{qr}\, Y^b{}_{;a}\,\psi_{,b}  -   2\,g^{ij}\, Y^b{}_{;i}\, X^a{}_{;b}\, \psi_{,a;j}  
- \Delta(\nabla_Y X)^a \,\psi_{,a}  \cr
&=  2\,g^{ij}\, Y^b \, X^a{}_{;i}\, \psi_{,c}  R^c{}_{ajb} +  2\,g^{ij}\, Y^b \, X^c R_{ecbi} \,g^{ae}  \psi_{,a;j}  
+ Y^b \, \Delta(X)^a{}_{;b} \,\psi_{,a}
 +  Y^b\,X^r\, g^{aq}\, R_{qr;b}\, \psi_{,a}    \cr
  &\quad -2\,g^{ij}\, X^a{}_{;i}\, Y^b{}_{;a;j}\,\psi_{,b}  -      2\,g^{ij}\, X^a{}_{;i}\, Y^b{}_{;a}\,\psi_{,b;j}
-       2\,g^{ij}\, X^a{}_{;i}\, Y^b{}_{;j}\,\psi_{,b;a}   \cr
&\quad   - \Delta(X)^a \,Y^b{}_{;a}\,\psi_{,b}    -   X^r\, g^{aq}\, R_{qr}\, Y^b{}_{;a}\,\psi_{,b}  -   2\,g^{ij}\, Y^b{}_{;i}\, X^a{}_{;b}\, \psi_{,a;j}  
- \Delta(\nabla_Y X)^a \,\psi_{,a}  \cr
&=  \big( 2\,g^{ij}\, Y^b \, X^e{}_{;i}\,   R^a{}_{ejb} 
+ Y^b \, \Delta(X)^a{}_{;b} 
 +  Y^b\,X^r\, g^{aq}\, R_{qr;b}   \cr
  &\quad -2\,g^{ij}\, X^c{}_{;i}\, Y^a{}_{;c;j}   - \Delta(X)^e \,Y^a{}_{;e}    -   X^r\, g^{cq}\, R_{qr}\, Y^a{}_{;c} 
- \Delta(\nabla_Y X)^a \big) \,\psi_{,a}  \cr
  &\quad -      2\,g^{ij}\, X^a{}_{;i}\, Y^b{}_{;a}\,\psi_{,b;j}
-       2\,g^{ij}\, X^a{}_{;i}\, Y^b{}_{;j}\,\psi_{,b;a}    +  2\,g^{ij}\, Y^b \, X^c R_{ecbi} \,g^{ae}  \psi_{,a;j}  
-   2\,g^{ij}\, Y^b{}_{;i}\, X^a{}_{;b}\, \psi_{,a;j}  
\end{align*}
We check that
\begin{align*}
\big(\Delta(\nabla_Y X) & -\nabla_{\Delta Y} X-\nabla_Y(\Delta X)\big)^a \cr
&= 2 g^{ij} Y^b{}_{;i}\, X^a{}_{;b;j} +2g^{ij} Y^b \, X^c{}_{;j}\, R^a{}_{cib}
+ g^{ij} Y^b\,X^c\, R^a{}_{cib;j} +g^{ij} Y^b\, X^a{}_{;p}\, R^p{}_{ibj}
\end{align*}
and then in our last expression for $2m\, \lambda^{-2}\big(\rho(P(X,Y))  + Y^b\, X^a\, V_{,a;b}\big)\psi$, 
the coefficient of $\psi_{,a}  $ can be rewritten as
 \begin{align*} 
 &\quad 2\,g^{ij}\, Y^b \, X^e{}_{;i}\,   R^a{}_{ejb} 
 +  Y^b\,X^r\, g^{aq}\, R_{qr;b}   \cr
  &\quad -2\,g^{ij}\, X^c{}_{;i}\, Y^a{}_{;c;j}   -  (\nabla_{\Delta(X)}Y)^a - (\nabla_{\Delta(Y)}X)^a   -   X^r\, g^{cq}\, R_{qr}\, Y^a{}_{;c} 
 \cr
&\quad -\big(2 g^{ij} Y^b{}_{;i}\, X^a{}_{;b;j} +2g^{ij} Y^b \, X^c{}_{;j}\, R^a{}_{cib}
+ g^{ij} Y^b\,X^c\, R^a{}_{cib;j} +g^{ij} Y^b\, X^a{}_{;p}\, R^p{}_{ibj}\Big) \cr
&=  Y^b\,X^c\, (g^{aq}\, R_{qc;b} - g^{ij} \, R^a{}_{cib;j} ) -   X^b\, g^{cq}\, R_{qb}\, Y^a{}_{;c}  - g^{ij} Y^b\, X^a{}_{;c}\, R^c{}_{ibj}  \cr
  &\quad -2 g^{ij} Y^b{}_{;i}\, X^a{}_{;b;j}  -2\,g^{ij}\, X^c{}_{;i}\, Y^a{}_{;c;j}   -  (\nabla_{\Delta(X)}Y)^a - (\nabla_{\Delta(Y)}X)^a.   
\end{align*}
Here,
\begin{align*}
g^{aq}\, R_{qc;b} & - g^{ij} \, R^a{}_{cib;j} = g^{aq}\, g^{ij}(R_{iqjc;b} - R_{qcib;j}) \cr
&=    - g^{aq}\, g^{ij}(R_{qijc;b} + R_{qcib;j})   \cr
&=     g^{aq}\, g^{ij}(R_{qicb;j} + R_{qibj;c} +   R_{qibc;j}    +R_{qbci;j})\cr
&=     g^{aq}\, g^{ij}( R_{qibj;c}    +R_{qbci;j})\cr
&=    - g^{aq}\, g^{ij}( R_{qijb;c}    +R_{qbic;j})\  ,
\end{align*}
which is symmetric in $b,c$, so the total is symmetric in swapping $X$ and $Y$, as required.  \end{proof}

\subsection{Commutator of a vector field and the differential of one}

We begin by writing 
 $P(X,Y)=P_0(X,Y)+\lambda\,P_1(X,Y)$ to order $\lambda$, where $P_0(X,Y)$ has been chosen to satisfy the lowest order requirements in $\lambda$. 
 Of course, this decomposition of  $P(X,Y)$ is not unique, rearranging the order within a term of $P_0(X,Y)$ will change its value while introducing higher order terms which can go into $P_1(X,Y)$. However, there is one principle we can use to try to solve this problem; if our functions and vector fields are real then, to $O(\lambda^0)$ terms formed from them are Hermitian. The only source of complex numbers (ignoring the Hilbert space) is the imaginary $\lambda$. In other words, we expect $\lambda\,P_1(X,Y)$ to be anti-Hermitian to order $\lambda$. Then from Proposition~\ref{real5} we expect to have to order $\lambda$,
\begin{equation}\label{P0N}
2\,P(X,Y) = P_0(X,Y) + P_0(X,Y)^* +N(X,Y)\ .
\end{equation}
We set
 \begin{align}\label{P0}
P_0(X,Y) &= - \, \widehat{\extd x^i}\, \big(\nabla_{\nabla_i X} Y + \nabla_{\nabla_i Y} X\big) - (2m)^{-1} g^{ij}\,\theta'\big(\nabla_i X\, \nabla_j Y+
\nabla_i Y\, \nabla_j X\big) \nonumber\\ 
&\quad - \theta'\, Y^b\, X^a\, V_{,a;b}    +  m^{-1}g^{ij}\, \theta'\,Y^b \, X^c R_{ecbi} \,g^{ae} (\partial_j \partial_a - \lambda\Gamma^k{}_{aj}\partial_k) 
\end{align}
which gives the order two derivatives of $\psi$ (and therefore the lowest order terms in the algebra of differential operators) in Proposition~\ref{bukk}, and satisfies the condition in Proposition~\ref{gup}.

\begin{lemma}\label{lem:PNP*} To order $\lambda$,
\begin{align*}  P_0&(X,Y)^* +N(X,Y) -P_0(X,Y)\\
   &=  \lambda\, \big(
Y^pX^q R_{qp;i} + X^q{}_{;w} Y^p R^w{}_{piq} + X^q Y^p{}_{;w} R^w{}_{qip} \big) \,   \extd x^i
  \cr
& \quad+\lambda m^{-1}\theta' \, g^{ij}\, \big( (\nabla_j \nabla_iX)^u\,\nabla_u Y +
X^u{}_{;i}\,\nabla_j \nabla_u Y + 
(\nabla_j \nabla_iY)^u\,\nabla_u X +
Y^u{}_{;i}\,\nabla_j \nabla_u X \big)  \cr
&\quad  +  \lambda m^{-1}\theta' \,g^{ij}\, ( (Y^b \, X^c+ X^b \, Y^c ) R_{ecbi})_{;j} \,g^{ae}\, \partial_a  \cr
  &\quad -  \lambda\,m^{-1}\,\theta'\, g^{ij}\, Y^p R_{pi}  \, \nabla_j X 
 - \lambda\,m^{-1}\,\theta'\, g^{ij}\, X^q R_{qi}    \, \nabla_j Y 
\end{align*}
\end{lemma}
\begin{proof} Working to  order $\lambda$,
 \begin{align}  \label{poty}
P_0(X,Y)^* &= - \, \big(\nabla_{\nabla_i X} Y + \nabla_{\nabla_i Y} X\big)^*\, \widehat{\extd x^i} - (2m)^{-1} \theta'\big((\nabla_i X)^*\, (\nabla_j Y)^*+   (\nabla_i Y)^*\, (\nabla_j X)^*\big)   \, g^{ij}\cr
&\quad - \theta'\, Y^b\, X^a\, V_{,a;b}    +   m^{-1} \theta'\,(\partial_a{}^* \partial_j{}^* +\lambda\,\partial_k{}^* \Gamma^k{}_{aj}) \, 
g^{ij}\, Y^b \, X^c R_{ecbi} \,g^{ae}   \cr
&= - \, \big(\nabla_{\nabla_i X} Y + \nabla_{\nabla_i Y} X\big)\, \widehat{\extd x^i} -(2 m)^{-1} \theta'\big((\nabla_i X)\, (\nabla_j Y)+   (\nabla_i Y)\, (\nabla_j X)\big)   \, g^{ij}\cr
&\quad - \theta'\,Y^b\, X^a\, V_{,a;b}    +   m^{-1}\theta' (\partial_a \partial_j +\lambda\,\partial_k \Gamma^k{}_{aj}) \, 
g^{ij}\, Y^b \, X^c R_{ecbi} \,g^{ae}   \cr
&\quad - \,\lambda\, \mathrm{div}\big(\nabla_{\nabla_i X} Y + \nabla_{\nabla_i Y} X\big)\, \widehat{\extd x^i} -  m^{-1} \theta' \lambda\big( \mathrm{div}(\nabla_i X)\, (\nabla_j Y)+   \mathrm{div}(\nabla_i Y)\, (\nabla_j X)\big)   \, g^{ij}\cr
&\quad   +  \lambda m^{-1}\theta'\, (\Gamma^p{}_{ap} \partial_j + \Gamma^p{}_{jp} \partial_a   ) \, 
g^{ij}\, Y^b \, X^c R_{ecbi} \,g^{ae}  
\end{align}
where we use $\mathrm{div}(\del_j)=\Gamma^p{}_{jp}$. If we add the last two lines of (\ref{poty}) to $N(X,Y)$ we get,
to order $\lambda$,
 \begin{align}  \label{poty2}
&
+\lambda\, \extd x^i\big(
 Y^p{}_{;i} \mathrm{div}(X)_{;p} +  X^p{}_{;i} \mathrm{div}(Y)_{;p}  \big) - \,\lambda\, \mathrm{div}\big(\nabla_{\nabla_i X} Y + \nabla_{\nabla_i Y} X\big)\, \widehat{\extd x^i}   \cr
  & + \lambda\, \extd\big( X^q\, Y^p\, R_{pq} + Y^p{}_{;q} X^q{}_{;p}\big) \cr
  & + \lambda\,m^{-1}\,\theta'\, g^{ij}\, \mathrm{div}(Y)_{,i}\, \nabla_j X 
 + \lambda\,m^{-1}\,\theta'\, g^{ij}\,\mathrm{div}(X)_{,i}\, \nabla_j Y \cr
&  -  m^{-1} \theta' \lambda\big( \mathrm{div}(\nabla_i X)\, (\nabla_j Y)+   \mathrm{div}(\nabla_i Y)\, (\nabla_j X)\big)   \, g^{ij}\cr
&   +  \lambda m^{-1} \theta'\, (\Gamma^p{}_{ap} \partial_j + \Gamma^p{}_{jp} \partial_a ) \, 
g^{ij}\, Y^b \, X^c R_{ecbi} \,g^{ae}   \cr
=& \, \lambda\, \extd x^i\big(
-Y^p{}_{;i}X^q R_{qp} - Y^p{}_{;i;q} X^q{}_{;p} - Y^p{}_{;j} X^q{}_{;p}\Gamma^j{}_{iq}
- X^q{}_{;i} Y^p R_{pq} - X^q{}_{;i;p} Y^p{}_{;q} - X^q{}_{;j} Y^p{}_{;q} \Gamma^j{}_{ip}
\big)  \cr
  & + \lambda\, \extd\big( X^q\, Y^p\, R_{pq} + Y^p{}_{;q} X^q{}_{;p}\big) \cr
  & - \lambda\,m^{-1}\,\theta'\, g^{ij}\, (Y^p R_{pi} + \Gamma^u{}_{pi}Y^p{}_{;u}   )  \, \nabla_j X 
 - \lambda\,m^{-1}\,\theta'\, g^{ij}\,( X^q R_{qi} + \Gamma^u{}_{qi}X^q{}_{;u}   ) \, \nabla_j Y \cr
&   +  \lambda m^{-1}  \theta'\, (\Gamma^p{}_{ap} \partial_j + \Gamma^p{}_{jp} \partial_a  ) \, 
g^{ij}\, Y^b \, X^c R_{ecbi} \,g^{ae}   \cr
=& \, \lambda\, \extd x^i\big(
Y^pX^q R_{qp;i} + X^q{}_{;w} Y^p R^w{}_{piq} + X^q Y^p{}_{;w} R^w{}_{qip}
- Y^p{}_{;j} X^q{}_{;p}\Gamma^j{}_{iq}
- X^q{}_{;j} Y^p{}_{;q} \Gamma^j{}_{ip}
\big)  \cr
  & - \lambda\,m^{-1}\,\theta'\, g^{ij}\, (Y^p R_{pi} + \Gamma^u{}_{pi}Y^p{}_{;u}   )  \, \nabla_j X 
 - \lambda\,m^{-1}\,\theta'\, g^{ij}\,( X^q R_{qi} + \Gamma^u{}_{qi}X^q{}_{;u}   ) \, \nabla_j Y \cr
&   +  \lambda m^{-1}  \theta'\, (\Gamma^p{}_{ap} \partial_j + \Gamma^p{}_{jp} \partial_a   ) \, 
g^{ij}\, Y^b \, X^c R_{ecbi} \,g^{ae}   
\end{align}

We use (\ref{orp45}) to rewrite the first two lines of the final expression for $P_0(X,Y)^*$ in  (\ref{poty}) to order $\lambda$ as
 \begin{align*}  
&- \, \big(\nabla_{\nabla_i X} Y + \nabla_{\nabla_i Y} X\big)\, \widehat{\extd x^i} - (2m)^{-1} \theta'\big((\nabla_i X)\, (\nabla_j Y)+   (\nabla_i Y)\, (\nabla_j X)\big)   \, g^{ij}\cr
&\quad - \theta'\,Y^b\, X^a\, V_{,a;b}    +   m^{-1}\theta' (\partial_a \partial_j +\lambda\,\Gamma^k{}_{aj}  \partial_k   ) \, 
g^{ij}\, Y^b \, X^c R_{ecbi} \,g^{ae}   \cr
&= - \, \widehat{\extd x^i} \big(\nabla_{\nabla_i X} Y + \nabla_{\nabla_i Y} X\big)- (2m)^{-1} \theta'   g^{ij}   \big((\nabla_i X)\, (\nabla_j Y)+   (\nabla_i Y)\, (\nabla_j X)\big)    \cr
&\quad - \theta'\,Y^b\, X^a\, V_{,a;b}    +  m^{-1}\theta'   \, 
g^{ij}\, Y^b \, X^c R_{ecbi} \,g^{ae}   (\partial_a \partial_j +\lambda\, \Gamma^k{}_{aj}  \partial_k  ) \cr
& \quad + \lambda\, \big(\nabla_{\nabla_i X} Y + \nabla_{\nabla_i Y} X\big)^p\,\Gamma^i{}_{pj} \extd x^j
+\lambda m^{-1}\theta' \, g^{ij}\,\nabla_j \big(\nabla_{\nabla_i X} Y + \nabla_{\nabla_i Y} X\big)  \cr
& \quad +\lambda m^{-1}\theta' (\Gamma^i{}_{au}\, g^{uj} + \Gamma^j{}_{au}\, g^{iu} )(X^a{}_{;i}\,\nabla_j Y + Y^a{}_{;j}\,\nabla_i X) \cr
&\quad  +   m^{-1}\theta' \,\big[ \partial_j \,,\,g^{ij}\, Y^b \, X^c R_{ecbi} \,g^{ae}\big] \, \partial_a
+  m^{-1}\theta' \,\big[ \partial_a \,,\,g^{ij}\, Y^b \, X^c R_{ecbi} \,g^{ae}\big] \, \partial_j  \cr
&= - \, \widehat{\extd x^i} \big(\nabla_{\nabla_i X} Y + \nabla_{\nabla_i Y} X\big)- (2m)^{-1} \theta'   g^{ij}   \big((\nabla_i X)\, (\nabla_j Y)+   (\nabla_i Y)\, (\nabla_j X)\big)    \cr
&\quad - \theta'\,Y^b\, X^a\, V_{,a;b}    +  m^{-1}\theta'   \, 
g^{ij}\, Y^b \, X^c R_{ecbi} \,g^{ae}   (\partial_a \partial_j +\lambda\, \Gamma^k{}_{aj}  \partial_k  ) \cr
& \quad + \lambda\, \big(\nabla_{\nabla_i X} Y + \nabla_{\nabla_i Y} X\big)^p\,\Gamma^i{}_{pj} \extd x^j
+ \lambda m^{-1}\theta' \, g^{ij}\,\nabla_j \big(\nabla_{\nabla_i X} Y + \nabla_{\nabla_i Y} X\big)  \cr
& \quad + \lambda m^{-1}\theta' (\Gamma^i{}_{au}\, g^{uj} + \Gamma^j{}_{au}\, g^{iu} )(X^a{}_{;i}\,\nabla_j Y + Y^a{}_{;j}\,\nabla_i X) \cr
&\quad  +  \lambda m^{-1}\theta' \,g^{ij}\, (Y^b \, X^c R_{ecbi})_{;j} \,g^{ae}\, \partial_a
+   \lambda m^{-1}\theta' \,g^{ij}\, (Y^b \, X^c R_{ecbi})_{;a} \,g^{ae}\, \partial_j \cr
&\quad -   \lambda m^{-1}\theta' \,Y^b \, X^c R_{ecbi} \,(g^{iu}\,  g^{ae} \Gamma^j{}_{ua} + g^{ij}\,  g^{ue} \Gamma^a{}_{ua}  )\, \partial_j \cr
&\quad -  \lambda m^{-1}\theta' \,Y^b \, X^c R_{ecbi} \,(g^{iu}\,  g^{ae} \Gamma^j{}_{uj} + g^{ij}\,  g^{ue} \Gamma^a{}_{uj}  )\, \partial_a \cr
&= - \, \widehat{\extd x^i} \big(\nabla_{\nabla_i X} Y + \nabla_{\nabla_i Y} X\big)- (2m)^{-1} \theta'   g^{ij}   \big((\nabla_i X)\, (\nabla_j Y)+   (\nabla_i Y)\, (\nabla_j X)\big)    \cr
&\quad - \theta'\,Y^b\, X^a\, V_{,a;b}    +   m^{-1}\theta'   \, 
g^{ij}\, Y^b \, X^c R_{ecbi} \,g^{ae}   (\partial_a \partial_j - \lambda\, \Gamma^k{}_{aj}  \partial_k  ) \cr
& \quad + \lambda\, \big(\nabla_{\nabla_i X} Y + \nabla_{\nabla_i Y} X\big)^p\,\Gamma^i{}_{pj} \extd x^j
+ \lambda m^{-1}\theta' \, g^{ij}\,\nabla_j \big(\nabla_{\nabla_i X} Y + \nabla_{\nabla_i Y} X\big)  \cr
& \quad + \lambda m^{-1}\theta' (\Gamma^i{}_{au}\, g^{uj} + \Gamma^j{}_{au}\, g^{iu} )(X^a{}_{;i}\,\nabla_j Y + Y^a{}_{;j}\,\nabla_i X) \cr
&\quad  +  \lambda m^{-1}\theta' \,g^{ij}\, (Y^b \, X^c R_{ecbi})_{;j} \,g^{ae}\, \partial_a
+   \lambda m^{-1}\theta' \,g^{ij}\, (Y^b \, X^c R_{ecbi})_{;a} \,g^{ae}\, \partial_j \cr
&\quad -   \lambda m^{-1}\theta' \,Y^b \, X^c R_{ecbi} \, g^{ij}\,  g^{ue} \Gamma^a{}_{ua}  \, \partial_j 
-  \lambda m^{-1}\theta' \,Y^b \, X^c R_{ecbi} \,g^{iu}\,  g^{ae} \Gamma^j{}_{uj}  \, \partial_a\ .
\end{align*}
We recognise the first two lines of the last expression as $P_0(X,Y)$, and hence
\begin{align*} 
 &P_0(X,Y)^* +N(X,Y) -P_0(X,Y) =\cr
 &=    \lambda\, \big(\nabla_{\nabla_i X} Y + \nabla_{\nabla_i Y} X\big)^p\,\Gamma^i{}_{pj} \extd x^j
+\lambda m^{-1}\theta' \, g^{ij}\,\nabla_j \big(\nabla_{\nabla_i X} Y + \nabla_{\nabla_i Y} X\big)  \cr
& \quad + \lambda m^{-1}\theta' (\Gamma^i{}_{au}\, g^{uj} + \Gamma^j{}_{au}\, g^{iu} )(X^a{}_{;i}\,\nabla_j Y + Y^a{}_{;j}\,\nabla_i X) \cr
&\quad  +  \lambda m^{-1}\theta' \,g^{ij}\, (Y^b \, X^c R_{ecbi})_{;j} \,g^{ae}\, \partial_a
+   \lambda m^{-1}\theta' \,g^{ij}\, (Y^b \, X^c R_{ecbi})_{;a} \,g^{ae}\, \partial_j \cr
&\quad -   \lambda m^{-1}\theta' \,Y^b \, X^c R_{ecbi} \, g^{ij}\,  g^{ue} \Gamma^a{}_{ua}  \, \partial_j 
-  \lambda m^{-1}\theta' \,Y^b \, X^c R_{ecbi} \,g^{iu}\,  g^{ae} \Gamma^j{}_{uj}  \, \partial_a \cr
&\quad +  \lambda\, \extd x^i\big(
Y^pX^q R_{qp;i} + X^q{}_{;w} Y^p R^w{}_{piq} + X^q Y^p{}_{;w} R^w{}_{qip}
- Y^p{}_{;j} X^q{}_{;p}\Gamma^j{}_{iq}
- X^q{}_{;j} Y^p{}_{;q} \Gamma^j{}_{ip}
\big)  \cr
  &\quad - \lambda\,m^{-1}\,\theta'\, g^{ij}\, (Y^p R_{pi} + \Gamma^u{}_{pi}Y^p{}_{;u}   )  \, \nabla_j X 
 - \lambda\,m^{-1}\,\theta'\, g^{ij}\,( X^q R_{qi} + \Gamma^u{}_{qi}X^q{}_{;u}   ) \, \nabla_j Y \cr
&\quad   +  \lambda m^{-1}  \theta'\, (\Gamma^p{}_{ap} \partial_j + \Gamma^p{}_{jp} \partial_a   ) \, 
g^{ij}\, Y^b \, X^c R_{ecbi} \,g^{ae}   \cr
 &=    \lambda\, ( Y^p{}_{;q}\, X^q{}_{;j} + X^p{}_{;q}\, Y^q{}_{;j})\,\Gamma^j{}_{pi} \extd x^i
+ \lambda m^{-1}\theta' \, g^{ij}\,\nabla_j \big(\nabla_{\nabla_i X} Y + \nabla_{\nabla_i Y} X\big)  \cr
& \quad + \lambda m^{-1}\theta' (\Gamma^i{}_{au}\, g^{uj} + \Gamma^j{}_{au}\, g^{iu} )(X^a{}_{;i}\,\nabla_j Y + Y^a{}_{;j}\,\nabla_i X) \cr
&\quad  +  \lambda m^{-1}\theta' \,g^{ij}\, (Y^b \, X^c R_{ecbi})_{;j} \,g^{ae}\, \partial_a
+   \lambda m^{-1}\theta' \,g^{ij}\, (Y^b \, X^c R_{ecbi})_{;a} \,g^{ae}\, \partial_j \cr
&\quad +  \lambda\, \extd x^i\big(
Y^pX^q R_{qp;i} + X^q{}_{;w} Y^p R^w{}_{piq} + X^q Y^p{}_{;w} R^w{}_{qip}
- Y^p{}_{;j} X^q{}_{;p}\Gamma^j{}_{iq}
- X^q{}_{;j} Y^p{}_{;q} \Gamma^j{}_{ip}
\big)  \cr
  &\quad - \lambda\,m^{-1}\,\theta'\, g^{ij}\, (Y^p R_{pi} + \Gamma^u{}_{pi}Y^p{}_{;u}   )  \, \nabla_j X 
 - \lambda\,m^{-1}\,\theta'\, g^{ij}\,( X^q R_{qi} + \Gamma^u{}_{qi}X^q{}_{;u}   ) \, \nabla_j Y \cr
  &=    \lambda\, \extd x^i\big(
Y^pX^q R_{qp;i} + X^q{}_{;w} Y^p R^w{}_{piq} + X^q Y^p{}_{;w} R^w{}_{qip} \big) 
  \cr
& \quad+ \lambda m^{-1}\theta' \, g^{ij}\,\nabla_j \big(\nabla_{\nabla_i X} Y + \nabla_{\nabla_i Y} X\big)  \cr
& \quad + \lambda m^{-1}\theta' (\Gamma^i{}_{au}\, g^{uj} + \Gamma^j{}_{au}\, g^{iu} )(X^a{}_{;i}\,\nabla_j Y + Y^a{}_{;j}\,\nabla_i X) \cr
&\quad  +  \,\lambda m^{-1}\theta' \,g^{ij}\, (Y^b \, X^c R_{ecbi})_{;j} \,g^{ae}\, \partial_a
+   \lambda m^{-1}\theta' \,g^{ij}\, (Y^b \, X^c R_{ecbi})_{;a} \,g^{ae}\, \partial_j \cr
  &\quad - \lambda\,m^{-1}\,\theta'\, g^{ij}\, (Y^p R_{pi} + \Gamma^u{}_{pi}Y^p{}_{;u}   )  \, \nabla_j X 
 - \lambda\,m^{-1}\,\theta'\, g^{ij}\,( X^q R_{qi} + \Gamma^u{}_{qi}X^q{}_{;u}   ) \, \nabla_j Y \cr
   &=    \lambda\, \extd x^i\big(
Y^pX^q R_{qp;i} + X^q{}_{;w} Y^p R^w{}_{piq} + X^q Y^p{}_{;w} R^w{}_{qip} \big) 
  \cr
& \quad+\lambda m^{-1}\theta' \, g^{ij}\,\nabla_j \big(\nabla_{\nabla_i X} Y + \nabla_{\nabla_i Y} X\big)  \cr
& \quad + \lambda m^{-1}\theta' \,\Gamma^i{}_{au}\, g^{uj} Y^a{}_{;j}\,\nabla_i X
+ \lambda m^{-1}\theta' \,\Gamma^j{}_{au}\, g^{iu} X^a{}_{;i}\,\nabla_j Y   \cr
&\quad  +  \lambda m^{-1}\theta' \,g^{ij}\, (Y^b \, X^c R_{ecbi})_{;j} \,g^{ae}\, \partial_a
+   \lambda m^{-1}\theta' \,g^{ij}\, (Y^b \, X^c R_{ecbi})_{;a} \,g^{ae}\, \partial_j \cr
  &\quad - \lambda\,m^{-1}\,\theta'\, g^{ij}\, Y^p R_{pi}  \, \nabla_j X 
 - \lambda\,m^{-1}\,\theta'\, g^{ij}\, X^q R_{qi}    \, \nabla_j Y \cr
    &=    \lambda\, \extd x^i\big(
Y^pX^q R_{qp;i} + X^q{}_{;w} Y^p R^w{}_{piq} + X^q Y^p{}_{;w} R^w{}_{qip} \big) 
  \cr
& \quad+ \lambda m^{-1}\theta' \, g^{ij}\,\nabla_j \big(    X^u{}_{;i}\,\nabla_u Y + Y^u{}_{;i}\,\nabla_u X \big)  \cr
& \quad + \lambda m^{-1}\theta' \,\Gamma^u{}_{ai}\, g^{ij} Y^a{}_{;j}\,\nabla_u X
+ \lambda m^{-1}\theta' \,\Gamma^u{}_{aj}\, g^{ij} X^a{}_{;i}\,\nabla_u Y   \cr
&\quad  +  \lambda m^{-1}\theta' \,g^{ij}\, (Y^b \, X^c R_{ecbi})_{;j} \,g^{ae}\, \partial_a
+   \lambda m^{-1}\theta' \,g^{ij}\, (Y^b \, X^c R_{ecbi})_{;a} \,g^{ae}\, \partial_j \cr
  &\quad - \lambda\,m^{-1}\,\theta'\, g^{ij}\, Y^p R_{pi}  \, \nabla_j X 
 - \lambda\,m^{-1}\,\theta'\, g^{ij}\, X^q R_{qi}    \, \nabla_j Y 
\end{align*}
which gives the result stated. \end{proof}

\begin{proposition} \label{coordin}
We have
 \begin{align*}
P(X,Y) &= - \, \widehat{\extd x^i}\, \big(\nabla_{\nabla_i X} Y + \nabla_{\nabla_i Y} X\big)
 +{\lambda\over 2m}\,\theta'\, g^{ij}\nabla_j \big(\nabla_{\nabla_i X} Y + \nabla_{\nabla_i Y} X\big)
 \cr
 &\quad - \theta'\, Y^b\, X^a\, V_{,a;b}    +  m^{-1}g^{ij}\, \theta'\,Y^b \, X^c R_{ecbi} \,g^{ae} (\partial_j \partial_a - \lambda\Gamma^k{}_{aj}\partial_k) \cr
  &\quad - (2m)^{-1} g^{ij}\,\theta'\big(\nabla_i X\, \nabla_j Y+\nabla_i Y\, \nabla_j X\big)
  +  {\lambda\over 2m} g^{ij}\,\theta' \,  \Gamma^k{}_{qj} \, \big ( X^q{}_{;i\, }\nabla_k Y+  Y^q{}_{;i\, }\nabla_k X\big) \cr
       &\quad +  \tfrac12\,  \lambda\, \big(
Y^pX^q R_{qp;i} + X^q{}_{;w} Y^p R^w{}_{piq} + X^q Y^p{}_{;w} R^w{}_{qip} \big) \,  ( \extd x^i- m^{-1} \theta'\,g^{ij}\del_j)
  \cr
&\quad   +  {\lambda\over 2m}\theta' \, \big(
Y^pX^q (R_{pi;q}+ R_{qi;p} -R_{qp;i}) -Y^p{}_{;w} \, X^q R^w{}_{piq}   -  X^q{}_{;w} \, Y^p  R^w{}_{qip} 
  \big) \,g^{ji}\, \partial_j   \cr
  &\quad - {\lambda\over 2m}\,\theta'\, g^{ij}\, Y^p R_{pi}  \, \nabla_j X 
 -{\lambda\over 2m}\,\theta'\, g^{ij}\, X^q R_{qi}    \, \nabla_j Y \ .
\end{align*}
Together with Proposition~\ref{gup} this gives the commutator as
\[
[    Y  , \extd X  ] =\lambda\, \extd(\nabla_Y X)+\lambda\, P(X,Y)\ .
\]
\end{proposition}
\begin{proof} We use equation (\ref{P0}) for $P_0(X,Y)$ and Lemma~\ref{lem:PNP*} for $P_0(X,Y)^* +N(X,Y) -P_0(X,Y)$. Then from Proposition~\ref{real5}, we have
\[
2\,P(X,Y) = P_0(X,Y) + P_0(X,Y)^* +N(X,Y)= 2\,P_0(X,Y) + P_0(X,Y)^* +N(X,Y)-P_0(X,Y) 
\]
giving 
 \begin{align}
2 P(X,Y) &= - 2 \widehat{\extd x^i}\, \big(\nabla_{\nabla_i X} Y + \nabla_{\nabla_i Y} X\big) - m^{-1} g^{ij}\,\theta'\big(\nabla_i X\, \nabla_j Y+
\nabla_i Y\, \nabla_j X\big) \nonumber\cr
&\quad - 2\theta'\, Y^b\, X^a\, V_{,a;b}    +  2m^{-1}g^{ij}\, \theta'\,Y^b \, X^c R_{ecbi} \,g^{ae} (\partial_j \partial_a - \lambda\Gamma^k{}_{aj}\partial_k) \nonumber\cr
     &\quad + \lambda\, \big(
Y^pX^q R_{qp;i} + X^q{}_{;w} Y^p R^w{}_{piq} + X^q Y^p{}_{;w} R^w{}_{qip} \big) \,   \extd x^i
  \nonumber\cr
& \quad+ \lambda m^{-1}\theta' \, g^{ij}\, \big( (\nabla_j \nabla_iX)^u\,\nabla_u Y +
X^u{}_{;i}\,\nabla_j \nabla_u Y + 
(\nabla_j \nabla_iY)^u\,\nabla_u X +
Y^u{}_{;i}\,\nabla_j \nabla_u X \big)  \nonumber\cr
&\quad  +  \lambda m^{-1}\theta' \,g^{ij}\, ( (Y^b \, X^c+ X^b \, Y^c ) R_{ecbi})_{;j} \,g^{ae}\, \partial_a  \nonumber\\ 
  &\quad - \lambda m^{-1}\,\theta'\, g^{ij}\, Y^p R_{pi}  \, \nabla_j X 
 -\lambda m^{-1}\,\theta'\, g^{ij}\, X^q R_{qi}    \, \nabla_j Y \ .
\end{align}
We split this first result for $P(X,Y)$ into well defined bits:
 \begin{align*}
2P(X,Y) &=  - 2 \widehat{\extd x^i}\, \big(\nabla_{\nabla_i X} Y + \nabla_{\nabla_i Y} X\big)
 +\lambda m^{-1}\,\theta'\, g^{ij}\nabla_j \big(\nabla_{\nabla_i X} Y + \nabla_{\nabla_i Y} X\big)
 \cr
 &\quad - 2\theta'\, Y^b\, X^a\, V_{,a;b}    + 2 m^{-1}g^{ij}\, \theta'\,Y^b \, X^c R_{ecbi} \,g^{ae} (\partial_j \partial_a - \lambda\Gamma^k{}_{aj}\partial_k) \cr
  &\quad - m^{-1} g^{ij}\,\theta'\big(\nabla_i X\, \nabla_j Y+\nabla_i Y\, \nabla_j X\big)
  +\lambda m^{-1} g^{ij}\,\theta'\big(\nabla_{\nabla_i X} \nabla_j Y+\nabla_{\nabla_i Y} \nabla_j X\big) \cr
       &\quad +   \lambda\, \big(
Y^pX^q R_{qp;i} + X^q{}_{;w} Y^p R^w{}_{piq} + X^q Y^p{}_{;w} R^w{}_{qip} \big) \,   \extd x^i
  \cr
&\quad    - \lambda m^{-1} g^{ij}\,\theta'\big(\nabla_{\nabla_i X} \nabla_j Y+\nabla_{\nabla_i Y} \nabla_j X\big) \cr
 &\quad - \lambda m^{-1}\,\theta'\, g^{ij}\nabla_j \big(\nabla_{\nabla_i X} Y + \nabla_{\nabla_i Y} X\big) \cr
& \quad+ \lambda m^{-1}\theta' \, g^{ij}\, \big( (\nabla_j \nabla_iX)^u\,\nabla_u Y +
X^u{}_{;i}\,\nabla_j \nabla_u Y + 
(\nabla_j \nabla_iY)^u\,\nabla_u X +
Y^u{}_{;i}\,\nabla_j \nabla_u X \big)  \cr
&\quad  +  \lambda m^{-1}\theta' \,g^{ij}\, ( (Y^b \, X^c+ X^b \, Y^c ) R_{ecbi})_{;j} \,g^{ae}\, \partial_a  \cr
  &\quad - \lambda m^{-1}\,\theta'\, g^{ij}\, Y^p R_{pi}  \, \nabla_j X 
 - \lambda m^{-1}\,\theta'\, g^{ij}\, X^q R_{qi}    \, \nabla_j Y \ .
\end{align*}
The last five lines of this are
\begin{align*}
&\quad    - \lambda m^{-1} g^{ij}\,\theta'\big(\nabla_{\nabla_i X} \nabla_j Y+\nabla_{\nabla_i Y} \nabla_j X\big) + \lambda m^{-1}\theta' \, g^{ij}\, \big( 
\Gamma^k{}_{uj}\, X^u{}_{;i} \,\nabla_k Y +\Gamma^k{}_{uj}\, Y^u{}_{;i} \,\nabla_k X \big)  \cr
&\quad  +  \lambda m^{-1}\theta' \,g^{ij}\, ( (Y^b \, X^c+ X^b \, Y^c ) R_{ecbi})_{;j} \,g^{ae}\, \partial_a   - \lambda m^{-1}\,\theta'\, g^{ij}\, Y^p R_{pi}  \, \nabla_j X 
 - 2\lambda m^{-1}\,\theta'\, g^{ij}\, X^q R_{qi}    \, \nabla_j Y \cr
 &=    - \lambda m^{-1} g^{ij}\,\theta'\big(\nabla_{\nabla_i X} \nabla_j Y+\nabla_{\nabla_i Y} \nabla_j X\big) + \lambda m^{-1}\theta' \, g^{ij}\, \big( 
\Gamma^k{}_{uj}\, X^u{}_{;i} \,\nabla_k Y +\Gamma^k{}_{uj}\, Y^u{}_{;i} \,\nabla_k X \big)  \cr
&\quad  +  \lambda m^{-1}\theta' \,g^{ij}\, ( (Y^b \, X^c+ X^b \, Y^c ) R_{ecbi})_{;j} \,g^{ae}\, \partial_a   - \lambda m^{-1}\,\theta'\, g^{ij}\, Y^p R_{pi}  \, \nabla_j X 
 - 2\lambda m^{-1}\,\theta'\, g^{ij}\, X^q R_{qi}    \, \nabla_j Y \cr
  &=    - \lambda m^{-1} g^{ij}\,\theta'\big(   X^q{}_{;i}\, Y^p{}_{;j;q}\,\del_p   
  + Y^p{}_{;i}\, X^q{}_{;j;p}\,\del_q  \big)   +  \lambda m^{-1}\theta' \,g^{ij}\, ( (Y^b \, X^c+ X^b \, Y^c ) R_{ecbi})_{;j} \,g^{ae}\, \partial_a  \cr
  &\quad - \lambda m^{-1}\,\theta'\, g^{ij}\, Y^p R_{pi}  \, \nabla_j X 
 - \lambda m^{-1}\,\theta'\, g^{ij}\, X^q R_{qi}    \, \nabla_j Y \ .
\end{align*}
Then
 \begin{align*}
2P(X,Y) &= - 2 \widehat{\extd x^i}\, \big(\nabla_{\nabla_i X} Y + \nabla_{\nabla_i Y} X\big)
 +\lambda m^{-1}\theta'\, g^{ij}\nabla_j \big(\nabla_{\nabla_i X} Y + \nabla_{\nabla_i Y} X\big)
 \cr
 &\quad -2 \theta'\, Y^b\, X^a\, V_{,a;b}    + 2 m^{-1}g^{ij}\, \theta'\,Y^b \, X^c R_{ecbi} \,g^{ae} (\partial_j \partial_a - \lambda\Gamma^k{}_{aj}\partial_k) \cr
  &\quad - m^{-1} g^{ij}\,\theta'\big(\nabla_i X\, \nabla_j Y+\nabla_i Y\, \nabla_j X\big)
  +\lambda m^{-1}g^{ij}\,\theta'\big(\nabla_{\nabla_i X} \nabla_j Y+\nabla_{\nabla_i Y} \nabla_j X\big) \cr
       &\quad + \ \lambda\, \big(
Y^pX^q R_{qp;i} + X^q{}_{;w} Y^p R^w{}_{piq} + X^q Y^p{}_{;w} R^w{}_{qip} \big) \,   \extd x^i
  \cr
&\quad    - \lambda m^{-1} g^{ij}\,\theta'\big(   X^q{}_{;i}\, Y^p{}_{;j;q}\,\del_p   
  + Y^p{}_{;i}\, X^q{}_{;j;p}\,\del_q  \big) \cr
&\quad  +  \lambda m^{-1}\theta' \,g^{uw}\, ( Y^p \, X^q R_{iqpu} +  X^q \, Y^p  R_{ipqu})_{;w} \,g^{ji}\, \partial_j  \cr
  &\quad - \lambda m^{-1}\,\theta'\, g^{ij}\, Y^p R_{pi}  \, \nabla_j X 
 -2  \lambda\,m^{-1}\,\theta'\, g^{ij}\, X^q R_{qi}    \, \nabla_j Y \ .
\end{align*}
We can rewrite the fourth and sixth lines as
 \begin{align*}
       &\quad   \lambda\, \big(
Y^pX^q R_{qp;i} + X^q{}_{;w} Y^p R^w{}_{piq} + X^q Y^p{}_{;w} R^w{}_{qip} \big) \,   \extd x^i
  \cr
&\quad  +  \lambda m^{-1}\theta' \,g^{uw}\, ( Y^p \, X^q R_{iqpu} +  X^q \, Y^p  R_{ipqu})_{;w} \,g^{ji}\, \partial_j  \cr
&=  \lambda\, \big(
Y^pX^q R_{qp;i} + X^q{}_{;w} Y^p R^w{}_{piq} + X^q Y^p{}_{;w} R^w{}_{qip} \big) \,   \extd x^i
  \cr
&\quad  +  \lambda m^{-1}\theta' \,( -Y^p \, X^q R^w{}_{piq} -  X^q \, Y^p  R^w{}_{qip})_{;w} \,g^{ji}\, \partial_j  \cr
&=   \lambda\, \big(
Y^pX^q R_{qp;i} + X^q{}_{;w} Y^p R^w{}_{piq} + X^q Y^p{}_{;w} R^w{}_{qip} \big) \,  ( \extd x^i-m^{-1} \theta'\,g^{ij}\del_j)
  \cr
&\quad  +  \lambda m^{-1}\theta' \, \big(
Y^pX^q R_{qp;i} + X^q{}_{;w} Y^p R^w{}_{piq} + X^q Y^p{}_{;w} R^w{}_{qip} +( -Y^p \, X^q R^w{}_{piq} -  X^q \, Y^p  R^w{}_{qip})_{;w}\big) \,g^{ji}\, \partial_j  \cr
&=  \lambda\, \big(
Y^pX^q R_{qp;i} + X^q{}_{;w} Y^p R^w{}_{piq} + X^q Y^p{}_{;w} R^w{}_{qip} \big) \,  ( \extd x^i-m^{-1} \theta'\,g^{ij}\del_j)
  \cr
&\quad  +  \lambda m^{-1}\theta' \, \big(
Y^pX^q R_{qp;i}  -Y^p{}_{;w} \, X^q R^w{}_{piq} -Y^p{} \, X^q R^w{}_{piq;w}  -  X^q{}_{;w} \, Y^p  R^w{}_{qip} 
-  X^q \, Y^p  R^w{}_{qip;w}   \big) \,g^{ji}\, \partial_j  \ .
\end{align*}
and we note that
\[
-R^w{}_{piq;w} = R^w{}_{pqw;i}  + R^w{}_{pwi;q}  = - R_{pq;i} + R_{pi;q},
\]
so the fourth and sixth lines become
 \begin{align*}
 &=  \lambda\, \big(
Y^pX^q R_{qp;i} + X^q{}_{;w} Y^p R^w{}_{piq} + X^q Y^p{}_{;w} R^w{}_{qip} \big) \,  ( \extd x^i-m^{-1} \theta'\,g^{ij}\del_j)
  \cr
&\quad  +  \lambda m^{-1}\theta' \, \big(
Y^pX^q (R_{pi;q}+ R_{qi;p} -R_{qp;i}) -Y^p{}_{;w} \, X^q R^w{}_{piq}   -  X^q{}_{;w} \, Y^p  R^w{}_{qip} 
  \big) \,g^{ji}\, \partial_j  \ ,
\end{align*}
which gives
 \begin{align*}
2P(X,Y) &= - 2 \widehat{\extd x^i}\, \big(\nabla_{\nabla_i X} Y + \nabla_{\nabla_i Y} X\big)
 +\lambda m^{-1}\theta'\, g^{ij}\nabla_j \big(\nabla_{\nabla_i X} Y + \nabla_{\nabla_i Y} X\big)
 \cr
 &\quad -2 \theta'\, Y^b\, X^a\, V_{,a;b}    + 2 m^{-1}g^{ij}\, \theta'\,Y^b \, X^c R_{ecbi} \,g^{ae} (\partial_j \partial_a - \lambda\Gamma^k{}_{aj}\partial_k) \cr
  &\quad - m^{-1} g^{ij}\,\theta'\big(\nabla_i X\, \nabla_j Y+\nabla_i Y\, \nabla_j X\big)
  +\lambda m^{-1} g^{ij}\,\theta'\big(\nabla_{\nabla_i X} \nabla_j Y+\nabla_{\nabla_i Y} \nabla_j X\big) \cr
       &\quad +    \lambda\, \big(
Y^pX^q R_{qp;i} + X^q{}_{;w} Y^p R^w{}_{piq} + X^q Y^p{}_{;w} R^w{}_{qip} \big) \,  ( \extd x^i-m^{-1} \theta'\,g^{ij}\del_j)
  \cr
&\quad    - \lambda m^{-1} g^{ij}\,\theta'\big(   X^q{}_{;i}\, Y^p{}_{;j;q}\,\del_p   
  + Y^p{}_{;i}\, X^q{}_{;j;p}\,\del_q  \big) \cr
&\quad   +  \lambda m^{-1}\theta' \, \big(
Y^pX^q (R_{pi;q}+ R_{qi;p} -R_{qp;i}) -Y^p{}_{;w} \, X^q R^w{}_{piq}   -  X^q{}_{;w} \, Y^p  R^w{}_{qip} 
  \big) \,g^{ji}\, \partial_j   \cr
  &\quad - \lambda m^{-1}\,\theta'\, g^{ij}\, Y^p R_{pi}  \, \nabla_j X 
 -\lambda m^{-1}\,\theta'\, g^{ij}\, X^q R_{qi}    \, \nabla_j Y \ .
\end{align*}
Finally, we combine the last part of the third line with the fifth line to give the stated answer. 
\end{proof}

\begin{remark}\label{rembasis} \rm The formula for $P(X,Y)$ is written in a coordinate basis but is both coordinate invariant and applies in any (local) basis. To see this, we set a new basis of 1-forms and vector fields
\[
\del_i=\Lambda^a{}_{i}  \, \del_a\ ,\quad \extd x^i=\Lambda^{-1i}{}_b\, f^b   \ .
\]
For the purposes if this remark only, we use $a,b,c$ for the new basis labels and $i,j,k$ for the coordinate basis. Then, for $E_i\, \extd x^i$ a 1-form valued in a vector bundle (for which we do not write indices)
\begin{align*}
\widehat{\extd x^i}\,   E_i & - {\lambda\over 2m}\theta'\, g^{ij}\,\nabla_j E_i = 
\widehat{(\Lambda^{-1i}{}_c\, f^c) }\, \Lambda^b{}_{i}  \, E_b - {\lambda\over 2m}\theta'\, \Lambda^{-1i}{}_c\, g^{ca}\,\nabla_a (\Lambda^b{}_{i}  \, E_b)  \cr
&= 
\widehat{f^c }\, \Lambda^b{}_{i}  \, \Lambda^{-1i}{}_c\,E_b  -   {\lambda\over 2m}\theta'\, g^{pq}  \del_q( \Lambda^{-1i}{}_c)\, \Lambda^c{}_{p}    
\,  \Lambda^b{}_{i}  \, E_b  - {\lambda\over 2m}\theta'\, \Lambda^{-1i}{}_c\, g^{ca}\,\nabla_a (\Lambda^b{}_{i}  \, E_b)  \cr
&= 
\widehat{f^c }\,E_c  -   {\lambda\over 2m}\theta'\, g^{ac}  \del_a( \Lambda^{-1i}{}_c)   
\,  \Lambda^b{}_{i}  \, E_b
 - {\lambda\over 2m}\theta'\, \Lambda^{-1i}{}_c\, g^{ca}\,\nabla_a (\Lambda^b{}_{i}  \, E_b)  \cr
 &= 
\widehat{f^c }\,E_c  -   {\lambda\over 2m}\theta'\, g^{ac}  \del_a( \Lambda^{-1i}{}_c
\,  \Lambda^b{}_{i} ) \, E_b
 - {\lambda\over 2m}\theta'\, \Lambda^{-1i}{}_c\, g^{ca}\,\Lambda^b{}_{i}  \, \nabla_a (E_b)  \cr
&=     \widehat{\extd x^i}\,   E_i  - {\lambda\over 2m}\theta'\, g^{ab}\,\nabla_a E_b\ .
\end{align*}
This equation serves two purposes. First, change to another coordinate basis shows the coordinate independence of the expression on the noncommutative algebra. Second, it provides a formula in a more general context than a coordinate basis, which will be useful later.
Next we define the Christoffel symbols for any basis. To do this, calculate
\begin{align*}
\nabla_b \del_a &= \Lambda^{-1p}{}_b\, \nabla_p( \Lambda^{-1j}{}_a\, \del_j)= 
 \Lambda^{-1p}{}_b\, \del_p( \Lambda^{-1j}{}_a)\, \del_j +  \Lambda^{-1p}{}_b\, \Lambda^{-1j}{}_a\, \Gamma^k{}_{pj}\, \del_k \cr
 &=   \del_b( \Lambda^{-1j}{}_a)\, \Lambda^c{}_{j} \, \del_c +  \Lambda^{-1p}{}_b\, \Lambda^{-1j}{}_a\, \Gamma^k{}_{pj}
 \, \Lambda^c{}_{k} \, \del_c \cr
  &=  - \Lambda^{-1j}{}_a\,  \del_b(\Lambda^c{}_{j}) \, \del_c +  \Lambda^{-1p}{}_b\, \Lambda^{-1j}{}_a\, \Gamma^k{}_{pj}
 \, \Lambda^c{}_{k} \, \del_c\ ,
\end{align*}
as $ \Lambda^{-1j}{}_a\, \Lambda^c{}_{j}=\delta^c{}_a$. We define $\Gamma^c{}_{ab}$ in the new basis by $\nabla_b \del_a =\Gamma^c{}_{ab}\, \del_c$. Then
\begin{align*}
\partial_j \partial_i - \lambda\Gamma^k{}_{ij}\partial_k&= \Lambda^a{}_{j}  \, \del_a\, \Lambda^b{}_{i}  \, \del_b
- \lambda\Gamma^k{}_{ij} \, \Lambda^c{}_{k}  \, \del_c \cr
&=  \Lambda^a{}_{j}  \,  \Lambda^b{}_{i}  \, \del_a\,\del_b + \lambda\, \Lambda^a{}_{j}  \, \del_a(\Lambda^c{}_{i})  \, \del_c
- \lambda\Gamma^k{}_{ij} \, \Lambda^c{}_{k}  \, \del_c\cr
&=  \Lambda^a{}_{j}  \,  \Lambda^b{}_{i}  \,\big( \del_a\,\del_b + \lambda\, (\Lambda^{-1p}{}_{b}  \, \del_a(\Lambda^c{}_{p})  
- \Lambda^{-1p}{}_{b}  \, \Lambda^{-1q}{}_{a}  \,\Gamma^k{}_{pq} \, \Lambda^c{}_{k})  \, \del_c \big)\cr
&=  \Lambda^a{}_{j}  \,  \Lambda^b{}_{i}  \,\big( \partial_a \partial_b - \lambda\Gamma^c{}_{ab}\partial_c \big)\ .
\end{align*}
The last of the expression we need to consider for $2m P(X,Y)$ is
 \begin{align*}
  &-  g^{ij}\,\theta'\, \nabla_i X\, \nabla_j Y
  +  {\lambda} g^{ij}\,\theta' \,  \Gamma^k{}_{qj} \,  X^q{}_{;i\, }\nabla_k Y\cr
    &= -  g^{ij}\,\theta'\, \nabla_i X\,(\Lambda^a{}_{j}  \,  \nabla_a Y)
  +  {\lambda} g^{ij}\,\theta' \,  \Gamma^k{}_{qj} \,  X^q{}_{;i\, }\nabla_k Y\cr
      &= -  g^{ij}\,\theta'\, \Lambda^a{}_{j}  \,  \nabla_i X\,(\nabla_a Y)
       - {\lambda} g^{ij}\,\theta'\, X^q{}_{;i }\,\partial_q(\Lambda^a{}_{j})  \,  \nabla_a Y
  +  {\lambda} g^{ij}\,\theta' \,  \Gamma^k{}_{qj} \,  X^q{}_{;i}\nabla_k Y\cr
        &= -  g^{ab}\,\theta' \,  \nabla_b X\,\nabla_a Y
       + {\lambda} \theta'\, (-\partial_q(\Lambda^c{}_{j})  
  +   \Gamma^k{}_{qj}  \, \Lambda^c{}_{k} ) \, g^{ij}\, X^q{}_{;i} \,\nabla_c Y\cr
          &= -  g^{ab}\,\theta' \,  \nabla_b X\,\nabla_a Y
       + {\lambda} \theta'\, (   \Gamma^c{}_{ab}  \, \Lambda^b{}_{q}  \Lambda^a{}_{j}  ) \, g^{ij}\, X^q{}_{;i} \,\nabla_c Y\cr
                 &= - g^{ab}\,\theta' \,  \nabla_b X\,\nabla_a Y
                  +  {\lambda} g^{da}\,\theta' \,  \Gamma^c{}_{ba} \,  X^b{}_{;d}\nabla_c Y
\end{align*}
and so the first three lines of the formula for $P(X,Y)$ in Proposition~\ref{coordin} are coordinate independent and true in more general bases (given the formula for the Christoffel symbols used here). The remaining lines are manifestly coordinate invariant by standard differential geometry. 
\end{remark}

\subsection{Check of Schr\"odinger representation of differential of a vector field}

It remains to check an identity used in the derivation that amounts to consistency of the proposed Schr\"odinger representation
of differentials of vector fields.

\begin{proposition}
\begin{align*}
\rho(P(X,Y)&-\tfrac12 P_0(X,Y) -\tfrac12 P_0(X,Y)^* - \tfrac12\, N(X,Y)     )\psi 
\cr
 &=   {\lambda^2\over 4m} g^{ab}\,\Gamma^i{}_{ab}   \big(
Y^pX^q R_{qp;i} + X^q{}_{;w} Y^p R^w{}_{piq} + X^q Y^p{}_{;w} R^w{}_{qip} \big) \,  \psi.
\end{align*}
\end{proposition}
\begin{proof} First we calculate
\begin{align*}
&2m\rho( \widehat{\extd x^i}\, \big(\nabla_{\nabla_i X} Y + \nabla_{\nabla_i Y} X\big) )\psi
+g^{ij}\,  \rho\big(\nabla_i X\, \nabla_j Y+ \nabla_i Y\, \nabla_j X\big) \psi  \cr
&=2m  \lambda\rho( \widehat{\extd x^i} )\, ( X^b{}_{;i}\, Y^a{}_{;b} + Y^b{}_{;i}\, X^a{}_{;b} )\,\psi_{,a}
+{\lambda} g^{ij}\,  \big(\rho(\nabla_i X)\,  Y^a{}_{;j}\,\psi_{,a}      +       \rho(\nabla_i Y)\, X^a{}_{;j}\,\psi_{,a}  \big)  \cr
&= \lambda^2( 2\,g^{ic}\,\tfrac{\del}{\partial x^c}-g^{pq}\,\Gamma^i{}_{pq}    )\, ( X^b{}_{;i}\, Y^a{}_{;b} + Y^b{}_{;i}\, X^a{}_{;b} )\,\psi_{,a}
+{\lambda} g^{ij}\,  \big(\rho(\nabla_i X)\,  Y^a{}_{;j}\,\psi_{,a}      +       \rho(\nabla_i Y)\, X^a{}_{;j}\,\psi_{,a}  \big)  \cr
&=
 \lambda^2( -g^{pq}\,\Gamma^i{}_{pq}    )\, ( X^b{}_{;i}\, Y^a{}_{;b} + Y^b{}_{;i}\, X^a{}_{;b} )\,\psi_{,a} \cr
&\quad+ \lambda^2( 2\,g^{ic}   )\, ( X^b{}_{;i}\, Y^a{}_{;b} + Y^b{}_{;i}\, X^a{}_{;b} )\,\psi_{,a;c} \cr
&\quad+ \lambda^2( 2\,g^{ic}   )\, ( X^b{}_{;i;c}\, Y^a{}_{;b} + Y^b{}_{;i;c}\, X^a{}_{;b} 
+ X^b{}_{;i}\, Y^a{}_{;b;c} + Y^b{}_{;i}\, X^a{}_{;b;c}   )\,\psi_{,a} \cr
&\quad +\lambda^2  g^{ij}\,  \big(   X^b{}_{;i}\,  Y^a{}_{;j}    +     Y^b{}_{;i}\, X^a{}_{;j} \big) \,\psi_{,a;b} 
+\lambda^2  g^{ij}\,  \big(   X^b{}_{;i}\,  Y^a{}_{;j;b}    +     Y^b{}_{;i}\, X^a{}_{;j;b} \big) \,\psi_{,a}   \cr
&\quad+ \lambda^2( 2\,g^{ic}  )\, ( X^b{}_{;k}\, Y^a{}_{;b} + Y^b{}_{;k}\, X^a{}_{;b} )\,\psi_{,a} \,\Gamma^k{}_{ic}   \cr
&\quad  +\lambda^2  g^{ij}\,  \big(   X^b{}_{;i}\,  Y^a{}_{;k}\,\psi_{,a}      +     Y^b{}_{;i}\,X^a{}_{;k}\,\psi_{,a}  \big) 
 \,\Gamma^k_{bj}\ .
\end{align*}
Now from Proposition~\ref{bukk},
\begin{align*}
2m\,& \lambda^{-2}\rho(P(X,Y)-P_0(X,Y))\psi 
\cr
&=  \big(  Y^b\,X^c\, (   - g^{aq}\, g^{ij}(R_{qijc;b} + R_{qcib;j})   ) -   g^{cq}\, R_{qb}\, (X^b\, Y^a{}_{;c}  + Y^b\, X^a{}_{;c}) \cr
  &\quad -2 g^{ij} Y^b{}_{;i}\, X^a{}_{;b;j}  -2\,g^{ij}\, X^c{}_{;i}\, Y^a{}_{;c;j}   -  (\nabla_{\Delta(X)}Y)^a - (\nabla_{\Delta(Y)}X)^a    \big) \,\psi_{,a}  \cr
&\quad+   2\,g^{ij}   \, ( X^b{}_{;i;j}\, Y^a{}_{;b} + Y^b{}_{;i;j}\, X^a{}_{;b} 
+ X^b{}_{;i}\, Y^a{}_{;b;j} + Y^b{}_{;i}\, X^a{}_{;b;j}   )\,\psi_{,a} \cr
&\quad 
+ g^{ij}\,  \big(   X^b{}_{;i}\,  Y^a{}_{;j;b}    +     Y^b{}_{;i}\, X^a{}_{;j;b} \big) \,\psi_{,a}   \cr
&\quad+  g^{ij}  \, ( X^b{}_{;k}\, Y^a{}_{;b} + Y^b{}_{;k}\, X^a{}_{;b} )\,\psi_{,a} \,\Gamma^k{}_{ij}   \cr
&\quad  +g^{ij}\,  \big(   X^b{}_{;i}\,  Y^a{}_{;k}\,\psi_{,a}      +     Y^b{}_{;i}\,X^a{}_{;k}\,\psi_{,a}  \big) 
 \,\Gamma^k_{bj}    \cr
 &=  \big(  -Y^b\,X^c\, g^{aq}\, g^{ij}(R_{qijc;b} + R_{qcib;j})    -   g^{cq}\, R_{qb}\, (X^b\, Y^a{}_{;c}  + Y^b\, X^a{}_{;c}) \big)\,\psi_{,a} \cr
&\quad+   g^{ij}   \, ( X^b{}_{;i;j}\, Y^a{}_{;b} + Y^b{}_{;i;j}\, X^a{}_{;b}   )\,\psi_{,a} 
+ g^{ij}\,  \big(   X^b{}_{;i}\,  Y^a{}_{;j;b}    +     Y^b{}_{;i}\, X^a{}_{;j;b} \big) \,\psi_{,a}   \cr
&\quad+  g^{ij}  \, ( X^b{}_{;k}\, Y^a{}_{;b} + Y^b{}_{;k}\, X^a{}_{;b} )\,\psi_{,a} \,\Gamma^k{}_{ij}   
  +g^{ij}\,  \big(   X^b{}_{;i}\,  Y^a{}_{;k}     +     Y^b{}_{;i}\,X^a{}_{;k} \big) \,\psi_{,a} 
 \,\Gamma^k_{bj}    \ .
\end{align*}
We use the symmetries of the Riemann tensor
\[
g^{ij}\, R_{qijc;b} = g^{ij}\, R_{jcqi;b} = - g^{ij}\, R_{jcbq;i} - g^{ij}\, R_{jcib;q} =  g^{ij}\, R_{qbic;j} - R_{cb;q} 
\]
to rewrite this as
\begin{align*}
2m\,& \lambda^{-2}\rho(P(X,Y)-P_0(X,Y))\psi 
\cr
 &=  \big(  (Y^b\,X^c+ X^b\,Y^c)\, g^{ae}\, g^{ij}\, R_{ecbi;j}    + Y^b\,X^c\, g^{aq}\, R_{cb;q}  -   g^{ji}\, R_{ib}\, (X^b\, Y^a{}_{;j}  + Y^b\, X^a{}_{;j}) \big)\,\psi_{,a}\cr
&\quad+   g^{ij}   \, ( X^b{}_{;i;j}\, Y^a{}_{;b} + Y^b{}_{;i;j}\, X^a{}_{;b}   )\,\psi_{,a} 
+ g^{ij}\,  \big(   X^b{}_{;i}\,  Y^a{}_{;j;b}    +     Y^b{}_{;i}\, X^a{}_{;j;b} \big) \,\psi_{,a}   \cr
&\quad+  g^{ij}  \, ( X^b{}_{;k}\, Y^a{}_{;b} + Y^b{}_{;k}\, X^a{}_{;b} )\,\psi_{,a} \,\Gamma^k{}_{ij}   
  +g^{ij}\,  \big(   X^b{}_{;i}\,  Y^a{}_{;k}     +     Y^b{}_{;i}\,X^a{}_{;k} \big) \,\psi_{,a} 
 \,\Gamma^k_{bj}    \ .
\end{align*}
Next we calculate
\begin{align*} 
 &\rho(P_0(X,Y)^* +N(X,Y) -P_0(X,Y) )(\psi)\cr
     &=  \lambda\, \big(
Y^pX^q R_{qp;i} + X^q{}_{;w} Y^p R^w{}_{piq} + X^q Y^p{}_{;w} R^w{}_{qip} \big) \,  \rho( \extd x^i)(\psi)
  \cr
& \quad+ \lambda m^{-1} \, g^{ij}\, \rho\big( (\nabla_j \nabla_iX)^u\,\nabla_u Y +
X^u{}_{;i}\,\nabla_j \nabla_u Y + 
(\nabla_j \nabla_iY)^u\,\nabla_u X +
Y^u{}_{;i}\,\nabla_j \nabla_u X \big) (\psi) \cr
&\quad  +  \lambda m^{-1}\,g^{ij}\, ( (Y^b \, X^c+ X^b \, Y^c ) R_{ecbi})_{;j} \,g^{ae}\,\rho( \partial_a)(\psi)  \cr
  &\quad - \lambda\,m^{-1}\, g^{ij}\, Y^p R_{pi}  \, \rho(\nabla_j X)(\psi) 
 - \lambda\,m^{-1}\, g^{ij}\, X^q R_{qi}    \, \rho( \nabla_j Y)(\psi) \ ,
\end{align*}
and then
\begin{align*} 
 & m\, \lambda^{-2}\rho(P_0(X,Y)^* +N(X,Y) -P_0(X,Y) )(\psi)\cr
     &= - \tfrac12\, g^{ab}\,\Gamma^i{}_{ab}   \big(
Y^pX^q R_{qp;i} + X^q{}_{;w} Y^p R^w{}_{piq} + X^q Y^p{}_{;w} R^w{}_{qip} \big) \,  \psi \cr
& \quad + g^{ai}\, \big(
Y^pX^q R_{qp;i} + X^q{}_{;w} Y^p R^w{}_{piq} + X^q Y^p{}_{;w} R^w{}_{qip} \big) \,  \psi_{,a}
  \cr
& \quad + g^{ij}\, \big( (\nabla_j \nabla_iX)^u\,Y^a{}_{;u} +
X^u{}_{;i}\,(\nabla_j \nabla_u Y)^a + 
(\nabla_j \nabla_iY)^u\,X^a{}_{;u} +
Y^u{}_{;i}\,(\nabla_j \nabla_u X)^a \big) \, \psi_{,a} \cr
&\quad  +  g^{ij}\, ( (Y^b \, X^c+ X^b \, Y^c ) R_{ecbi})_{;j} \,g^{ae}\,\psi_{,a} \cr
  &\quad - g^{ij}\, Y^p R_{pi}  \, X^a{}_{;j}\,\psi_{,a}
 - g^{ij}\, X^q R_{qi}    \,Y^a{}_{;j}\,\psi_{,a}\ .
\end{align*}
Hence,
\begin{align*}
2m\,& \lambda^{-2}\rho(P(X,Y)-\tfrac12 P_0(X,Y) -\tfrac12 P_0(X,Y)^* - \tfrac12\, N(X,Y)     )\psi 
\cr
 &=   \tfrac12\, g^{ab}\,\Gamma^i{}_{ab}   \big(
Y^pX^q R_{qp;i} + X^q{}_{;w} Y^p R^w{}_{piq} + X^q Y^p{}_{;w} R^w{}_{qip} \big) \,  \psi  \cr
&\quad+   g^{ij}   \, ( X^b{}_{;i;j}\, Y^a{}_{;b} + Y^b{}_{;i;j}\, X^a{}_{;b}   )\,\psi_{,a} 
+ g^{ij}\,  \big(   X^b{}_{;i}\,  Y^a{}_{;j;b}    +     Y^b{}_{;i}\, X^a{}_{;j;b} \big) \,\psi_{,a}   \cr
&\quad+  g^{ij}  \, ( X^b{}_{;k}\, Y^a{}_{;b} + Y^b{}_{;k}\, X^a{}_{;b} )\,\psi_{,a} \,\Gamma^k{}_{ij}   
  +g^{ij}\,  \big(   X^b{}_{;i}\,  Y^a{}_{;k}     +     Y^b{}_{;i}\,X^a{}_{;k} \big) \,\psi_{,a} 
 \,\Gamma^k_{bj}  \cr
& \quad - g^{ai}\, \big(  X^q{}_{;w} Y^p R^w{}_{piq} + X^q Y^p{}_{;w} R^w{}_{qip} \big) \,  \psi_{,a}
  \cr
& \quad - g^{ij}\, \big( (\nabla_j \nabla_iX)^u\,Y^a{}_{;u} +
X^u{}_{;i}\,(\nabla_j \nabla_u Y)^a + 
(\nabla_j \nabla_iY)^u\,X^a{}_{;u} +
Y^u{}_{;i}\,(\nabla_j \nabla_u X)^a \big) \, \psi_{,a} \cr
&\quad  -  g^{ij}\, (Y^b \, X^c+ X^b \, Y^c )_{;j}  R_{ecbi}\,g^{ae}\,\psi_{,a} \cr
 &=   \tfrac12\, g^{ab}\,\Gamma^i{}_{ab}   \big(
Y^pX^q R_{qp;i} + X^q{}_{;w} Y^p R^w{}_{piq} + X^q Y^p{}_{;w} R^w{}_{qip} \big) \,  \psi  \cr
&\quad +   g^{ij}   \, ( (X^b{}_{;i;j} -  (\nabla_j \nabla_iX)^b     )\, Y^a{}_{;b} + (Y^b{}_{;i;j}  - 
  (\nabla_j \nabla_iY)^b      )\, X^a{}_{;b}   )\,\psi_{,a}  \cr
&\quad + g^{ij}\,  \big(   X^b{}_{;i}\,  (Y^a{}_{;j;b}  -   (\nabla_j \nabla_b Y)^a   )   +   
  Y^b{}_{;i}\, (X^a{}_{;j;b} -  (\nabla_j \nabla_b X)^a    ) \big) \,\psi_{,a}   \cr
&\quad+  g^{ij}  \, ( X^b{}_{;k}\, Y^a{}_{;b} + Y^b{}_{;k}\, X^a{}_{;b} )\,\psi_{,a} \,\Gamma^k{}_{ij}   
  +g^{ij}\,  \big(   X^b{}_{;i}\,  Y^a{}_{;k}     +     Y^b{}_{;i}\,X^a{}_{;k} \big) \,\psi_{,a} 
 \,\Gamma^k_{bj}  \cr
& \quad - g^{ai}\, \big(  X^q{}_{;w} Y^p R^w{}_{piq} + X^q Y^p{}_{;w} R^w{}_{qip} \big) \,  \psi_{,a}
  \cr
&\quad  -  g^{ij}\, (Y^b \, X^c+ X^b \, Y^c )_{;j}  R_{ecbi}\,g^{ae}\,\psi_{,a} \cr
 &=   \tfrac12\, g^{ab}\,\Gamma^i{}_{ab}   \big(
Y^pX^q R_{qp;i} + X^q{}_{;w} Y^p R^w{}_{piq} + X^q Y^p{}_{;w} R^w{}_{qip} \big) \,  \psi  \cr
&\quad +   g^{ij}   \, ( (X^b{}_{;i;j} -  (\nabla_j \nabla_iX)^b     )\, Y^a{}_{;b} + (Y^b{}_{;i;j}  - 
  (\nabla_j \nabla_iY)^b      )\, X^a{}_{;b}   )\,\psi_{,a}  \cr
&\quad + g^{ij}\,  \big(   X^b{}_{;i}\,  (Y^a{}_{;j;b}  -   (\nabla_b \nabla_j Y)^a   )   +   
  Y^b{}_{;i}\, (X^a{}_{;j;b} -  (\nabla_b \nabla_j X)^a    ) \big) \,\psi_{,a}   \cr
&\quad+  g^{ij}  \, ( X^b{}_{;k}\, Y^a{}_{;b} + Y^b{}_{;k}\, X^a{}_{;b} )\,\psi_{,a} \,\Gamma^k{}_{ij}   
  +g^{ij}\,  \big(   X^b{}_{;i}\,  Y^a{}_{;k}     +     Y^b{}_{;i}\,X^a{}_{;k} \big) \,\psi_{,a} 
 \,\Gamma^k_{bj}  \cr
& \quad - g^{ae}\,g^{ij} \big(  X^q{}_{;i} Y^p R_{jpeq} + X^q Y^p{}_{;i} R_{jqep} \big) \,  \psi_{,a}
- g^{ij} \big( X^b{}_{;i} R^a{}_{pjb} Y^p + Y^b{}_{;i} R^a{}_{pjb} X^p
\big) \,  \psi_{,a}
  \cr
&\quad  -  g^{ij}\, (Y^b \, X^c+ X^b \, Y^c )_{;i}  R_{ecbj}\,g^{ae}\,\psi_{,a} \cr
 &=   \tfrac12\, g^{ab}\,\Gamma^i{}_{ab}   \big(
Y^pX^q R_{qp;i} + X^q{}_{;w} Y^p R^w{}_{piq} + X^q Y^p{}_{;w} R^w{}_{qip} \big) \,  \psi  \cr
&\quad +   g^{ij}   \, (    - X^b{}_{;k} \Gamma^k{}_{ ij }   Y^a{}_{;b} - Y^b{}_{;k}\Gamma^k{}_{ij} X^a{}_{;b}   )\,\psi_{,a}  
 + g^{ij}\,  \big(  - X^b{}_{;i}\,  Y^a{}_{;k} \Gamma^k{}_{jb}     -   
  Y^b{}_{;i}\,  X^a{}_{;k}  \Gamma^k{}_{jb}  \big) \,\psi_{,a}   \cr
&\quad+  g^{ij}  \, ( X^b{}_{;k}\, Y^a{}_{;b} + Y^b{}_{;k}\, X^a{}_{;b} )\,\psi_{,a} \,\Gamma^k{}_{ij}   
  +g^{ij}\,  \big(   X^b{}_{;i}\,  Y^a{}_{;k}     +     Y^b{}_{;i}\,X^a{}_{;k} \big) \,\psi_{,a} 
 \,\Gamma^k_{bj} 
\end{align*}
and the Christoffel symbols in the last two lines cancel. We see that the result is given by the action of an algebra element as stated of order $\lambda^2$ and which therefore vanishes at order $\lambda$ as in (\ref{P0N}). \end{proof}

\subsection{Differential calculus to order $\lambda$}\label{sec:prod}

Finally, while we have a set of commutation relations for $\Omega^1_{\CD(M)}$ to order $\lambda^2$, we should complete its specification. The model for how this data extends to a product was explained below (\ref{firw})  in the case of functions and differentials of functions. For the other products one can proceed case by case,  but here we note that since our proposed relations are invariantly defined, it is sufficient to define the products in a local coordinate chart. In this case,  we take  $\extd x^\mu, \extd p_\nu, \theta'$ as locally a basis over the algebra from the left, and define products of these from the right via the commutation relations. More generally,  we define  $(a\extd x^\mu).b:= a(\extd x^\mu b)$ and $(a\extd p_\nu).b:= a(\extd p_\nu b)$ for $a,b$ in $\CD(M)$. There are no issues for $\theta'$ as our construction has this central, but for this procedure to specify a right action, we need $[\extd x^\mu, \ ]$ and $[\extd p_\nu,\ ]$ to be derivations. Our relations tell us what these are on functions and vector fields but that they extend as derivations is not automatic. Indeed, requiring this on  $[a,b]=ab-ba$ is just the Jacobi identity. We will see in the next section that this holds to order $\lambda$ (but not necessarily at order $\lambda^2$), hence we have a bimodule at least to order $\lambda$. That $\extd$ is a derivation was part of our construction to find the relations, but can also be verified explicitly to order $\lambda$. 

\section{Jacobiators}\label{sec:jac}

We define the Jacobiator
\begin{align} \label{bvy4}
J(x,y,z) &:= [  x  ,[    y  , z  ] ] +  [ z   ,[   x   ,  y ] ] +  [  y  ,[  z    , x  ] ] 
\end{align}
for elements $x,y,z$ elements of the algebra or its 1-forms.
Note that applying a permutation to $x,y,z$ simply multiplies the Jacobiator by the sign of the permutation.
 If we have associativity then all the Jacobiators will vanish.

\begin{proposition}   \label{bhd}
For all functions $f,h\in C^\infty(M)$,  1-forms $\xi\in\Omega^1(M)$ and vector fields $X,Y$, to order $\lambda^2$ we have
\begin{align*}
&J(f,h,\widehat\xi)=0\ ,\quad J(f,Y,\widehat\xi)= 0  \cr
& J(Y,X,\widehat\xi)  =   \lambda^2 \,   Y^a \, X^c\,  \xi_i\,  R^i{}_{jca}\,\big(   \widehat{\extd x^j} -  m^{-1}\theta' \, g^{ej}\, \partial_e  \big). 
\end{align*}
\end{proposition}
\begin{proof} The first calculation is omitted as easier, and known since we have a (symmetric version of) a standard centrally extended calculus on a manifold. For the second result,
\begin{align*}
&J(f,Y,\widehat\xi) 
= [  f  ,[    Y  , \widehat\xi  ] ] -\lambda\,  [\widehat \xi   ,Y(f) ] +  [  Y  ,[ \widehat \xi    , f  ] ]   \cr
&= [  f  , \lambda\,(\widehat{ \nabla_Y \xi  }) ] + (2m)^{-1} [  f  , - 2\lambda\,\theta'\, (g^{ij}\,\xi_i\, \nabla_j Y)  ]  -\lambda^2m^{-1}\,  g^{ij}\, \xi_j   (Y(f))_{, i} \, \theta' + \lambda m^{-1}\,  [  Y  , g^{ij}\, \xi_j\, f_{,i}\,\theta'    ]   \cr
&= -\lambda^2m^{-1} g^{ij}\, f_{,i}\, (\nabla_Y \xi  )_j \theta' +\lambda^2\,  m^{-1} \theta'\, g^{ij}\,\xi_i\, (\nabla_j Y)^a\,f_{,a}  \cr & \quad -\lambda^2m^{-1}\,  g^{ij}\, \xi_j   (Y(f))_{, i} \, \theta' + \lambda^2m^{-1}\,   Y^a\,  g^{ij}\, (\xi_j\, f_{,i})_{,a}\,\theta'     \cr
&=\, \lambda^2\,  m^{-1} \theta'\, g^{ij}\big(
 - f_{,i}\, (\nabla_Y \xi  )_j +  \xi_i\, (\nabla_j Y)^a\,f_{,a}  -  \xi_j   (Y(f))_{, i}  +  Y^a\,  (\xi_j\, f_{,i})_{,a}
\big)
\end{align*}
which vanishes. For the third result, by definition,
\begin{align*}
J(X,Y,\widehat\xi) 
&=[  X  ,[    Y  , \widehat\xi  ] ]  - [  Y  ,[    X  ,\widehat \xi  ] ] + \lambda [\widehat \xi   ,[ X  ,  Y ]_{\rm Lie}] \ .
\end{align*}
We begin with
\begin{align*}
 m\, &[  Y  ,[    X  , \xi  ] ] = \big[ Y,  m\, \lambda\,(\widehat{ \nabla_X \xi  })  - \lambda\,\theta'\, (g^{ij}\,\xi_i\, \nabla_j X) 
-{\lambda^2\over 2} \theta'\big(X^a\, \xi_p\, g^{pq}\, R_{qa} + g^{ij}\, X^a{}_{;j}\, \xi_{i;a}\big)   \cr
&= m\, \lambda\, \big[ Y,  (\widehat{ \nabla_X \xi  }) \big] - \lambda\,\theta'\, \big[ Y, g^{ij}\,\xi_i \big] \, \nabla_j X 
- \lambda\,\theta'\, g^{ij}\,\xi_i\,\big[ Y,  \nabla_j X  \big] \cr
&= m\, \lambda^2\,(\widehat{ \nabla_Y  \nabla_X \xi   })  - \lambda^2\,\theta'\, g^{ij}\, (\nabla_X \xi )_i\, \nabla_j Y
- \lambda^2\,\theta'\, Y^a\, (g^{ij}\,\xi_i )_{,a} \, \nabla_j X 
- \lambda^2\,\theta'\, g^{ij}\,\xi_i\, [ Y,  \nabla_j X]_{\rm Lie} \cr
&= m\, \lambda^2\,(\widehat{ \nabla_Y  \nabla_X \xi   })  - \lambda^2\,\theta'\, g^{ij}\,X^a\,\xi_{i;a}\, Y^b{}_{;j}\,\partial_b
- \lambda^2\,\theta'\, Y^a\, g^{ij}\,\xi_{i;a}  \, \nabla_j X \cr
&\quad + \lambda^2\,\theta'\, Y^a\, g^{ik}\,\xi_i \,\Gamma^j_{ka} \, \nabla_j X 
- \lambda^2\,\theta'\, g^{ij}\,\xi_i\, [ Y,  \nabla_j X ]_{\rm Lie} \cr
&= m\, \lambda^2\,(\widehat{ \nabla_Y  \nabla_X \xi   })  - \lambda^2\,\theta'\, g^{ij}\,\xi_{i;a}\, (X^a\,\nabla_j Y
+Y^a\,  \nabla_j X )\cr
&\quad + \lambda^2\,\theta'\, Y^a\, g^{ik}\,\xi_i \,\Gamma^j_{ka} \, \nabla_j X 
- \lambda^2\,\theta'\, g^{ij}\,\xi_i\, Y^a\,\nabla_a  \nabla_j X   +  \lambda^2\,\theta'\, g^{ij}\,\xi_i\,X^b{}_{;j}  \nabla_b Y
\cr
&= m\, \lambda^2\,(\widehat{ \nabla_Y  \nabla_X \xi   })  - \lambda^2\,\theta'\, g^{ij}\,\xi_{i;a}\, (X^a\,\nabla_j Y
+Y^a\,  \nabla_j X )\cr
&\quad + \lambda^2\,\theta' \, g^{ij}\,\xi_i \,\big(Y^a\, \Gamma^c_{ja} \, \nabla_c X 
-   Y^a\,\nabla_a  \nabla_j X   +  X^b{}_{;j}  \nabla_b Y\big)
\cr
&= m\, \lambda^2\,(\widehat{ \nabla_Y  \nabla_X \xi   })  - \lambda^2\,\theta'\, g^{ij}\,\xi_{i;a}\, (X^a\,\nabla_j Y
+Y^a\,  \nabla_j X )\cr
&\quad + \lambda^2\,\theta' \, g^{ij}\,\xi_i \,\big(
-   Y^a\,  (X^e{}_{;j}  )_{;a}  +  X^b{}_{;j}  \, Y^e{}_{;b}\big)\partial_e
\cr
&= m\, \lambda^2\,(\widehat{ \nabla_Y  \nabla_X \xi   })  - \lambda^2\,\theta'\, g^{ij}\,\xi_{i;a}\, (X^a\,\nabla_j Y
+Y^a\,  \nabla_j X )\cr
&\quad + \lambda^2\,\theta' \, g^{ij}\,\xi_i \,\big(
-   Y^a\,  (X^e{}_{;a}  )_{;j}  +  X^b{}_{;j}  \, Y^e{}_{;b}\big)\partial_e
- \lambda^2\,\theta' \, g^{ij}\,\xi_i \,\big(
 Y^a\, R^e{}_{caj}\, X^c \big)\partial_e
\end{align*}
so
\begin{align*}
m\, [  Y  ,[    X  , \xi  ] ]  &-  m\, [  X  ,[    Y  , \xi  ] ] - \lambda\,  m\, [\widehat \xi   ,[  X  ,  Y ]_{\rm Lie}] \cr
=&\,  m\, \lambda^2\,(\widehat{ \nabla_Y  \nabla_X \xi   })  -  m\, \lambda^2\,(\widehat{ \nabla_X  \nabla_Y \xi   }) 
\cr
&\quad + \lambda^2\,\theta' \, g^{ij}\,\xi_i \,\big(
-   Y^a\,  (X^e{}_{;a}  )_{;j}  +  X^b{}_{;j}  \, Y^e{}_{;b}\big)\partial_e
- \lambda^2\,\theta' \, g^{ij}\,\xi_i \,
 Y^a\,  X^c \, R^e{}_{caj}\,\partial_e \cr
 &\quad - \lambda^2\,\theta' \, g^{ij}\,\xi_i \,\big(
-   X^a\,  (Y^e{}_{;a}  )_{;j}  +  Y^b{}_{;j}  \, X^e{}_{;b}\big)\partial_e
+ \lambda^2\,\theta' \, g^{ij}\,\xi_i \,
 X^a\,  Y^c \, R^e{}_{caj}\,\partial_e \cr
 &\quad - \lambda  m\, [  [Y  ,  X ]_{\rm Lie},\widehat \xi ] \cr
 =&\,  m\, \lambda^2\,(\widehat{ \nabla_Y  \nabla_X \xi   })  -  m\, \lambda^2\,(\widehat{ \nabla_X  \nabla_Y \xi   }) 
 -  m\, \lambda^2\,(\widehat{ \nabla_{[Y,X]_{\rm Lie}} \xi   }) 
\cr
&\quad + \lambda^2\,\theta' \, g^{ij}\,\xi_i \,\big(
-   Y^a\,  (X^e{}_{;a}  )_{;j}  +  X^b{}_{;j}  \, Y^e{}_{;b}\big)\partial_e
- \lambda^2\,\theta' \, g^{ij}\,\xi_i \,
 Y^a\,  X^c \, R^e{}_{caj}\,\partial_e \cr
 &\quad - \lambda^2\,\theta' \, g^{ij}\,\xi_i \,\big(
-   X^a\,  (Y^e{}_{;a}  )_{;j}  +  Y^b{}_{;j}  \, X^e{}_{;b}\big)\partial_e
- \lambda^2\,\theta' \, g^{ij}\,\xi_i \,
 X^c\,  Y^a \, R^e{}_{ajc}\,\partial_e \cr
 &\quad + \lambda^2\,\theta'\, g^{ij}\,\xi_i\, \nabla_j [Y,X]_{\rm Lie}  \cr
  =&\,   \lambda^2\,\theta' \, g^{ij}\,\xi_i \,
 X^c\,  Y^a \, R^e{}_{jca}\,\partial_e -  m\, \lambda^2\, Y^i\, X^j\,R^b{}_{aij}\, \xi_b\, \widehat{\extd x^a}     \cr
   =&\,   \lambda^2\,\theta' \, g^{ij}\,\xi_i \,
 X^c\,  Y^a \, R^e{}_{jca}\,\partial_e -  m\, \lambda^2\, Y^a \, X^c\,R^b{}_{jac}\, \xi_b\, \widehat{\extd x^j}     \cr
    =&\,  \lambda^2  m\,   Y^a \, X^c\,   \big( m^{-1}\theta' \, g^{ij}\,\xi_i \,g^{ep}\, R_{pjca}\,\partial_e 
    - R^i{}_{jac}\, \xi_i\, \widehat{\extd x^j}  \big)   \cr
        =&\,  \lambda^2  m\,   Y^a \, X^c\,  \xi_i\,  \big( -  m^{-1}\theta' \, g^{ij} \,g^{ep}\, R_{jpca}\,\partial_e 
    + R^i{}_{jca}\, \widehat{\extd x^j}  \big)   \cr
            =&\,  \lambda^2  m\,   Y^a \, X^c\,  \xi_i\,  \big( -  m^{-1}\theta' \, g^{ep}\, R^i{}_{pca}\,\partial_e 
    + R^i{}_{jca}\, \widehat{\extd x^j}  \big)   \cr
 =&\,  \lambda^2  m\,   Y^a \, X^c\,  \xi_i\,  \big(  R^i{}_{jca}\, \widehat{\extd x^j} -  m^{-1}\theta' \, g^{ej}\, R^i{}_{jca}\,\partial_e  \big)   \cr
 =&\,  \lambda^2  m\,   Y^a \, X^c\,  \xi_i\,  R^i{}_{jca}\,\big(   \widehat{\extd x^j} -  m^{-1}\theta' \, g^{ej}\, \partial_e  \big)  
\end{align*}    
giving the answer. \end{proof}

Hence the calculus is not associative at order $\lambda^2$. Note that we assumed associativity in deriving (\ref{hjs}), however this only required the vanishing of the Jacobi relation for two functions and a vector field, which we see does hold.

\begin{proposition} \label{smallj}
We have $J(f,h,\extd X)=0$ and 
\begin{align*}
&J(f,Y,\extd X) 
=   \lambda^2 \, \extd x^i\,Y^b X^c R^a{}_{cbi} f_{,a} 
  -  \lambda^2m^{-1}g^{ae}\, \theta'\,Y^b \, X^c R^j{}_{bce} f_{,a}\partial_j \cr
  &\quad =   \lambda^2 \,Y^b X^c\, f_{,a}  \big(R^a{}_{cbi} \, \extd x^i
  -  m^{-1}g^{ae}\, \theta'\,g^{ij}\, R_{ibce} \partial_j \big) \cr
    &\quad =   \lambda^2 \,Y^b X^c\, f_{,a} R^a{}_{cbi} \,  \big(\extd x^i
  -  m^{-1}\, \theta'\,g^{ij}\,\partial_j \big)
\end{align*}
to order $\lambda^2$.
\end{proposition}
\begin{proof} Begin with
\begin{align} \label{bvy}
J(f,h,\extd X) &= [  f  ,[    h  , \extd X  ] ] +  [ \extd X   ,[   f   ,  h ] ] +  [  h  ,[  \extd X    , f  ] ] \cr
&=   [  f  ,[    h  , \extd X  ] ]   - [  h  ,[    f  , \extd X  ] ]   \ ,\cr
J(f,Y,\extd X) &=[  f  ,[    Y  , \extd X  ] ] +  [ \extd X   ,[   f   ,  Y ] ] +  [  Y  ,[  \extd X    , f  ] ] \cr
&= [  f  ,[    Y  , \extd X  ] ] -\lambda\,  [ \extd X   ,Y(\extd f) ] +  [  Y  ,[  \extd X    , f  ] ] \ .
\end{align}
We only need the commutators to first order in $\lambda$ for this, so set
\begin{align*}
[\extd X,f ] &=  \lambda\,(\widehat{ X^a{}_{;i } \, f_{,a}})  + \lambda m^{-1}\,\theta'\, (g^{ij}\, f_{,i}\, \nabla_j X) 
\big)\ ,\cr
[[\extd X,f ] ,h]&=  \lambda\,[(\widehat{ X^a{}_{;i } \, f_{,a}  }) ,h] + \lambda m^{-1}\,\theta'\, [(g^{ij}\, f_{,i}\, \nabla_j X) ,h]
\end{align*}
then from (\ref{firw}),
\begin{align*}
&[[\extd X,f ] ,h]= \lambda^2m^{-1}\,g^{ij}\, X^a{}_{;j } \, f_{,a}  \,\theta'  h_{,i} 
+ \lambda^2 m^{-1}\,\theta'\, g^{ij}\, f_{,i}\, X^a{}_{;j} \, h_{,a}
\end{align*}
and this is symmetric in $f,h$ so $J(f,h,\extd X)=0$. Next
\begin{align*}
&J(f,Y,\extd X) 
= [  f  , \lambda\, \extd(\nabla_Y X)  ]     +  [  f  ,  \lambda\, P_0(X,Y)   ]
-\lambda\,  [ \extd X   ,Y(\extd f) ] +  [  Y  ,[  \extd X    , f  ] ] \cr
&= -  \lambda^2\,(\nabla_Y X)^a{}_{;i } \, f_{,a} \,\extd x^i - \lambda^2m^{-1}\,\theta'\, g^{ij}\, f_{,i}\, \nabla_j \nabla_Y X\cr
&\quad  - \lambda\big[f, \widehat{\extd x^i}\, \big(\nabla_{\nabla_i X} Y + \nabla_{\nabla_i Y} X\big)\big]
 - \lambda (2m) ^{-1} g^{ij}\,\theta' \big[ f, \big(\nabla_i X\, \nabla_j Y+
\nabla_i Y\, \nabla_j X\big)\big] \cr
&\quad  +  \lambda m^{-1}g^{ij}\, \theta'\,Y^b \, X^c R_{ecbi} \,g^{ae} [f,\partial_j \partial_a ] \cr
&\quad -  \lambda^2\,X^a{}_{;i } \, (Y(\extd f))_{,a}\,\extd x^i - \lambda^2m^{-1}\,\theta'\, g^{ij}\, (Y(\extd f))_{,i}\, \nabla_j X  \cr
&\quad + \lambda \big[Y,  X^a{}_{;i } \, f_{,a}\, \extd x^i + m^{-1}\,\theta'\, g^{ij}\, f_{,i}\, \nabla_j X \big]   \cr
&= -  \lambda^2\,(\nabla_Y X)^a{}_{;i } \, f_{,a} \,\extd x^i - \lambda^2m^{-1}\,\theta'\, g^{ij}\, f_{,i}\, \nabla_j \nabla_Y X\cr
&\quad  +    \lambda^2m^{-1}\,g^{ij} \,f_{,j}\,\theta'    (\nabla_{\nabla_i X} Y + \nabla_{\nabla_i Y} X)
 + \lambda^2 \widehat{\extd x^i}\, (\nabla_{\nabla_i X} Y + \nabla_{\nabla_i Y} X)^a\, f_{,a} \cr
&\quad +  \lambda^2 m^{-1} g^{ij}\,\theta'  f_{,a}(X^a{}_{;i}\, \nabla_j Y + Y^a{}_{;i}\, \nabla_j X) 
  +  \lambda m^{-1}g^{ij}\, \theta'\,Y^b \, X^c R_{ecbi} \,g^{ae} [f,\partial_j \partial_a ] \cr
&\quad -  \lambda^2\,X^a{}_{;i } \, (Y(\extd f))_{,a}\, \extd x^i - \lambda^2m^{-1}\,\theta'\, g^{ij}\, (Y(\extd f))_{,i}\, \nabla_j X  \cr
&\quad + \lambda^2 \nabla_Y(X^a{}_{;i } \, f_{,a}\, \extd x^i )
- \lambda^2m^{-1}\,\theta' \, g^{ij}\,X^a{}_{;i } \, f_{,a}\, \nabla_j Y \\
&\quad + \lambda^2m^{-1}\,\theta'\, Y^a(g^{ij}\, f_{,i})_{,a}\, \nabla_j X 
+ \lambda^2m^{-1}\,\theta'\, g^{ij}\, f_{,i}\, [Y,\nabla_j X]_{\rm Lie} \cr
&= - \lambda^2m^{-1}\,\theta'\, g^{ij}\, f_{,i}\, [\nabla_j ,\nabla_Y] X
  +    \lambda^2m^{-1}\,g^{ij} \,f_{,j}\,\theta'    \, \nabla_{\nabla_i Y} X
 + \lambda^2 \, \extd x^i\, (  X^b{}_{;i}Y^a{}_{;b} - Y^b X^a{}_{;b;i}  )\, f_{,a} \cr
&\quad + \lambda^2 m^{-1} g^{ij}\,\theta'  f_{,a}(X^a{}_{;i}\, \nabla_j Y + Y^a{}_{;i}\, \nabla_j X) 
  +  \lambda m^{-1}g^{ij}\, \theta'\,Y^b \, X^c R_{ecbi} \,g^{ae} [f,\partial_j \partial_a ] \cr
&\quad -  \lambda^2\,X^b{}_{;i } \, (  Y^a{}_{;b}f_{,a} + Y^a f_{,a;b}  )  \, \extd x^i - \lambda^2m^{-1}\,\theta'\, g^{ij}\, (  Y^a{}_{;i}f_{,a} + Y^a f_{,a;i}  ) \, \nabla_j X  \cr
&\quad + \lambda^2 \nabla_Y(X^a{}_{;i } \, f_{,a}\, \extd x^i )
- \lambda^2m^{-1}\,\theta' \, g^{ij}\,X^a{}_{;i } \, f_{,a}\, \nabla_j Y  + \lambda^2m^{-1}\,\theta'\, Y^a(g^{ij}\, f_{,i})_{,a}\, \nabla_j X 
\cr
&= - \lambda^2m^{-1}\,\theta'\, g^{ij}\, f_{,i}\, [\nabla_j ,\nabla_Y] X
  +    \lambda^2m^{-1}\,g^{ij} \,f_{,j}\,\theta'    \, \nabla_{\nabla_i Y} X
 + \lambda^2 \, \extd x^i\, (   - Y^b X^a{}_{;b;i}  )\, f_{,a} \cr
&\quad
  +  \lambda m^{-1}g^{ij}\, \theta'\,Y^b \, X^c R_{ecbi} \,g^{ae} [f,\partial_j \partial_a ] 
   -  \lambda^2\,X^b{}_{;i } \, (   Y^a f_{,a;b}  )  \, \extd x^i  \cr
&\quad + \lambda^2 \nabla_Y(X^a{}_{;i } \, f_{,a}\, \extd x^i )
  - \lambda^2m^{-1}\,\theta'\, Y^a\, g^{ip}\,\Gamma^j{}_{ap} f_{,i}\, \nabla_j X 
\cr
&=  \lambda^2m^{-1}\,\theta'\, g^{ij}\, f_{,i}\, [\nabla_Y,\nabla_j] X
  +    \lambda^2m^{-1}\,g^{ij} \,f_{,j}\,\theta'    \, \nabla_{\nabla_i Y} X
 + \lambda^2 \, \extd x^i\,Y^b  (  X^a{}_{;i;b} - X^a{}_{;b;i}  )\, f_{,a} \cr
&\quad
  +  \lambda m^{-1}g^{ij}\, \theta'\,Y^b \, X^c R_{ecbi} \,g^{ae} [f,\partial_j \partial_a ] 
  - \lambda^2m^{-1}\,\theta'\, Y^a\, g^{ip}\,\Gamma^j{}_{ap} f_{,i}\, \nabla_j X 
\cr
&= \lambda^2m^{-1}\,\theta'\, g^{ij}\, f_{,i}\, R(Y,\del_j)X
 + \lambda^2 \, \extd x^i\,Y^b X^c R^a{}_{cbi} f_{,a} 
  +  \lambda m^{-1}g^{ij}\, \theta'\,Y^b \, X^c R_{ecbi} \,g^{ae} [f,\partial_j \partial_a ] 
\cr
&=   \lambda^2 \, \extd x^i\,Y^b X^c R^a{}_{cbi} f_{,a} 
  -  \lambda^2m^{-1}g^{ae}\, \theta'\,Y^b \, X^c R^j{}_{bce} f_{,a}\partial_j
\ ,
\end{align*}
as required. \end{proof}

\begin{proposition} \label{bigj}
We have
\begin{align*}
&J(X,Y,\extd Z) 
            = \lambda^2\, \extd(R(X,Y) Z)   - \lambda^2\,\extd x^i\,  \nabla_i(R(X,Y) Z)  +\lambda^2\,\theta'\, (R(X,Y) Z)^k\, V_{,k} \cr
&\quad + \lambda^2\,\big(
R^b{}_{rqp;i} X^q Y^pZ^r \,\del_b
-  R^a{}_{rqi} X^q  Z^r\nabla_a Y
 - R^a{}_{rip}  Y^p Z^r\, \nabla_a X
- R^a{}_{iqp} X^q Y^p\, \nabla_a Z \big) (\extd x^i  -  m^{-1} g^{ij}\,\theta'    \del_j )
\end{align*}
to order $\lambda^2$.
\end{proposition}
\begin{proof} 
Begin with
\begin{align*}
J(X,Y,\extd Z) 
&= [  X  ,[    Y  , \extd Z  ] ]  -  [  Y  ,[    X  , \extd Z  ] ] -\lambda\,  [   [  X  ,  Y]_{\rm Lie}  , \extd Z ]   \cr
&=  \lambda\,  [  X  , \extd(\nabla_Y Z)  ]  +  \lambda\,  [  X  ,  P(Z,Y)   ] 
 - \lambda\, [  Y  ,   \extd(\nabla_X Z) ]       -  \lambda\, [  Y  ,   P(Z,X) ] \cr
&\quad   -\lambda^2 \, \extd(\nabla_{  [   X  ,  Y ]_{\rm Lie}  } Z)    -\lambda^2\, P(Z, [  X  ,  Y ]_{\rm Lie} ) \cr
&= \lambda^2\, \extd(\nabla_X \nabla_Y Z) +  \lambda^2\, P(\nabla_Y Z,X) +  \lambda\,  [  X  ,  P(Z,Y)   ]  \cr
&\quad -\lambda^2\, \extd(\nabla_Y \nabla_X Z) -  \lambda^2\, P(\nabla_X Z,Y)     -  \lambda\, [  Y  ,   P(Z,X) ] \cr
&\quad   -\lambda^2 \, \extd(\nabla_{  [ X  ,  Y ]_{\rm Lie}  } Z)    -\lambda^2\, P(Z, [  X  ,  Y]_{\rm Lie} ) \cr
&= \lambda^2\, \extd(R(X,Y) Z) +  \lambda^2\, P(\nabla_Y Z,X) +  \lambda\,  [  X  ,  P(Z,Y)   ]  \cr
&\quad  -  \lambda^2\, P(\nabla_X Z,Y)     -  \lambda\, [  Y  ,   P(Z,X) ]    -\lambda^2\, P(Z, [ X  ,  Y ]_{\rm Lie} ) 
\end{align*}
which gives 
\begin{align} \label{bvy66}
&J(X,Y,\extd Z) 
= \lambda^2\, \extd(R(X,Y) Z) \cr
&\quad  +  \lambda\,  [  X  ,  P_0(Z,Y)   ]  -  \lambda^2\, P_0(\nabla_X Z,Y)  -\lambda^2\, P_0(Z, \nabla_XY )  \cr
&\quad    -  \lambda\, [  Y  ,   P_0(Z,X) ]  +  \lambda^2\, P_0(\nabla_Y Z,X)    + \lambda^2\, P_0(Z, \nabla_Y X   ) \ .
\end{align}

Now we use, to order $\lambda$,
 \begin{align} \label{osos}
[X, & P_0( Z,Y)] = - \, [X,\widehat{\extd x^i}]\, \big(\nabla_{\nabla_i  Z} Y + \nabla_{\nabla_i Y}  Z\big)
-  \widehat{\extd x^i}\, [X,\big(\nabla_{\nabla_i  Z} Y + \nabla_{\nabla_i Y}  Z\big)]  \cr
&\quad - (2m) ^{-1} \theta' [X,g^{ij}\,\big(\nabla_i  Z\, \nabla_j Y+
\nabla_i Y\, \nabla_j  Z\big)] \cr
&\quad - \theta'\, [X,Y^b\,  Z^a\, V_{,a;b}]    + m^{-1}\theta'\,[X,g^{ij}\, Y^b \,  Z^c R_{ecbi} \,g^{ae} (\partial_j \partial_a ) ] \cr
&=  \lambda(X^a\Gamma^i{}_{ak}\,\extd x^k +   m^{-1}\,\theta'\, g^{ij}\,\nabla_j X    )\,  \big(\nabla_{\nabla_i  Z} Y + \nabla_{\nabla_i Y}  Z\big)  \cr
&\quad -  \extd x^i\, [X,\big(\nabla_{\nabla_i  Z} Y + \nabla_{\nabla_i Y}  Z\big)]  
 - (2m) ^{-1} \theta' [X,g^{ij}\,\big(\nabla_i  Z\, \nabla_j Y+
\nabla_i Y\, \nabla_j  Z\big)] \cr
&\quad - \theta'\, [X,Y^b\,  Z^a\, V_{,a;b}]    +  m^{-1}\theta'\,[X,g^{ij}\, Y^b \,  Z^c R_{ecbi} \,g^{ae} (\partial_j \partial_a ) ] \ .
\end{align}
Now write to order $\lambda^2$ the terms containing  $\extd x^i$ in second line of (\ref{bvy66}) as
 \begin{align} \label{osoy6}
&  \lambda^2\, X^a\Gamma^s{}_{ai}\,\extd x^i   \,  \big(\nabla_{\nabla_s  Z} Y + \nabla_{\nabla_s Y}  Z\big)  
 - \lambda\,  \extd x^i\, [X,\big(\nabla_{\nabla_i  Z} Y + \nabla_{\nabla_i Y}  Z\big)]  \cr
 &\quad  +  \lambda^2\,\extd x^i\, (\nabla_{\nabla_i  \nabla_X Z} Y + \nabla_{\nabla_i Y}  \nabla_X Z) 
 +  \lambda^2\,\extd x^i\, (\nabla_{\nabla_i  Z} \nabla_X Y + \nabla_{\nabla_i \nabla_X Y}  Z) \cr
 &=  \lambda^2\, X^a\Gamma^s{}_{ai}\,\extd x^i   \,  \big(\nabla_{\nabla_s  Z} Y + \nabla_{\nabla_s Y}  Z\big)  
 + \lambda^2\,  \extd x^i\, \big(\nabla_{\nabla_{\nabla_i  Z} Y} X + \nabla_{\nabla_{\nabla_i Y}  Z} X \big) \cr
 &\quad  +  \lambda^2\,\extd x^i\, (\nabla_{\nabla_i  \nabla_X Z} Y + [\nabla_{\nabla_i Y} , \nabla_X] Z) 
 +  \lambda^2\,\extd x^i\, ([\nabla_{\nabla_i  Z}  , \nabla_X ]Y + \nabla_{\nabla_i \nabla_X Y}  Z)  \cr
  &=  \lambda^2\, X^a\Gamma^s{}_{ai}\,\extd x^i   \,  \big(\nabla_{\nabla_s  Z} Y + \nabla_{\nabla_s Y}  Z\big)  
 + \lambda^2\,  \extd x^i\, \big(\nabla_{\nabla_{\nabla_i  Z} Y} X + \nabla_{\nabla_{\nabla_i Y}  Z} X \big) \cr
 &\quad  +  \lambda^2\,\extd x^i\, (\nabla_{\nabla_i  \nabla_X Z} Y + R(\nabla_i Y,X)Z
 + \nabla_{  [  \nabla_i Y ,X]_{\rm Lie}   } Z   )   \cr
 &\quad +  \lambda^2\,\extd x^i\, ( R(\nabla_i  Z  , X )Y 
 + \nabla_{  [\nabla_i  Z,X]_{\rm Lie}  }Y+ \nabla_{\nabla_i \nabla_X Y}  Z)  \cr
   &=  \lambda^2\, X^a\Gamma^s{}_{ai}\,\extd x^i   \,  \big(\nabla_{\nabla_s  Z} Y + \nabla_{\nabla_s Y}  Z\big)  
 + \lambda^2\,  \extd x^i\, \big(\nabla_{\nabla_{\nabla_i  Z} Y} X + \nabla_{\nabla_{\nabla_i Y}  Z} X \big) \cr
 &\quad  +  \lambda^2\,\extd x^i\, (\nabla_{[\nabla_i , \nabla_X] Z} Y + R(\nabla_i Y,X)Z
 + \nabla_{   \nabla_{ \nabla_i Y} X   } Z   )   \cr
 &\quad +  \lambda^2\,\extd x^i\, ( R(\nabla_i  Z  , X )Y 
 + \nabla_{  \nabla_{ \nabla_i  Z}X  }Y+ \nabla_{[\nabla_i, \nabla_X] Y}  Z)  \ .
\end{align}
Now using $[\frac{\del}{\del x^i},X]_{\rm Lie} =\nabla_i X-X^a\Gamma^s{}_{ia}\frac{\del}{\del x^s}$,  this is
 \begin{align} \label{osoy7}
   &=   \lambda^2\,  \extd x^i\, \big(\nabla_{\nabla_{\nabla_i  Z} Y} X + \nabla_{\nabla_{\nabla_i Y}  Z} X \big) \cr
 &\quad  +  \lambda^2\,\extd x^i\, (\nabla_{  R(  \frac{\del}{\del x^i},X) Z} Y  + \nabla_{   \nabla_{\nabla_i X} Z} Y
 + R(\nabla_i Y,X)Z
 + \nabla_{   \nabla_{ \nabla_i Y} X   } Z   )   \cr
 &\quad +  \lambda^2\,\extd x^i\, ( R(\nabla_i  Z  , X )Y 
 + \nabla_{  \nabla_{ \nabla_i  Z}X  }Y+ \nabla_{  \nabla_{\nabla_i X}Y}  Z    + \nabla_{ R(  \frac{\del}{\del x^i},X) Y}  Z)  \ .
\end{align}
so we get the total $\extd x^i$ contribution to the Jacobi operator as
 \begin{align} \label{osoy79}
   &= 
   \lambda^2\,\extd x^i\, (\nabla_{  R(  \frac{\del}{\del x^i},X) Z} Y 
 + R(\nabla_i Y,X)Z 
 +  R(\nabla_i  Z  , X )Y 
 + \nabla_{ R(  \frac{\del}{\del x^i},X) Y}  Z)  \cr
 &\quad -  \lambda^2\,\extd x^i\, (\nabla_{  R(  \frac{\del}{\del x^i},Y) Z} X 
 + R(\nabla_i X,Y)Z 
 +  R(\nabla_i  Z  , Y )X
 + \nabla_{ R(  \frac{\del}{\del x^i},Y) X}  Z) \ .
\end{align}
Now we write to order $\lambda^2$ the terms not containing  $\extd x^i$ and not containing $V$ in second line of (\ref{bvy66}), using
(\ref{osos}) as
 \begin{align*} 
&  \lambda^2m^{-1}\,\theta'\, g^{ij}\,\nabla_j X    \,  \big(\nabla_{\nabla_i  Z} Y + \nabla_{\nabla_i Y}  Z\big)  
 - \lambda(2m) ^{-1} \theta' [X,g^{ij}\,\big(\nabla_i  Z\, \nabla_j Y+\nabla_i Y\, \nabla_j  Z\big)] \cr
&\quad    +   \lambda m^{-1}\theta'\,[X,g^{ij}\, Y^b \,  Z^c R_{ecbi} \,g^{ae} (\partial_j \partial_a ) ] \cr 
& \quad +  \lambda^2(2m) ^{-1} g^{ij}\,\theta'\big(\nabla_i  \nabla_X Z\, \nabla_j Y+ \nabla_i Y\, \nabla_j \nabla_X  Z\big) 
 -  \lambda^2m^{-1}g^{ij}\, \theta'\,Y^b \,  (\nabla_X Z)^c R_{ecbi} \,g^{ae} \partial_j \partial_a   \cr 
& \quad +  \lambda^2(2m) ^{-1} g^{ij}\,\theta'\big(\nabla_i  Z\,  \nabla_j \nabla_X Y+ \nabla_i \nabla_X Y\, \nabla_j  Z\big)
 -  \lambda^2m^{-1}g^{ij}\, \theta'\,(\nabla_X Y)^b \,  Z^c R_{ecbi} \,g^{ae} \partial_j \partial_a  \cr
 &=  \lambda^2m^{-1}\,\theta'\, g^{ij}\,\nabla_j X    \,  \big(\nabla_{\nabla_i  Z} Y + \nabla_{\nabla_i Y}  Z\big)  
 - \lambda(2m) ^{-1} \theta' [X,g^{ij} ] \,\big(\nabla_i  Z\, \nabla_j Y + \nabla_i Y\, \nabla_j  Z\big) \cr
 &\quad  - \lambda^2 (2m) ^{-1} \theta' g^{ij}\,\big( [X, \nabla_i  Z]_{\rm Lie}\, \nabla_j Y + [X, \nabla_i Y]_{\rm Lie} \, \nabla_j  Z
 + \nabla_i  Z\, [X, \nabla_j Y]_{\rm Lie} + \nabla_i Y\, [X, \nabla_j  Z]_{\rm Lie} \big) \cr
&\quad    +   \lambda m^{-1}\theta'\,[X,g^{ij}\, Y^b \,  Z^c R_{ecbi} \,g^{ae} ] \, \partial_j \partial_a 
 +   \lambda m^{-1}\theta'\,g^{ij}\, Y^b \,  Z^c R_{ecbi} \,g^{ae} [X,\partial_j \partial_a  ] \cr 
& \quad +  \lambda^2(2m) ^{-1} g^{ij}\,\theta'\big(\nabla_i  \nabla_X Z\, \nabla_j Y+ \nabla_i Y\, \nabla_j \nabla_X  Z\big) 
 -  \lambda^2m^{-1}g^{ij}\, \theta'\,Y^b \,  (\nabla_X Z)^c R_{ecbi} \,g^{ae} \partial_j \partial_a   \cr 
& \quad +  \lambda^2(2m) ^{-1} g^{ij}\,\theta'\big(\nabla_i  Z\,  \nabla_j \nabla_X Y+ \nabla_i \nabla_X Y\, \nabla_j  Z\big)
 -  \lambda^2m^{-1}g^{ij}\, \theta'\,(\nabla_X Y)^b \,  Z^c R_{ecbi} \,g^{ae} \partial_j \partial_a  \cr
 &=  \lambda^2m^{-1}\,\theta'\, g^{ij}\,\nabla_j X    \,  \big(\nabla_{\nabla_i  Z} Y + \nabla_{\nabla_i Y}  Z\big)  
 -  \lambda m^{-1} \theta' [X,g^{ij} ] \,\nabla_i  Z\, \nabla_j Y  \cr
 &\quad  - \lambda^2 m^{-1} \theta' g^{ij}\,\big( [X, \nabla_i  Z]_{\rm Lie}\, \nabla_j Y + [X, \nabla_i Y]_{\rm Lie} \, \nabla_j  Z\big) \cr
&\quad    +  \lambda m^{-1}\theta'\,[X,g^{ij}\, Y^b \,  Z^c R_{ecbi} \,g^{ae} ] \, \partial_j \partial_a 
 +   \lambda m^{-1}\theta'\,g^{ij}\, Y^b \,  Z^c R_{ecbi} \,g^{ae} [X,\partial_j \partial_a  ] \cr 
& \quad + \lambda^2m^{-1} g^{ij}\,\theta'  \nabla_i  \nabla_X Z\, \nabla_j Y
 -  \lambda^2m^{-1}g^{ij}\, \theta'\,Y^b \,  (\nabla_X Z)^c R_{ecbi} \,g^{ae} \partial_j \partial_a   \cr 
& \quad +  \lambda^2m^{-1} g^{ij}\,\theta' \nabla_i  Z\,  \nabla_j \nabla_X Y
 -  \lambda^2m^{-1}g^{ij}\, \theta'\,(\nabla_X Y)^b \,  Z^c R_{ecbi} \,g^{ae} \partial_j \partial_a  \cr
  &=  \lambda^2m^{-1}\,\theta'\, g^{ij}\,\nabla_j X    \,  \big(\nabla_{\nabla_i  Z} Y + \nabla_{\nabla_i Y}  Z\big)  
 - \lambda m^{-1} \theta' [X,g^{ij} ] \,\nabla_i  Z\, \nabla_j Y  \cr
 &\quad  + \lambda^2 m^{-1} \theta' g^{ij}\,\big( \nabla_{\nabla_i Z}X \,\nabla_j Y +
 \nabla_{\nabla_i Y}X \, \nabla_j  Z\big) \cr
&\quad    +   \lambda^2m^{-1}\theta'\,X^q\,g^{ij}\, Y^b \,  Z^c R_{ecbi;q} \,g^{ae}  \, \partial_j \partial_a 
   -   \lambda^2m^{-1}\theta'\,X^q\,\Gamma^j{}_{kq}    \,g^{ik}\, Y^b \,  Z^c R_{ecbi} \,g^{ae}  \, \partial_j \partial_a 
   \cr
&\quad    -   \lambda^2m^{-1}\theta'\,X^q\,\Gamma^a{}_{kq}    \,g^{ij}\, Y^b \,  Z^c R_{ecbi} \,g^{ke}  \, \partial_j \partial_a +  \lambda m^{-1}\theta'\,g^{ij}\, Y^b \,  Z^c R_{ecbi} \,g^{ae} [X,\partial_j \partial_a  ] \cr 
& \quad +  \lambda^2m^{-1} g^{ij}\,\theta'  [\nabla_i , \nabla_X] Z\, \nabla_j Y
  +  \lambda^2m^{-1} g^{ij}\,\theta' \nabla_i  Z\,  [\nabla_j ,\nabla_X] Y
 \cr
   &=  \lambda^2m^{-1}\,\theta'\, g^{ij}\,\nabla_j X    \,  \big(\nabla_{\nabla_i  Z} Y + \nabla_{\nabla_i Y}  Z\big)  
 + \lambda^2m^{-1} \theta' X^q(  g^{kj}\,\Gamma^i{}_{qk} +g^{ik}\,\Gamma^j{}_{qk}  )\,\nabla_i  Z\, \nabla_j Y  \cr
 &\quad  + \lambda^2 m^{-1} \theta' g^{ij}\,\big( \nabla_{\nabla_i Z}X \,\nabla_j Y +
 \nabla_{\nabla_i Y}X \, \nabla_j  Z\big) \cr
&\quad    +   \lambda^2m^{-1}\theta'\,X^q\,g^{ij}\, Y^b \,  Z^c R_{ecbi;q} \,g^{ae}  \, \partial_j \partial_a 
 -   \lambda^2m^{-1}\theta'\,g^{ij}\, Y^b \,  Z^c R_{ecbi} \,g^{ae}(\nabla_j X\,\partial_a+ \nabla_a X\,\partial_j) \cr 
& \quad +  \lambda^2m^{-1} g^{ij}\,\theta' \big( R(\del_i , X) Z\, \nabla_j Y + \nabla_{ [\del_i ,X]_{\rm Lie} } Z\, \nabla_j Y
  + \nabla_i  Z\,  R(\del_j ,X) Y   + \nabla_i  Z\,  \nabla_{ [ \del_j ,X]_{\rm Lie} } Y \big)
 \cr
    &=  \lambda^2m^{-1}\,\theta'\, g^{ij}    \,  \big(\nabla_{\nabla_i  Z} Y \,\nabla_j X + \nabla_{\nabla_i Y}  Z \,\nabla_j X   +
    \nabla_{\nabla_i Z}X \,\nabla_j Y +
 \nabla_{\nabla_i Y}X \, \nabla_j  Z     +   \nabla_{ \nabla_j X } Y  \, \nabla_i  Z+ \nabla_{ \nabla_i X} Z\, \nabla_j Y  \big) \cr
&\quad    +   \lambda^2m^{-1}\theta'\,X^q\,g^{ij}\, Y^b \,  Z^c R_{ecbi;q} \,g^{ae}  \, \partial_j \partial_a 
 -   \lambda^2m^{-1}\theta'\,g^{ij}\, Y^b \,  Z^c R_{ecbi} \,g^{ae}(\nabla_j X\,\partial_a+ \nabla_a X\,\partial_j) \cr 
& \quad +  \lambda^2m^{-1} g^{ij}\,\theta' \big( R(\del_i , X) Z\, \nabla_j Y 
  + \nabla_i  Z\,  R(\del_j ,X) Y    \big)
 \cr
     &=  \lambda^2m^{-1}\,\theta'\, g^{ij}    \,  \big(\nabla_{\nabla_i  Z} Y \,\nabla_j X + \nabla_{\nabla_i Y}  Z \,\nabla_j X   +
    \nabla_{\nabla_i Z}X \,\nabla_j Y +
 \nabla_{\nabla_i Y}X \, \nabla_j  Z     +   \nabla_{ \nabla_j X } Y  \, \nabla_i  Z+ \nabla_{ \nabla_i X} Z\, \nabla_j Y  \big) \cr
&\quad    +   \lambda^2m^{-1}\theta'\,X^q\,g^{ij}\, Y^b \,  Z^c R_{ecbi;q} \,g^{ae}  \, \partial_j \partial_a \cr 
& \quad +  \lambda^2m^{-1} g^{ij}\,\theta' \big( R(\del_i , X) Z\, \nabla_j Y 
  +  R(\del_j ,X) Y \, \nabla_i  Z   + R(\del_i , Y) Z\, \nabla_j X + R(\del_i , Z) Y\, \nabla_j X    \big)
 \end{align*}
So the total $\theta'$ contribution to the Jacobi operator from terms not containing $V$ is
\begin{align*}
&\quad     \lambda^2m^{-1}\theta'\,X^q\,g^{ij}\, Y^b \,  Z^c (R_{ecbi;q} - R_{ecqi;b}) \,g^{ae}  \, \partial_j \partial_a \cr 
& \quad +  \lambda^2m^{-1} g^{ij}\,\theta' \big(  R(\del_j ,X) Y \, \nabla_i  Z    + R(\del_i , Z) Y\, \nabla_j X   
- R(\del_j ,Y) X\, \nabla_i  Z    - R(\del_i , Z) X\, \nabla_j Y     \big)\ .
 \end{align*}
The terms containing $V$ are easily computed separately. Then
\begin{align*}
&J(X,Y,\extd Z) 
= \lambda^2\, \extd(R(X,Y) Z)  +\lambda^2\,\theta'\, X^qY^bZ^a\, R^k{}_{aqb}\, V_{,k}\cr
 &\quad
 +  \lambda^2\,\extd x^i\, (\nabla_{  R(  \frac{\del}{\del x^i},X) Z} Y 
 + R(\nabla_i Y,X)Z 
 +  R(\nabla_i  Z  , X )Y 
 + \nabla_{ R(  \frac{\del}{\del x^i},X) Y}  Z)  \cr
 &\quad -  \lambda^2\,\extd x^i\, (\nabla_{  R(  \frac{\del}{\del x^i},Y) Z} X 
 + R(\nabla_i X,Y)Z 
 +  R(\nabla_i  Z  , Y )X
 + \nabla_{ R(  \frac{\del}{\del x^i},Y) X}  Z) \cr
 &\quad  +   \lambda^2m^{-1}\theta'\,X^q\,g^{ij}\, Y^b \,  Z^c (R_{ecbi;q} - R_{ecqi;b}) \,g^{ae}  \, \partial_j \partial_a \cr 
& \quad +  \lambda^2m^{-1} g^{ij}\,\theta' \big(  R(\del_j ,X) Y \, \nabla_i  Z    + R(\del_i , Z) Y\, \nabla_j X   
- R(\del_j ,Y) X\, \nabla_i  Z    - R(\del_i , Z) X\, \nabla_j Y     \big)\cr
&= \lambda^2\, \extd(R(X,Y) Z)  +\lambda^2\,\theta'\, X^qY^bZ^a\, R^k{}_{aqb}\, V_{,k}
+   \lambda^2m^{-1}\theta'\,X^q\,g^{ij}\, Y^b \,  Z^c R_{ecbq;i}  \,g^{ae}  \, \partial_j \partial_a \cr
 &\quad
 +  \lambda^2\,\extd x^i\, (R^a{}_{riq} X^q Y^b{}_{;a} Z^r\,\del_b + R^a{}_{rpi} X^b{}_{;a} Y^p Z^r\,\del_b
  + R^b{}_{rpq} X^q{}_{;i} Y^p Z^r \,\del_b \cr
 &\quad + R^b{}_{rpq} X^q Y^p{}_{;i} Z^r \,\del_b
 +  R^b{}_{rpq} Y^p X^q Z^r{}_{;i}\,\del_b
+ R^a{}_{ipq} X^q Y^p Z^b{}_{;a}\,\del_b)  \cr
& \quad +  \lambda^2m^{-1} g^{ij}\,\theta' \big(  R^b{}_{jpq} X^q Y^p\,\del_b \, \nabla_i  Z    + R^b{}_{pir} Z^r Y^p\,\del_b\, \nabla_j X   
    - R^b{}_{qir} Z^r X^q\,\del_b\, \nabla_j Y     \big)\cr
    &= \lambda^2\, \extd(R(X,Y) Z)  +\lambda^2\,\theta'\, X^qY^bZ^a\, R^k{}_{aqb}\, V_{,k}
+   \lambda^2m^{-1}\theta'\,g^{ij}\, X^q\,Y^p \,  Z^r R^a{}_{rpq;i}   \, \partial_j \partial_a \cr
 &\quad
 +  \lambda^2\,\extd x^i\, (R^a{}_{riq} X^q Y^b{}_{;a} Z^r\,\del_b + R^a{}_{rpi} X^b{}_{;a} Y^p Z^r\,\del_b
  + \nabla_i(R^b{}_{rpq} X^q Y^p Z^r \,\del_b) \cr
 &\quad - R^b{}_{rpq;i} X^q Y^pZ^r \,\del_b
+ R^a{}_{ipq} X^q Y^p Z^b{}_{;a}\,\del_b)  \cr
& \quad +  \lambda^2m^{-1} g^{ij}\,\theta' \big(  R^b{}_{jpq} X^q Y^p\,\del_b \, \nabla_i  Z    + R^b{}_{pir} Z^r Y^p\,\del_b\, \nabla_j X   
    - R^b{}_{qir} Z^r X^q\,\del_b\, \nabla_j Y     \big)\cr
        &= \lambda^2\, \extd(R(X,Y) Z)   - \lambda^2\,\extd x^i\,  \nabla_i(R(X,Y) Z)  +\lambda^2\,\theta'\, (R(X,Y) Z)^k\, V_{,k} \cr
&\quad + \lambda^2\,R^b{}_{rqp;i} X^q Y^pZ^r \,\del_b\, (\extd x^i
 -   m^{-1}\theta'\,g^{ij}  \, \partial_j   )
 \cr
 &\quad
 +  \lambda^2\,\extd x^i\, R^a{}_{riq} X^q  Z^r\nabla_a Y      -  \lambda^2m^{-1} g^{ij}\,\theta' R^b{}_{qir} Z^r X^q\,\del_b\, \nabla_j Y 
 \cr
 &\quad + \lambda^2\,\extd x^i\, R^a{}_{rpi}  Y^p Z^r\, \nabla_a X     +  \lambda^2m^{-1} g^{ij}\,\theta' R^b{}_{pir} Z^r Y^p\,\del_b\, \nabla_j X   
\cr
 &\quad 
+ \lambda^2\,\extd x^i\, R^a{}_{ipq} X^q Y^p\, \nabla_a Z  +  \lambda^2m^{-1} g^{ij}\,\theta'   R^b{}_{jpq} X^q Y^p\,\del_b \, \nabla_i  Z,\end{align*}
which gives the result stated. \end{proof}

\begin{corollary}\label{cor:imjac} The images of all the Jacobiators above are in the kernel of $\rho$, in fact in the space spanned by elements of the form (\ref{keromega1}) and the expressions in Proposition~\ref{hdkk}.
\end{corollary}
\proof This is by inspection of most of the terms except for the last case if we write $U=R(X,Y)Z$ then the $-\lambda^2\extd x^i$ terms in the first line can be replaced by $-(\extd x^i-m^{-1}\theta'\del_j)\nabla_i U- m^{-1}\theta'\del_j U^a{}_{;i}\del_a$ and the second term here combines with the other terms on  the right to give an expression  in the kernel of the form in Proposition~\ref{hdkk} applied to $U$.\endproof

\section{Operator geodesic equations from associativity}\label{sec:geo}

We have constructed the calculus in the previous sections motivated by the Schr\"odinger representation and a chosen Hamiltonian. This calculus as we have seen has a Jacobiator (it is not associative) even between 0-forms and 1-forms i.e. $\Omega^1_{\CD(M)}$ is not quite a bimodule over $\CD(M)$ if there is sufficiently nontrivial curvature. We can, however, impose relations that kill the non-associativity if we want. Indeed, the Schr\"odinger representation maps to an associative operator algebra and hence all the Jacobiators must have their image in its kernel, hence it is natural to kill this kernel. We keep the option of a potential $V$ in the choice of $\ch$, although in our spacetime application, we will set this $V=0$. 

\begin{corollary} \label{corbimod} The quotient of $\Omega^1_{\CD(M)}$ by elements of the form (\ref{keromega1}) and the expressions in Proposition~\ref{hdkk} (including with potential $V$) in the kernel of the Schr\"odinger representation is a first order differential calculus to order $\lambda^2$. 
\end{corollary}
\proof As explained in Section~\ref{sec:prod}, it is enough to work in local coordinates where we can take $\extd x^\mu,\extd p_\nu,\theta'$ as a left basis over the algebra. After quotienting, we therefore obtain  $\theta'$ as basis over the algebra and this was taken as central and associative with respect to products by $\CD(M)$, so we have a bimodule to the stated order. Moreover, writing $\extd a=D(a)\theta'$ for $a\in \CD(M)$, it is also part of the construction in Section~\ref{sec:calc} that this $D$ is a derivation to the stated order (we will shortly interpret it as ${\extd\over\extd s}$), but we  check this explicitly. 
From  (\ref{keromega1}) and Prop~\ref{hdkk} respectively, 
\begin{align*}
m D(f)&=  g^{ij}\,f_{,j}\, \partial_i +{ \lambda\over 2m} g^{ij}\, f_{,i;j}\\
m D(X) &= \,g^{ij}\, X^a{}_{;i}\, (\del_a\del_j- \lambda\Gamma^k{}_{aj} \, \del_k)
+{\lambda\over 2}(\Delta X  +  X^a\, g^{ij}\, R_{ja}\,\del_i)
  -  m X(V).
\end{align*}
Hence, for functions $f,h$, 
\begin{align*}
m\,D(f)h&+m\,f\,D(h)=g^{ij}\big(f_{,j}\, (\del_i\ ) h+\frac{\lambda}2 f_{,i;j}h+f\,h_{,j}\,\del_i+ \frac{\lambda}2
f\,h_{,i;j}\big)\\
&=g^{ij}\big((f h)_{,j} \del_i +\lambda f_{,j}h_{,i}+\frac{\lambda}2 f_{,i;j}h+ \frac{\lambda}2
f\,h_{,i;j}\big)= m\,D(fh).\end{align*}
Next, we already used in Proposition~\ref{gup} that $\extd( f X)=\extd f.X+f \extd X$ in obtaining the relations, but proceeding directly,
\begin{align*}
m\,D(fX)&-m\,f\,D(X) =   g^{ij}\, f_{,i}\, X^a\, (\del_a\del_j- \lambda\Gamma^k{}_{aj} \, \del_k)
+{\lambda\over 2}    (    2\,g^{ij}\, f_{,i} \,X^a{}_{;j}\,\del_a + g^{ij}\, f_{,i;j}\, X) \cr
&=   g^{ij}\, f_{,i}\, X\,\del_j
-  g^{ij}\, f_{,i}\, X^a\,  \lambda\Gamma^k{}_{aj} \, \del_k
+{\lambda\over 2}    (    2\,g^{ij}\, f_{,i} \,X^a{}_{;j}\,\del_a + g^{ij}\, f_{,i;j}\, X) 
\end{align*}
so
\begin{align*}
m\,(D(fX)&-f\,D(X)-D(f)\,X) 
=   g^{ij}\, f_{,i}\, [X,\del_j]
-  g^{ij}\, f_{,i}\, X^a\,  \lambda\Gamma^k{}_{aj} \, \del_k
+\lambda  \,g^{ij}\, f_{,i} \,X^a{}_{;j}\,\del_a 
\end{align*}
which gives zero by the commutation relations between vector fields. $Xf$ is defined by $fX$ and the commutator $[f,X]$ and that this is compatible with $\extd$ a derivation was  Proposition~\ref{geid}. Finally, for  products of vector fields, we extend $D$ as a derivation on the tensor product of vector fields tensor over the field (which is automatic) and show that this is well defined once we introduce the relations on the algebra. One relation is $(X.f)Y=X(fY)$ for functions $f$, which we have dealt with. The other nontrivial relation is $XY-YX=[X,Y]$ where $[X,Y]$ is a single vector field. In Proposition~\ref{gup}, we obtained the relations there from $\extd$ applied to the commutator of vector fields and one can check explicitly that this can be pushed the other way to $D$ well-defined by the derivation rule. 
 \endproof

Next, moving towards  applications, it will be convenient to say what $D$ looks like on local coordinate vector fields $\del_i$. For this, we need the following lemma. 

\begin{lemma}\label{lapdel} The Laplace-Beltrami operator on coordinate basis vector fields is
\[\Delta \del_i
= (-g^{j k}\, R_{k i} + g^{ab} \, \Gamma^j{}_{ab,i})\del_j.\]
\end{lemma}
 \proof The general formula for the Laplacian on a vector field reduces to $\Delta (\del_i) =(g^{ab}(\Gamma^j{}_{ia,b}+\Gamma^j{}_{cb}\Gamma^c{}_{ia})-  \Gamma^j{}_{ic} \Gamma^c)\del_j$, which we then identify in terms of the Ricci tensor as stated. 
\endproof

Finally, we write $\theta'=\extd s$, where $s$ will have the interpretation as a `geodesic time' variable but for the moment this is just some central 1-form. Dividing through by this and using the preceding lemma,  the quotient relations as in Corollary~\ref{corbimod} become
\begin{align}\label{kerrho1a}m {\extd x^i\over\extd s}&=g^{ij} \del_j - \tfrac{\lambda}{2} \, \Gamma^i,\\ \label{kerrho2a}
m {\extd \del_i\over\extd s}&= \Gamma^j{}_{i a}g^{ab}(\del_j \del_b-\lambda\Gamma^k{}_{jb}\del_k)+{\lambda\over 2}g^{ab} \, \Gamma^j{}_{ab,i}\,\del_j - m V_{,i}.\end{align}
to order $\lambda^2$, where $\Gamma^i=\Gamma^i{}_{ab}g^{ab}$ and $\extd\over\extd s$ denotes the coefficient of $\extd s$ on applying $\extd$, i.e. what we called $D$ in the proof of the corollary. We view these quotient relations as a first order formalism for  noncommutative geodesic equations due to the following:

\begin{proposition} \label{propC} Eliminating $\del_i$ in terms of $\extd x^i\over\extd s$, we obtain to order $\lambda$ 
\[  \frac{\extd^2 x^i}{\extd s^2}  +  \Gamma^i{}_{jk} \,  \frac{\extd x^j}{\extd s} \,  \frac{\extd x^k}{\extd s} + {g^{ij} \over m} V_{,j} ={\lambda\over 2m}  C^i{}_j  {\extd x^j\over\extd s},\]
\[ C^{ij}=-g^{ab}(g^{i c}\Gamma^j{}_{c a,b}+g^{j c}\Gamma^i{}_{c a,b})+  g^{i b}\Gamma^j{}_{;b}-g^{j b}\Gamma^i{}_{;b}+\Gamma^{a bi}\Gamma^{j}{}_{ab}- \Gamma^{abj}\Gamma^{i}{}_{ab},\]
where we use the notation $\Gamma^i{}_{;j}:=\Gamma^i{}_{,j}+ \Gamma^i{}_{kj}\Gamma^k$.\end{proposition}
\proof We use $T$ for the operation that extracts the coefficient of $\theta'$, so
\begin{align}\label{T} m T(\widehat{\extd f})&=  g^{bc}\,f_{,c}\, \partial_b + {\lambda\over 2} g^{bc}\,f_{,b;c},\nonumber\\
m T(\extd \del_i)&=  g^{bc}\, \Gamma^a{}_{bi} \, (\del_a\del_c- \lambda\Gamma^k{}_{ac} \, \del_k)
+{\lambda \over 2} g^{ab} \, \Gamma^j{}_{ab,i}\del_j  -  m\,V_{,i}  \end{align}
for any function $f$ on $M$, and using Lemma~\ref{lapdel}. Then applying $\extd$ to (\ref{kerrho1a}), 
\begin{align*}
m^2 \frac{\extd^2 x^i }{\extd s^2} &= m\,T(\extd g^{ij })\, \del_j + 
 m\,g^{ij} T(\extd \del_j)  - \tfrac{\lambda}{2}  \, m\, T(\extd \Gamma^i ) \cr
 &=  \ \big(   g^{ab}\,  g^{ij}{}_{,b}  \, \del_a + {\lambda\over 2} g^{ab}\,g^{ij}{}_{,a,b}  - {\lambda\over 2} \Gamma^k \, g^{ij}{}_{,k}  \big)  \del_j  \cr
 &\quad +  g^{ij} \, 
 \big(\Gamma^c{}_{ja}g^{ab}(\del_c \del_b-\lambda\Gamma^d{}_{cb}\del_d)+ \frac{\lambda}{2}g^{ab} \, \Gamma^c{}_{ab,j}\,\del_c - {m\over 2} V_{,j}\big) \cr
 &\quad - \tfrac{\lambda}{2}  \,  \left(   g^{ab}\,  \Gamma^i{}_{,b}  \, \del_a +{ \lambda\over 2} g^{ab}\,\Gamma^i{}_{,a,b}  - {\lambda\over 2} \Gamma^j \, \Gamma^i{}_{,j}  \right) 
 \end{align*}
 to which we add
   \begin{align*}
&m^2 \Gamma^i{}_{ab} \,  \frac{\extd x^a}{\extd s} \,  \frac{\extd x^b}{\extd s} =  \Gamma^i{}_{ab} \, 
\left(g^{aj} \del_j - \tfrac{\lambda}{2} \, \Gamma^a \right) \,
\left(g^{b c} \del_c - \tfrac{\lambda}{2} \, \Gamma^b \right).
 \end{align*}
The quadratic in $\del$'s is order zero and this vanishes after matching indices and using an identity of the form $g^{ij}{}_{,k}= -\Gamma^i{}_{p k}g^{pj}- \Gamma^j{}_{p k}g^{pi}$. In doing so, we pick up a derivative of $g$ from moving a $\del_j$ to the right. The resulting order $\lambda$ terms are  $\del_j$ times 
   \begin{align*}
C^{ij}&= g^{ab}\,g^{ij}{}_{,a,b}  -  \Gamma^k \, g^{ij }{}_{,k}   +g^{ik } \, 
 \big(-2\,\Gamma^c{}_{ka}g^{ab}\, \Gamma^ j {}_{bc} + g^{ab} \, \Gamma^ j {}_{ab,k} \big) 
  -  g^{ j  b}\,  \Gamma^ i {}_{,b}    \cr
 &\quad +  \Gamma^ i {}_{ab} \, 
\big(2 \,g^{ak} \, g^{b j }{}_{,k}  
 - g^{a j }  \, \Gamma^b  
  - \Gamma^a  g^{b j }  \big) \cr
  &=   g^{ab}\,g^{ij}{}_{,a,b}  - \Gamma^b \, g^{ij}{}_{,b}   
    -   g^{ j  b}\,  \Gamma^ i {}_{,b} 
    + g^{ik }  \big(-2\,\Gamma^c{}_{ka}g^{ab}\, \Gamma^ j {}_{bc} + g^{ab} \, \Gamma^ j {}_{ab,k} \big) 
 \cr
 &\quad +  \Gamma^ i {}_{ab} \, 
\big(2 \,g^{ak} \, g^{b j }{}_{,k}  
 -2  g^{a j }  \, \Gamma^b
 \big)  \cr
   &=  - g^{ab}\,( \Gamma^ j {}_{pa}g^{p i }  +    \Gamma^ i {}_{pa} g^{p  j }    )_{,b}  +\Gamma^a \, (   \Gamma^ j {}_{pa}g^{p i }  +    \Gamma^ i {}_{pa} g^{p  j }  )   
    -   g^{ j  b}\,  \Gamma^ i {}_{,b} 
    \\
    &+ g^{ik } \, 
 \big(-2\,\Gamma^c{}_{k a}g^{ab}\, \Gamma^ j {}_{bc} + g^{ab} \, \Gamma^ j {}_{ab,k} \big) +  \Gamma^ i {}_{ab} \, 
\big(-2 \,g^{a k}(   \Gamma^b{}_{p k}g^{p j }+ \Gamma^ j {}_{p k}g^{p b}    )  
 -2  g^{a j }  \, \Gamma^b
 \big)  
  \end{align*}
  where at the end, we expanded out three derivatives of the metric tensor in terms of Christoffel symbols using the identity above. Similarly expanding the remaining derivative and making a lot of cancelations gives the result stated after replacing $\del_j$ by $m{\extd x^j\over\extd s}$ and a lowered index.
 \endproof
 
 The matrix $C^{ij}$ with indices raised has the first term symmetric and the remaining terms antisymmetric. It is not a tensor and indeed we do not want it to transform as one due to the noncommutative nature of the coordinates and calculus on $\CD(M)$. Also note that since our results are valid to order $\lambda^2$, one can also similarly determine the order $\lambda^2$ correction. Next, the Hamiltonian $\ch\in \CD(M)$ is necessarily constant under these equations. 
 
\begin{corollary}\label{cordh} Let
\[ \ch= {1\over 2m} g^{ij}\,   (\del_i\del_j- \lambda\Gamma^k{}_{ij} \, \del_k)    +V \in \CD(M).\]
 Then $\extd \ch=0$ in the quotient bimodule, i.e. ${\extd\ch\over\extd s}=0$ at least to order $\lambda$.
\end{corollary}
\proof This follows in principle from the way $\extd$ was defined via the Schr\"odinger representation and $[\rh,\ ]$, if we assume that (\ref{kerrho1a})-(\ref{kerrho2a}) generate the whole kernel so that $\rho$ becomes injective on the quotient. Here we just check it directly. We have
\begin{align*}
2m\,\extd \ch &= g^{ij}\,   (\extd\del_i \,\del_j + \del_i\,\extd\del_j  - \lambda\extd \Gamma^t{}_{ij} \, \del_t
- \lambda\Gamma^t{}_{ij} \, \extd\del_t)    + \extd(g^{ij})\,   (\del_i\del_j- \lambda\Gamma^t{}_{ij} \, \del_t)   +2m\, \extd V
\end{align*}
Using (\ref{T}), we have
\begin{align*}
4m^2\, T(\extd \ch) &= g^{ij}\,  \big( 2\,g^{bc}\, \Gamma^a{}_{bi} \, (\del_a\del_c- \lambda\Gamma^k{}_{ac} \, \del_k)
+\lambda g^{ab}\Gamma^k{}_{ab,i}\del_k   -  2m\,V_{,i}\big) \,\del_j  \cr
&\quad +   g^{ij}\, \del_j\, \big(2\,g^{bc}\, \Gamma^a{}_{bi} \, (\del_a\del_c- \lambda\Gamma^k{}_{ac} \, \del_k)
+\lambda  g^{ab}\Gamma^k{}_{ab,i}\del_k   -  2m\,V_{,i}\big) \cr
&\quad  - \lambda  g^{ij}  \,  \big( 2 g^{bc}\,( \Gamma^t{}_{ij})_{,c}\, \partial_b + \lambda\, g^{bc}\,( \Gamma^t{}_{ij})_{,b;c}\big)\, \del_t \cr
&\quad 
- \lambda g^{kj}\, \Gamma^i{}_{kj} \, \big(  2\,g^{bc}\, \Gamma^a{}_{bi} \, (\del_a\del_c- \lambda\Gamma^k{}_{ac} \, \del_k)
+\lambda  g^{ab}\Gamma^k{}_{ab,i}\del_k   -  2m\,V_{,i} \big) \cr
&\quad  +  \big( 2 g^{bc}\,(g^{ij})_{,c}\, \partial_b + \lambda\, g^{bc}\,(g^{ij})_{,b;c}\big)  \,   (\del_i\del_j- \lambda\Gamma^t{}_{ij} \, \del_t)     \cr
&\quad  + 2m\, \big( 2 g^{bc}\,V_{,c}\, \partial_b + \lambda\, g^{bc}\,V_{,b;c}\big),
\end{align*}
where we do not apply the covariant derivative to the indices of $g^{ij}$ and $ \Gamma^t{}_{ij}$ in the brackets. The potential terms cancel and the three terms without $\lambda$ in this expression total
\begin{align*}
& g^{ij}\,   2\,g^{bc}\, \Gamma^a{}_{bi} \, \del_a\del_c\del_j
 +   g^{ij}\, \del_j\, 2\,g^{bc}\, \Gamma^a{}_{bi} \, \del_a\del_c    +   2 g^{bc}\,g^{ij}{}_{,c}\, \partial_b   \,   \del_i\del_j    \cr
 &= g^{ij}\, [\del_j , 2\,g^{bc}\, \Gamma^a{}_{bi}] \, \del_a\del_c 
 + 2g^{ij}\,  \big(  g^{bc}\, \Gamma^a{}_{bi} \, \del_a\del_c
 +   g^{ba}\, \Gamma^c{}_{bi} \, \del_a\del_c    +   g^{ac}{}_{,i}  \,   \del_a\del_c \big) \, \partial_j   
\end{align*}
and as the bracket vanishes, and moving all coordinate vectors to the right, we get
\begin{align*}
4 m^2\, & T(\extd \ch) =\lambda( g^{ij}\,  \big( 2\,g^{bc}\, \Gamma^a{}_{bi} \, ( - \Gamma^k{}_{ac} \, \del_k)
+   g^{ab}\Gamma^k{}_{ab,i}\del_k    \big) \,\del_j  \cr
&\quad +   g^{ij}\, \del_j\, \big(2\,g^{bc}\, \Gamma^a{}_{bi} \, (-  \Gamma^k{}_{ac} \, \del_k)
+    g^{ab}\Gamma^k{}_{ab,i}\del_k     \big) \cr
&\quad  -    g^{ij}  \,  \big( 2 g^{bc}\,( \Gamma^k{}_{ij})_{,c}\, \partial_b + \lambda\, g^{bc}\,( \Gamma^k{}_{ij})_{,b;c}\big)\, \del_k \cr
&\quad 
-   g^{kj}\, \Gamma^i{}_{kj} \, \big(  2\,g^{bc}\, \Gamma^a{}_{bi} \, (\del_a\del_c- \lambda\Gamma^k{}_{ac} \, \del_k)
+\lambda  g^{ab}\Gamma^k{}_{ab,i}\del_k     \big) \cr
&\quad  +    g^{bc}\big((g^{ij})_{,b,c}-g^{ij}{}_{,a}\Gamma^a{}_{bc}\big)    (\del_i\del_j-  \lambda\Gamma^k{}_{ij} \, \del_k) 
-   2 g^{bc}\,(g^{ij})_{,c}\, \Gamma^k{}_{ij} \, \partial_b   \del_k    \cr
&\quad + 2 \, g^{ij}\, g^{bc}{}_{,j}\, \Gamma^a{}_{bi} \, \del_a\del_c +  2 \, g^{ij}\, g^{bc}\, \Gamma^a{}_{bi,j} \, \del_a\del_c  -  \lambda 2 g^{bc}\,(g^{ij})_{,c}\, (\Gamma^k{}_{ij,b}) \del_k
\end{align*}
Expanding the $(g^{ij})_{,b}$ part of $(g^{ij})_{,b,c}$ derivative generates derivatives of $\Gamma$'s and one can then check that all order $\lambda$ derivative of $\Gamma$ terms cancel. At order $\lambda$ we then expand all remaining derivatives of the metric, which generates $\Gamma^2$ terms and find that these also all cancel. 
 \endproof

The above results are all that we need for the applications that follow. However, our motivation came out of quantum geodesics and it remains to fill in some of this noncommutative geometry. Here we limit ourselves to finding the geodesic velocity vector field $\mathfrak{X}: \Omega^1_{\CD(M)}\to \CD(M)$ as a bimodule map to order $\lambda$.

\begin{proposition} \label{propX} There is a geodesic velocity field $\mathfrak{X}$ underlying the model given by
\begin{align*}\mathfrak{X}(\extd f)&={1\over m}\big(g^{ij}f_{,i}\del_j+{\lambda\over 2} \Delta f),\\
 \mathfrak{X}(\extd Y)&={1\over m}\big(g^{ij} Y^a{}_{;i}(\del_a\del_j -\lambda  \Gamma^b{}_{aj}\del_b   )+ {\lambda\over 2}\left(\Delta Y + Y^a\, g^{ij}\, R_{ja}\,\del_i\right)\big)- Y(V)
 \end{align*}
at least to order $\lambda$. We also set $\mathfrak{X}(\theta')=1$ so that $\mathfrak{X}$ vanishes on the kernel of $\rho$. 
\end{proposition}
\proof This follows naturally from the way we have constructed the differential calculus if we assume that $\rho$ is injective on the quotient, but we still have to identify it even in this case. We use $\ch$ as above and in view of (\ref{ghu}), we take 
\[
\mathfrak{X}(\extd a)= \lambda^{-1} [ \ch,a]\ .
\]
For $f\in C^\infty(M)$, we see easily that $
m \, \mathfrak{X}(\extd f)= g^{ij}\, f_{,i}\, \del_j + {\lambda\over 2} g^{ij}\, f_{,j;i}$ as stated. 
The more difficult calculation is for the vector field $Y$,
\begin{align*}
2\lambda m \, \mathfrak{X}&(\extd Y) = [g^{ij},Y]\,(\del_i\del_j-\lambda\, \Gamma^k{}_{ij}\del_k) + 2m\, [V,Y] 
 - \lambda\, g^{ij}\,  [\Gamma^k{}_{ij} ,Y]\del_k \cr
 &\quad + g^{ij}\,([\del_i,Y]\del_j + \del_i [\del_j,Y]  -  \lambda\, \Gamma^k{}_{ij}  [\del_k,Y] ) \cr
 &= -\lambda\,Y^p\,g^{ij}{}_{,p}\, (\del_i\del_j-\lambda\, \Gamma^k{}_{ij}\del_k) -2 \lambda m\, Y^p\, V_{,p}
 +  \lambda^2\, g^{ij}\, Y^p\, \Gamma^k{}_{ij,p}\,\del_k \cr
 &\quad + \lambda\,  g^{ij}\,\big( (\nabla_i Y - Y^p\,\Gamma^q{}_{pi}\,\del_q) \del_j + \del_i  (\nabla_j Y - Y^p\,\Gamma^q{}_{pj}\,\del_q) -  \lambda\, \Gamma^k{}_{ij}  (\nabla_k Y - Y^p\,\Gamma^q{}_{pk}\,\del_q)
  \big) 
\end{align*}
so
\begin{align*}
2m \, \mathfrak{X}&(\extd Y) 
 =  Y^p \,   g^{in}\,\Gamma^j{}_{np} \,  (\del_i\del_j-\lambda\, \Gamma^k{}_{ij}\del_k) -2 m\, Y^p\, V_{,p}
 +  \lambda\, g^{ij}\, Y^p\, \Gamma^k{}_{ij,p}\,\del_k \cr
 &\quad +  g^{ij}\,\big( (\nabla_i Y ) \del_j + \del_i  (\nabla_j Y - Y^p\,\Gamma^q{}_{pj}\,\del_q) -  \lambda\, \Gamma^k{}_{ij}  (\nabla_k Y - Y^p\,\Gamma^q{}_{pk}\,\del_q)
  \big) \cr
  &\quad -\lambda\, Y^p \,   g^{ji}\,\Gamma^m{}_{ip}  \Gamma^k{}_{mj}\del_k   \cr
  &=  Y^p \,   g^{in}\,\Gamma^j{}_{np} \,  (\del_i\del_j-\lambda\, \Gamma^k{}_{ij}\del_k) - 2m\, Y^p\, V_{,p}
 \cr
 &\quad +  g^{ij}\,\big(\nabla_i Y \, \del_j + \del_i \, \nabla_j Y  -  \lambda\, \Gamma^k{}_{ij}  \nabla_k Y 
  \big) \cr
   &\quad -  \lambda\, g^{ij}\, Y^p{}_{,i}\,\Gamma^q{}_{pj}\,\del_q
    -  \lambda \, g^{ij}\, Y^p\,\Gamma^q{}_{pj,i}\,\del_q
     -  g^{ij}\,  Y^p\,\Gamma^q{}_{pj}\,  \del_i \del_q
\cr
  &\quad 
  + \lambda\,Y^p\,g^{ij}( \Gamma^k{}_{ij}  \Gamma^q{}_{pk} 
  -   \Gamma^m{}_{ip}  \Gamma^q{}_{mj} 
   +  \Gamma^q{}_{ij,p})   \,\del_q \cr
     &=   - 2m\, Y^p\, V_{,p}
+  g^{ij}\,\big(  2\nabla_i Y \, \del_j + [\del_i , \nabla_j Y]  -  \lambda\, \Gamma^k{}_{ij}  \nabla_k Y 
  \big)  -  \lambda\, g^{ij}\, Y^p{}_{,i}\,\Gamma^q{}_{pj}\,\del_q
\cr
  &\quad 
  + \lambda\,Y^p\,g^{ij}( \Gamma^k{}_{ij}  \Gamma^q{}_{pk} 
  -   \Gamma^m{}_{ip}  \Gamma^q{}_{mj} 
   +  \Gamma^q{}_{ij,p}
   - \Gamma^m{}_{jp} \,   \Gamma^q{}_{im}
       -  \Gamma^q{}_{pj,i}  )\,\del_q \cr
            &=   - 2m\, Y^p\, V_{,p}
+  g^{ij}\,\big(  2\nabla_i Y \, \del_j 
+\lambda\, \nabla_i\nabla_j Y - \lambda\, Y^p{}_{;j}\,\Gamma^q{}_{ip}\,\del_q
 -  \lambda\, \Gamma^k{}_{ij}  \nabla_k Y 
  \big)  -  \lambda\, g^{ij}\, Y^p{}_{;i}\,\Gamma^q{}_{pj}\,\del_q
\cr
  &\quad 
  + \lambda\,Y^p\,g^{ij}( \Gamma^k{}_{ij}  \Gamma^q{}_{pk} 
   +  \Gamma^q{}_{ij,p}
   - \Gamma^m{}_{jp} \,   \Gamma^q{}_{im}
       -  \Gamma^q{}_{pj,i}  )\,\del_q \cr
                   &=   -2 m\, Y^p\, V_{,p}
+  2 g^{ij} \nabla_i Y \, \del_j 
+\lambda\, \Delta Y 
    - 2 \lambda\, g^{ij}\, Y^p{}_{;i}\,\Gamma^q{}_{pj}\,\del_q
  + \lambda\,Y^p\,g^{ij} \,  R^q{}_{jpi}  \,\del_q \cr
                     &=   - 2 m\, Y^p\, V_{,p}
+  2 g^{ij} \nabla_i Y \, \del_j 
+\lambda\, \Delta Y   - 2 \lambda\, g^{ij}\, Y^p{}_{;i}\,\Gamma^q{}_{pj}\,\del_q
  + \lambda\,Y^p\,g^{qn} \,  R_{np} \,\del_q
\end{align*}
which we write as stated. We can compare the result with (\ref{keromega1}) and the formula in Proposition~\ref{hdkk} to conclude that $\mathfrak{X}$ vanishes on these kernel elements if we set $\mathfrak{X}(\theta')=1$. The difference is that we are now using the commutator in $\CD(M)$ not its image under $\rho$ as we did in Section~\ref{sec:calc}.
\endproof

In principle, we also need a right bimodule connection $\nabla$ on $\Omega^1_{\CD(M)}$ at least to order $\lambda$, with respect to which $\mathfrak{X}$ obeys the geodesic velocity equations. This can in principle be found by similar methods \cite{BegMa:geo} but will be looked at  elsewhere.

Finally, while Proposition~\ref{propC} justifies our interpretation $\theta'=\extd s$ for proper time in a `generalised Heisenberg picture' for the evolution of algebra elements, this necessarily has a corresponding `Schr\"odinger picture'  with evolution of pure states according to
\begin{equation}\label{KGQM}- \lambda {\del\psi\over\del s}= \rh\psi\end{equation}
for suitable $\psi$. This is exactly the quantum geodesic amplitude flow equation $\nabla_E\psi=0$ from (\ref{nablaE}) if we identify $\extd s$ with the geodesic time parameter interval there. This justifies our interpretation of the theory. Even though wave functions $\psi$ on $M$  in the case where $M$ is spacetime are not something usually considered, we see that this arises naturally from quantum geodesics and our above results.

\section{Basic examples} \label{sec:ex}
Here we compute the geometric content of our formulae in various special cases as a check of consistency. The one for the Schwarzschild black hole will be used in applications in Section~\ref{sec:app}.

\subsection{The flat case}

When $M$ is flat in the sense that the Levi-Civita connection has zero curvature, the algebra of differential operators looks locally like the flat spacetime Heisenberg algebra but the difference is that our constructions are geometric and coordinate-invariant, which is still of interest. The general results above for $\Omega^1_{\CD(M)}$ to order $\lambda^2$ can be written in the flat case as 
\begin{align*} 
[\widehat \xi,f] &= m^{-1}\lambda\ \xi^\#(f) \theta' \\ 
 [X,\widehat\xi ] &= \lambda\,(\widehat{ \nabla_X \xi  })  - (2m)^{-1}\lambda\,\theta'\, \big(2 \nabla_{\xi^\#} X
+\lambda  X^a{}_{;i}\, \xi^{\#i}_{;a}\big)\\
 [\extd X,f ] &=  \lambda\,\widehat{(\<\nabla_i X,\extd f\>\extd x^i) }  + (2m)^{-1}\lambda\,\theta'\, \big(2 \nabla_{\extd f^\#} X 
+   \lambda\,  X^a{}_{;i}\, \extd f^{\#i}_{;a}\big)\\
[    Y  , \extd X  ] &=\lambda\, \extd(\nabla_Y X) - \lambda \theta'\, \<\nabla_Y(\extd V),X\>  - \lambda \widehat{\extd x^i}\, \big(\nabla_{\nabla_i X} Y + \nabla_{\nabla_i Y} X\big) \\
&\kern-40pt + (2m)^{-1}\lambda\theta' \, g^{ij}\, \big(\lambda \nabla_j(\nabla_{\nabla_i X} Y)-\nabla_i X\, \nabla_j Y+\lambda  \Gamma^k{}_{aj} (\nabla_i X)^a \nabla_k Y + X\leftrightarrow Y \big)
\end{align*}
where $\xi^\#$ is $\xi$ converted to a vector field via the metric. We have seen that all the Jacobi identities associated with being a bimodule then hold to this order. This reduces to \cite{BegMa:geo} when we identify the
image of $\del_i$ in $\CD(M)$ as $p_i$ and $\lambda=-\imath\hbar$ and choose special flat space coordinates where $\Gamma=0$ so that $\nabla_i(\del_j)=0$. Also note that in the flat case,${\rm Vect}(M)$ is a pre-Lie algebra with $X\circ Y=\nabla_X Y$ so that
\[ X\circ Y-Y\circ X=[X,Y]_{{\rm Vect(M)}}, \quad X\circ (Y\circ Z)-Y\circ(X\circ Z)=(X\circ Y-Y\circ X)\circ Z\]
and in this case $U({\rm Vect}(M))\subset \CD(M)$ has a calculus with $[Y, \extd X]=\lambda\, \extd(\nabla_Y X)$ as a general construction for pre-Lie algebras. Our construction has a bigger algebra but we see this as part of the relevant commutator.  

Finally, we have a Schr\"odinger representation of $\Omega^1_{\CD(M)}$ given by
\begin{align*} \rho(f)\psi&=f\psi,\quad \rho(X)\psi=\lambda X(\psi),\quad \rho(\hat\xi)\psi={\lambda\over 2m}
((\ ,\ )\nabla\xi)\psi + 2 \xi^{\#}(\psi)\\
\rho(\extd X)\psi&= -X(V)\psi + {\lambda^2\over 2m}( \Delta(X)(\psi)
+  2 g^{ij} \<\nabla_i X,\nabla_j\extd\psi\> ),\quad \rho(\theta')\psi=\psi\end{align*}
which, in the absence of an associativity obstruction, can be expected to extend to an entire exterior algebra $\Omega_{\CD(M)}$.

\subsection{The compact Lie group case}

This has the merit, as for the compact real form of any complex semisimple Lie group $G$, of a trivial tangent bundle allowing calculations to be written at a Lie algebra level. Here $\CD(G)=C^\infty(G)\rtimes U(\cg)$ where the Lie algebra $\cg$ of $G$ acts by left-invariant vector fields. This is because, on a Lie group, one has a global basis of left-invariant vector fields which by themselves generate $U(\cg)$. Any polynomial of functions and vector fields can be considered equivalently as a polynomial in functions and the basis of left-invariant vector fields by moving all coefficients to the left using the commutation relations.

We start with the algebra generated by the functions and its centrally extended differential calculus. This has generators $e^a$ for the anti-Hermitian basis of 1-forms over the algebra with dual basis $\del_a$ of left-invariant vector fields. The real structure constants are defined by the Lie bracket $[\del_a,\del_b]_g=c_{ab}{}^c\del_c$ and the Killing form, which in the compact case is {\em negative} definite in this basis. We take the metric $g^{ab}=(e^a,e^b)$ as given by this up to a normalisation. We also have $e^{a\#}=(e^a,\ )=g^{ab}\del_b$ to convert a 1-form to a vector field. Ad-invariance of the metric and its (more usual) inverse are respectively
\[ c_{ab}{}^d g^{cb}+ c_{ab}{}^c g^{bd}=0,\quad c_{ab}{}^d g_{dc}+ c_{ac}{}^d g_{bd}=0.\]

The calculus has form-function relations
\[ [e^a, f]= m^{-1}\lambda g(e^a,\extd f)\theta'= m^{-1}\lambda g^{ab}(\del_b f)\theta',\quad \extd f=(\del_a f)e^a+(2m)^{-1}\lambda(\Delta f)\theta'\]
where $\Delta=g^{ab}\del_a\del_b$ is the Laplacian in our conventions. The products of 1-forms and functions here are in the quantised $\Omega^1$, i.e. one could use $\bullet$ and put a hat on the $e^a$). The quantum version of a classical 1-form $fe_a$, writing $\bullet$ explicitly, is then 
\[ \widehat{f e^a}=f\bullet e^a+ (2m)^{-1}\lambda(\extd f,e^a)\theta'=f\bullet e^a+(2m)^{-1}\lambda(\del_b f)g^{ba}\theta'.\]

Next, on a Lie group, the Levi-Civita connection has generalised Christoffel symbols for the basis and curvature given by
\[ \nabla_{\del_a}\del_b={1\over 2}[\del_a,\del_b]_g={1\over 2} c_{ab}{}^c\del_c,\quad \nabla_{\del_a}e^b=-{1\over 2}c_{ac}{}^b e^c,\quad \Gamma^c{}_{ba}={1\over 2}c_{ab}{}^c \]
\[R(\del_a,\del_b)\del_c={1\over 4}[[\del_a,\del_b]_g,\del_c]_g={1\over 4}c_{ab}{}^e c_{ec}{}^d\del_d\]
and the latter in index conventions translates to
\[  R^d{}_{cab}={1\over 4} c_{ab}{}^e c_{ec}{}^d,\quad R_{cb}={1\over 4} c_{ab}{}^d c_{dc}{}^a={1\over 4}K_{cb},\]
where $K_{ab}$ is the Killing form. For the canonical Riemannian geometry on $G$, $g_{ab}$ will be proportional to this and we identify this constant as 
\begin{equation}\label{liemetric} g_{ab}={\dim\cg\over 4 R_{sc}}K_{ab}\end{equation}
where $R_{sc}$ is the Ricci scalar curvature.  

Next, for $\CD(G)$, we add $\del_a$ into the algebra with relations
\[ [\del_a,\del_b]=\lambda c_{ab}{}^c\del_c, \quad [\del_a,f]=\lambda \del_a f\]
and according to our general results, we take commutation relations in $\Omega^1(\CD(G))$ with
\begin{align*} [\del_a,e^b]&=\lambda \nabla_{\del_a}e^b- m^{-1}\lambda\theta'\nabla_{e^{b\#}}\del_a-(2m)^{-1}\lambda^2\theta'\big(R(\del_a,e^{b\#})+ {\rm Tr}((\nabla \del_a )(\nabla e^{b\#}))\big)\\
&= -{\lambda\over 2}c_{ac}{}^b e^c-(2m)^{-1}\lambda \theta' g^{bc}c_{ca}{}^d\del_d-(2m)^{-1}\lambda^2\theta'\big(R_{ac}g^{cb}-{1\over 4}{\rm Tr}(c_{ca}{}^d e^c\tens \del_d \tens g^{be}c_{fe}{}^p e^f\tens \del_p)\big)\\
&= -{\lambda\over 2}c_{ac}{}^b e^c- (2m)^{-1}\lambda \theta' g^{bc}c_{ca}{}^d\del_d-(2m)^{-1}{\lambda^2\over 4}\theta'\big( c_{fc}{}^d c_{da}{}^f   g^{cb}- c_{ca}{}^d g^{be}c_{de}{}^c  \big)\\
&= -{\lambda\over 2}c_{ac}{}^b e^c- (2m)^{-1}\lambda \theta' g^{bc}c_{ca}{}^d\del_d= {\lambda\over 2}c_{ca}{}^b(e^c+m^{-1}\theta' e^{c\#}),
 \end{align*}
where at the end we used  that the Ricci tensor is proportional to the Killing form and hence to the inverse metric $g_{ab}$. By similar calculations, we have
 \begin{align*}
 [\extd \del_a,f ] &= \lambda\,\widehat{( \id\tens\<\extd f,\ \>)\nabla \del_a }  + m^{-1}\lambda\,\theta'\,  \nabla_{\extd f^\#}\del_a 
+   (2m)^{-1}\lambda^2\, \theta'\big(R(\del_a,\extd f^{\#}) +{\rm Tr}((\nabla \del_a)( \nabla\extd f^{\#}))
\big)\\
&={\lambda\over 2}c_{ba}{}^c\del_c f e^b+(2m)^{-1}{\lambda^2\over 2}c_{ba}{}^c g^{db}\del_d \del_c f\theta'+m^{-1}\lambda\theta' \del_b f\nabla_{\extd e^{b\#}}\del_a\\
&\quad+(2m)^{-1}\lambda^2\theta' \del_b f\big(R(\del_a,e^{b\#}) +{\rm Tr}((\nabla \del_a)( \nabla e^{b\#}))\big)+ (2m)^{-1}\lambda^2\theta'{\rm Tr}((\nabla\del_a)(\extd \del_c f\tens e^{c\#}))\\
&={\lambda\over 2}c_{ba}{}^c\del_c f e^b+(2m)^{-1}\lambda\theta' \del_b f g^{bc}c_{ca}{}^d \del_d+(2m)^{-1}{\lambda^2\over 2}\theta'\big(c_{ba}{}^c g^{db}\del_d\del_c f+ c_{ba}{}^d g^{cb} \del_d\del_cf\big)\\
&=\del_b f\big(-{\lambda\over 2}c_{ac}{}^b  e^c+(2m)^{-1}\lambda\theta'   g^{bc}c_{ca}{}^d \del_d) ={\lambda\over 2}\del_b f c_{ca}{}^b\tilde{e^c}
\end{align*}
where
\[\tilde \xi=\xi- m^{-1}\theta'\xi^{\#}.\]
 For the 3rd equality, we recognised the previous vanishing Ricci + Tr expression but have an extra term due to $\nabla \extd f^\#$ not being tensorial in the coefficients of $\extd f$. The 4th equality uses ad-invariance of the metric. One can check that this is consistent with $\extd$ applied to the relations $[\del_a,f]=\lambda \del_a f$ when expanded by the Leibniz rule and with expressions of the form $[\del_a, f e^b]=[\del_a ,f]e^b+ f[\del_a,e^b]$ again expanded as usual. Finally
\begin{align*} [\del_a,\extd\del_b]&=  {\lambda\over 2}c_{ab}{}^c \extd \del_c + \lambda P^{(1)}(\del_a,\del_b)+\lambda P^{(2)}(\del_a,\del_b),
\end{align*}
where we break $P(X,Y)$ into terms without and with curvature in the general expression. Here
\begin{align*} P^{(1)}(\del_a,\del_b)&= -\theta'\<\nabla_{\del_b}\extd V,\del_a\>-(e^c\nabla_{\nabla_{\del c}\del_a}\del_b +(2m)^{-1}\lambda \theta' g^{cd}\nabla_{\del_d}\nabla_{\nabla_{\del_c}\del_a}\del_b+ a\leftrightarrow b)\\
&\quad - (2m)^{-1}\theta' g^{cd}\big(  (\nabla_{\del_c}\del_a)(\nabla_{\del_d}\del_b)-\lambda \Gamma^e{}_{sd}(\nabla_{\del c}\del_a)^s \nabla_{\del_e}\del_b+ a\leftrightarrow b\big)\\
&=-\theta'(\del_b(\del_a V)-{1\over 2}c_{ba}{}^c\del_c V)- {1\over 4}e^c(c_{ca}{}^dc_{db}{}^e+ c_{cb}{}^dc_{da}{}^e )\del_e\\
&\quad -{m^{-1}\lambda\over 16}\theta' g^{cd}c_{de}{}^f(  c_{ca}{}^sc_{sb}{}^e+c_{cb}{}^sc_{sa}{}^e )\del_f\\
&\quad  -{m^{-1}\over 8}\theta'g^{cd}( (c_{ca}{}^sc_{db}{}^t+c_{cb}{}^sc_{da}{}^t)  \del_s\del_t -{\lambda\over 2} c_{ds}{}^e (c_{ca}{}^s c_{eb}{}^f+c_{cb}{}^s c_{ea}{}^f) \del_f)\\
&=-\theta'(\del_b(\del_a V)-{1\over 2}c_{ba}{}^c\del_c V)- {1\over 4}e^c(c_{ca}{}^dc_{db}{}^e+ c_{cb}{}^dc_{da}{}^e )\del_e\\
&\quad -{m^{-1}\over 8}\theta'g^{cd} c_{cb}{}^sc_{da}{}^t( \del_s\del_t +\del_t\del_s), 
\end{align*}
where for the last equality we used the Jacobi identity and antisymmetry of the Lie bracket to cancel all the $\lambda/8$ terms.  We also have
\begin{align*} P^{(2)}(\del_a,\del_b)&=m^{-1}\theta' g^{cd}R^e{}_{abc}(\del_d\del_e-\lambda\Gamma^f{}_{ed}\del_f)\\
&\quad+{\lambda\over 2}\Big(\nabla_{\del_c}(R)(\del_a,\del_b)+ \<R(\del_c,\nabla_{\del_d}\del_a)\del_b,e^d\>+ \<R(\del_c,\nabla_{\del_d}\del_b)\del_a,e^d\>\Big)\tilde{e^c}\\
&\quad+\lambda (2m)^{-1}\theta'\Big( \nabla_{\del_a}(R)(\del_b,\del_c)+\nabla_{\del_b}(R)(\del_a,\del_c) - \nabla_{\del_c}(R)(\del_a,\del_b) \\
&\qquad\qquad -\<R(\del_c,\del_a)\nabla_{\del_d}\del_b+R(\del_c,\del_b)\nabla_{\del_d}\del_a,e^d\> \Big)e^{c\#}  \\
&\quad-\lambda(2m)^{-1}\theta' g^{cd}\big( R(\del_b,\del_c)\ \nabla_{\del_d}\del_a+R(\del_a,\del_c)\ \nabla_{\del_d}\del_b)\\
&=m^{-1}\theta' g^{cd}R^e{}_{abc}(\del_d\del_e-{\lambda\over 2}c_{de}{}^f\del_f)+ {\lambda\over 4}(R{}^d{}_{bce}c_{da}{}^e+ R{}^d{}_{ace}c_{db}{}^e)\tilde{e^c}\\
&\quad -{\lambda\over 4}m^{-1}\theta'(  R^d{}_{eca}c_{db}{}^e+ R^d{}_{ecb}c_{da}{}^e    )e{}^{c\#}  -{\lambda\over 4}m^{-1}\theta'g^{cd}(  R_{bc}c_{da}{}^e + R_{ac}c_{db}{}^e  )\del_e\\
&={m^{-1}\over 4}\theta' g^{cd} c_{cb}{}^s  c_{da}{}^t (\del_s\del_t-{\lambda\over 2}c_{st}{}^f\del_f)+ {\lambda\over 16}\big(c_{ce}{}^f c_{fb}{}^d c_{da}{}^e+ c_{ce}{}^f c_{fa}{}^d c_{db}{}^e\big)\tilde{e^c}, 
\end{align*}
where for the second equality we used
\[\nabla_{\del_c}(R)(\del_a,\del_b)= \del_c R(\del_a,\del_b)-R(\nabla_{\del_c}\del_a,\del_b)-R(\del_a,\nabla_{\del_c}\del_b) =0\]
since the Ricci tensor is a multiple of the metric and hence covariantly constant. We then used that
\[ c_{ca}{}^f c_{fe}{}^d c_{db}{}^e+ c_{cb}{}^f c_{fe}{}^d c_{da}{}^e=K(\del_b,[\del_a,\del_c])+ K(\del_a,[\del_b,\del_c])=0\]
by invariance and symmetry of the Killing form, so that there is no $e^{c\#}$ terms. We also use that $R_{ab}$ is a multiple of the metric, so that there is no $\del_e$ term. Finally, we put in the formula for $R^e{}_{abc}$ and used ad-invariance of the metric to cast the first term in certain form. This is arranged so that when we add $P^{(1)}$ and $P^{(2)}$, the last term of the former and the first term of the latter exactly cancel giving the final result
\[ P(\del_a,\del_b)=-\theta'(\del_b(\del_a V)-{1\over 2}c_{ba}{}^c\del_c V)- {1\over 4}e^c(c_{ca}{}^dc_{db}{}^e+ c_{cb}{}^dc_{da}{}^e )\del_e+  {\lambda\over 16}\big(c_{ce}{}^f c_{fb}{}^d c_{da}{}^e+ c_{ce}{}^f c_{fa}{}^d c_{db}{}^e\big)\tilde{e^c}.\]
The coefficient of $\tilde{e^c}$ here is a canonical totally symmetric 3-form on the Lie algebra.

Next, we can compute Jacobiators in our case. Following the results of Section~\ref{sec:jac}, the nonzero ones come out using the curvature above as
\begin{align*} J(\del_a,\del_b,e^c)&=-\lambda^2R(\del_a,\del_b)(e^c-m^{-1}\theta'g^{cd}\del_d))=- {\lambda^2\over 4}c_{ab}{}^e c_{ed}{}^c\tilde{e^d} \\
J(f,\del_a, \extd \del_b)&=\lambda^2\big((R(\del_a,\del_c)\del_b) (f)e^c-m^{-1}\theta'R(\del_b,(\extd f)^\#)\del_a\big)\\
&= { \lambda^2\over 4}(\del_df)\big( c_{ac}{}^e c_{eb}{}^d e^c-m^{-1}\theta' g^{dc}c_{bc}{}^e c_{ea}{}^f\del_f\big)= { \lambda^2\over 4}(\del_df)c_{ac}{}^ec_{eb}{}^d \tilde{e^c}\end{align*}
where at end we used ad-invariance of the metric. We also have
\begin{align*}
J(&\del_a,\del_b,\extd\del_c)=\lambda^2\Big(\extd (R(\del_a,\del_b)\del_c)- m^{-1}\theta' (g^{fe}\del_e)\nabla_{\del_f}(R(\del_a,\del_b)\del_c)+\theta' (R(\del_a,\del_b)\del_c)(V)\\
&\qquad\qquad\qquad-\tilde{e^d}\big(   R(\nabla_{\del_d}\del_a,\del_b)\del_c+ R(\del_a,\nabla_{\del_d}\del_b)\del_c+ R(\del_a,\del_b)\nabla_{\del_d}\del_c)\\
&\qquad\qquad\qquad+\nabla_{R(\del_a,\del_d)\del_c}\del_b+\nabla_{R(\del_d,\del_b)\del_c}\del_a+\nabla_{R(\del_a,\del_b)\del_d}\del_c \big)\Big)\\
&={\lambda^2}R^d{}_{cab}(\extd \del_d - (2m)^{-1}\theta' g^{fs}c_{fd}{}^t\del_s\del_t+\theta'\del_d V)\\
&-{\lambda^2\over 2}\big(R^f{}_{ceb}c_{da}{}^e+ R^f{}_{cae}c_{db}{}^e+R^f{}_{eab}c_{dc}{}^e+R^e{}_{cad}c_{eb}{}^f+ R^e{}_{cdb}c_{ea}{}^f+R^e{}_{dab}c_{ec}{}^f  \big)\tilde{e^d}\del_f\\
&={\lambda^2}R^d{}_{cab}(\extd \del_d +\theta'\del_d V)-{\lambda^2\over 2}\big(R^f{}_{eab}c_{dc}{}^e+R^e{}_{cad}c_{eb}{}^f+ R^e{}_{cdb}c_{ea}{}^f  \big)\tilde{e^d}\del_f\\
&={\lambda^2\over 4}c_{ab}{}^ec_{ec}{}^d (\extd \del_d +\theta'\del_d V)-{\lambda^2\over 8}\big( c_{ac}{}^s c_{sd}{}^t c_{tb}{}^f- c_{bc}{}^s c_{sd}{}^t c_{ta}{}^f \big)\tilde{e^d}\del_f.
\end{align*}
For the 3rd equality, we dropped the $\del_s\del_t $ term as these commute at order 1 and are contracted with something antisymmetric by invariance of the metric. We also cancelled 3 of the 6 similar terms after inserting the value of $R$ and using the Jacobi identity for the Lie algebra. The remaining 3 terms do not cancel but again using the Jacobi identity in the Lie algebra can be condensed to two for the 4th equality. 

Finally, we compute what the Schr\"odinger representation looks like in the Lie group case. Here,  
\[ \rho(f)\psi=f\psi,\quad \rho(\del_a)\psi=\lambda\del_a(\psi),\quad \rho(e^a)\psi=m^{-1}\lambda e^{a\#}(\psi),\quad\rho(\theta')\psi=\psi\]
since
\[ (\ ,\ )\nabla e^b=-{1\over 2} (e^a,e^c)c_{ac}{}^b=-{1\over 2}g^{ac}c_{ac}{}^b=0.\]
This extends the usual Schr\"odinger representation to 1-forms on $M$ by converting them to vector fields by the metric and the scale factor $m$. In addition, we have
\begin{align*}
\rho(\extd \del_a)\psi&= -( \del_a V)\psi +{\lambda^2\over 2m}\big(  (\Delta \del_a)(\psi)+ R(\del_b, \del_a) e^{b\#}(\psi)+ 2 g^{bc}\<\nabla_{\del_b}\del_a,  \nabla_{\del_c}\extd\psi\>  \big)\\
&= - \del_a(V)\psi+ {\lambda^2\over 2m}\big(  (\Delta \del_a)(\psi)+ R(\del_b, \del_a) e^{b\#}(\psi)+ 2 g^{bc}\<\nabla_{\del_b}\del_a,  \nabla_{\del_c}\extd\psi\>  \big)\\
&= - (\del_a V)\psi+ {\lambda^2\over 2m}\big( (\Delta\del_a)(\psi)+{R_{sc}\over\dim\cg}  \del_a\psi- {1\over 2}g^{bc}c_{ba}{}^fc_{cf}{}^d   \del_d\psi+ g^{bc}c_{ba}{}^d\del_c\del_d \psi\big)\\
&= - (\del_a V)\psi+ {\lambda^2\over 2m}\big((\Delta\del_a)(\psi)- {R_{sc}\over\dim\cg}  \del_a\psi\big)= - (\del_a V)\psi
\end{align*}
using by invariance of the metric so that
\[ g^{bc}c_{ba}{}^fc_{cf}{}^d=-c_{ba}{}^cg^{bf}c_{cf}^d=c_{ba}{}^cc_{cf}{}^b g^{fd}=K_{fa}g^{fd}={4R_{sc}\over\dim\cg}\delta_a^d.\]
We also used that $\del_c\del_d$ commute to order 1. We then used that in the Lie group case for our basis,
\[ \Delta(\del_a)=(\ ,\ )\nabla(e^b\tens \nabla_{\del_b}\del_a)=g^{bc}\nabla_{\del_b}\nabla_{\del_c}(\del_a)={1\over 4}g^{bc}c_{ca}{}^dc_{bd}{}^e\del_e= {R_{sc}\over\dim\cg}\del_a\]
by a similar computation. 

\begin{example}For $G=SU(2)=S^3$ we have the Lie algebra $[\del_i,\del_j]=\eps_{ijk}\del_k$ in a basis of left-invariant vector fields, where $\eps_{ijk}$ is the totally antisymmetric tensor. Then the Killing form and symmetric trilinear form are
\[K_{ij}= \< [\del_i,   [\del_j,\del_k]],e^k\>=\eps_{jk l}\eps_{ilk}=(\delta_{jk}\delta_{ki}-\delta_{ji}\delta_{kk})=-2\delta_{ij}\]
\begin{align*} K_{ijk}&=\<[\del_i, [\del_j,   [\del_k,\del_l]]]+[\del_j, [\del_i,   [\del_k,\del_l]]],e^l\> =  \eps_{klm}\eps_{jmn}\eps_{inl}+  \eps_{klm}\eps_{imn}\eps_{jnl}=0.
\end{align*}
We set $g_{ij}=-\delta_{ij}$ which corresponds under (\ref{liemetric}) to a certain radius so that the  Ricci scalar is  3/2. Then we have
\[ [\del_i,e^j]={\lambda\over 2}\eps_{ijk} (e^k-m^{-1}\theta'\del_k),\quad [\extd\del_i,f]={\lambda\over 2}\eps_{ijk}(\del_j f) \tilde{e^k};\quad \tilde{e^k}= e^k+ m^{-1}\theta'\del_k\]
\begin{align*} [\del_i,\extd\del_j]& = {\lambda\over 2}\eps_{ijk}\extd \del_k - \lambda\theta'\big(\del_j(\del_i V)-{1\over 2}\eps_{jik}\del_kV\big)-{\lambda\over 4}\big(e^j\del_i+e^i\del_j-2\delta_{ij}e^k\del_k\big).
\end{align*}
The Jacobiators are
\[ J(\del_i,\del_j,e^k)={\lambda^2\over 4}(\delta_{ik}\tilde{e^j}-\delta_{jk}\tilde{e^i},\quad J(f,\del_i,\extd \del_j)={\lambda^2\over 4}(\delta_{ij}\del_k f \tilde{e^k}- \del_i f\tilde{e^j})\]
\[ J(\del_i,\del_j,\extd\del_k)={\lambda^2\over 4}\Big(\delta_{ik}(\extd\del_j+\theta'\del_j V)-\delta_{jk}(\extd\del_i+\theta'\del_i V)\Big)+{\lambda^2\over 8}(\eps_{kil}\tilde{e^j}-\eps_{kjl}\tilde{e^i})\del_l.\]
Finally, the Schr\"odinger representation is
\[ \rho(f)\psi=f\psi,\quad \rho(\del_i)\psi=\lambda\del_i \psi,\quad \rho(e^i)\psi=-{\lambda\over m}\del_i\psi,\quad\rho(\extd\del_i)=-(\del_i V)\psi,\quad \rho(\theta')\psi=\psi.\]
\end{example}

\subsection{The Schwarzschild metric} \label{sec:bh}  

One can analyse the theory for a general static rotationally invariant spacetime. Here we just focus on the representative black-hole case with $r_s$ the `Schwarzschild radius' as a free parameter, and we also set the external potential $V=0$.
The Ricci tensor vanishes,  the metric and the  Christoffel symbols are
\[ g_{\mu\nu}\extd x^\mu\extd x^\nu=-(1-{r_s\over r})\extd t^2+{1\over 1-{r_s\over r}}\extd r^2+r^2(\extd\theta^2+\sin^2(\theta)\extd\phi^2)\]
\begin{align*}
\Gamma^t{}_{tr}&=\frac{r_s}{2 r^2 \left(1-\frac{r_s}{r}\right)},\quad \Gamma^r{}_{tt}={r_s\over 2 r^2}\left(1-\frac{r_s}{r}\right) ,\quad \Gamma^r{}_{rr}=- \Gamma^t{}_{tr},\quad 
\Gamma^r{}_{\theta\theta}=-r \left(1-\frac{r_s}{r}\right),\\ \Gamma^r{}_{\phi\phi}&= \sin^2(\theta)\,  \Gamma^r{}_{\theta\theta},\quad 
 \Gamma^\theta{}_{r\theta} = \tfrac1r\ ,\quad \Gamma^\theta{}_{\phi\phi}=-\sin\theta\, \cos\theta\ ,\quad 
\Gamma^\phi{}_{r\phi} = \tfrac1r\ ,\quad \Gamma^\phi{}_{\theta\phi} =\cot\theta.
\end{align*}
Recalling our notation $\Gamma^\mu:=\Gamma^{\mu}{}_{\alpha\beta}g^{\alpha\beta}$, these come out as
\[ \Gamma^t=\Gamma^\phi=0,\quad \Gamma^r=-{1\over r}(2-{r_s\over r}),\quad \Gamma^\theta=-{1\over r^2}\cot(\theta).\]
The Ricci tensor is zero but the Laplacian on the coordinate basis vector fields is not zero and we compute it as 
\[ \Delta\del_t=\Delta\del_\phi=0\quad \Delta\del_r=-{2\over r^3}(r-r_s)\del_r,\quad \Delta\del_\theta=-{1\over r^2}\left(2(r-r_s)\cot(\theta)\del_r+ (\cot(\theta)^2-1)\del_\theta\right).\]

We compute the kernel relations from (\ref{kerrho1}) for the coordinate basis as
\begin{equation}\label{ptpr} \del_t=- m  (1-{r_s\over r}){\extd t\over\extd s},\quad \del_r={m \over 1-{r_s\over r}} {\extd r\over\extd s}-\lambda{2 r-r_s\over 2 r(r-r_s)},\end{equation}
\begin{equation}\label{pang}\del_\phi= m  r^2\sin^2(\theta){\extd \phi\over\extd s},\quad \del_\theta= m  r^2 {\extd \theta\over\extd s}-{\lambda\over 2}\cot(\theta)\end{equation}
for the momentum operators. We also have (\ref{kerrho2}) as
\[ m {\extd \del_t\over\extd  s}= m {\extd \del_\phi\over\extd s}=0,\quad m {\extd \del_\theta\over\extd s}={\cos(\theta)\over r^2\sin^3(\theta)}\del_\phi^2+{\lambda\over 2 r^2\sin^2(\theta)}\del_\theta, \]
\[ m {\extd \del_r\over\extd s}= -{r_s\over 2 (r-r_s)^2}\del_t^2- {r_s\over 2 r^2}\del_r^2+{1\over r^3}\del_{sph}^2+{\lambda\over r^3}(r-{r_s})\del_r \]
where
\[ \del_{sph}^2:=\del_\theta^2+ {\del_\phi^2\over\sin^2(\theta)}+ \lambda \cot(\theta) \del_\theta \]
is the `spherical momentum'. These are our $\CD(M)$-valued geodesic equations in first order form. Note that  $\del_\mu\in \CD(M)$ are (locally defined) vector fields and we are not obliged to think of them as differential operators. 

\begin{proposition} The spherical momentum and the total momentum
\[ -\del^2_{tot}=-{\del_t^2\over 1-{r_s\over r}}+(1-{r_s\over r})\del_r^2  +{\lambda\over r}(2-{r_s\over r})\del_r+ {\del_{sph}^2\over r^2} \]
are constants,  ${\extd \del_{sph}^2\over\extd s}={\extd \del_{tot}^2\over\extd s}=0$ to order $\lambda$.
\end{proposition}
\proof (1)  the differential in $\CD(M)$ is $\extd f=\widehat{\extd f}$ and using (\ref{keromega1}) in the kernel of the Schr\"odinger representation, this becomes in particular 
\[ m{\extd f(\theta)\over\extd s}={1\over r^2}(f' \del_\theta+ {\lambda\over 2}(f''+\cot(\theta)f')).\]
from the form of $\Gamma^\theta$. We use this and  $[\del_\theta,f(\theta)]=\lambda f'$ along with our expressions for $\extd \del_\theta$ to compute that  
\[ m {\extd\over\extd  s}(\del_\theta^2+ {\del_\phi^2\over\sin^2(\theta)}) ={\lambda\over r^2\sin^2(\theta)}(\del_\theta^2- {\del_\phi^2\cot^2(\theta)\over\sin^2(\theta)})\]
to order $\lambda$. The addition of the quantum correction $\lambda \cot(\theta) \del_\theta$ exactly kills this. 

(2)  $-\del_{tot}^2:=g^{\mu\nu}\del_\mu \del_\nu -\lambda \Gamma^\mu \del_\mu$ is the expression stated but this is proportional to the Hamiltonian $\ch\in \CD(M)$ in our set-up. Hence this follows from Cor~\ref{cordh} applied in the case of the Schwarzschild metric. It can also be verified explicitly as an excellent check on our calculations, using 
\[ m{\extd f(r)\over\extd s}= (1-{r_s\over r})(f' \del_r + {\lambda\over 2} f'')+{\lambda\over 2r } (2-{r_s\over r}) f'\]
obtained from (\ref{keromega1}). \endproof

Since $\del_\mu\in \CD(M)$ map under the Schr\"odinger representation to momentum operators,  we think of them as momentum. Classically,  we would set  $\lambda=0$ and consider them as real momenta $p_\mu$.  Ditto for $p^2_{sph}=p_\theta^2+ p_\phi^2/\sin^2(\theta)$ for the spherical momentum and the total momentum
\begin{equation}\label{unitspeed} -p_{tot}^2=-{p_t^2\over 1-{r_s\over r}}+(1-{r_s\over r})p_r^2  + {p^2_{sph}\over r^2}. \end{equation}
Thus $p_{tot}=m$, the rest mass of the particle if $s$ is proper time. This can be used to express $p_r$ as a function of $r$ and the other three conserved quantities. These four constants of motion then allow one to fully compute geodesics by determining their values for any initial proper velocity. However, to act on wave functions we need the above expressions at least to order $\lambda$ and to view them as operators. Then requiring that the image of $\del_{tot}^2$ is a constant $m^2_{KG}$ becomes the Klein-Gordon equation for a particle of mass $m_{KG}$. 

We also check $*$-compatilibilty. From
\[ \del_\mu^*=\del_\mu + \lambda\Gamma^\nu{}_{\nu\mu}\]
we find
\[ \del_t^*=\del_t,\quad \del_\phi^*=\del_\phi,\quad \del_r^*=\del_r+{2\lambda\over r},\quad \del_\theta^*=\del_\theta+ \lambda \cot(\theta)\]
which implies to order $\lambda$, 
\begin{align*} (\del_{tot}^2)^*&={\del_t^2\over 1-{r_s\over r}}-(\del_r+{2\lambda\over r})^2(1-{r_s\over r})  - {\del_{sph}^2+2\lambda\cot(\theta)\del_\theta \over r^2} +{\lambda\over r}(2-{r_s\over r})\del_r+{\lambda\over r^2}\cot(\theta) \del_\theta
\\
&=\del_{tot}^2+ 2 {\lambda\over r}(2-{r_s\over r})\del_r- {4\lambda\over r}(1-{r_s\over r})\del_r -2 \lambda {r_s\over r^2}\del_r=\del_{tot}^2\end{align*}
and similarly, but more easily, $(\del_{sph}^2)^*=\del_{sph}^2$ as expected. Similarly, 
\[ m {\extd \over\extd s} (\del_r^*)=m {\extd \del_r\over\extd s} +2 \lambda(1-{r_s\over r})(-{1\over r^2})\del_r=m {\extd \del_r\over\extd s}- {2\lambda\over r^3}(r-r_s)\del_r\]
so that
\begin{align*} (m {\extd \del_r\over\extd s})^*&=- {r_s\over 2 (r-r_s)^2}\del_t^2- (\del_r+{2\lambda\over r})^2{r_s\over 2 r^2}+{1\over r^3}\del_{sph}^2-{\lambda\over r^3}(r-{r_s})\del_r\\
& =- {r_s\over 2 (r-r_s)^2}\del_t^2- {r_s\over 2 r^2}\del_r^2+{1\over r^3}\del_{sph}^2-{\lambda\over r^3}(r-{r_s})\del_r =m {\extd \over\extd s} (\del_r^*)\end{align*}
as expected. Similarly, and more easily, for $\del_\theta$, and trivially for $\del_t,\del_\phi$.

\section{Quantum mechanics-like applications}\label{sec:app}

The theory developed in previous sections can be applied in two contexts. The first is $M$ a Riemannian manifold for `space' and the geodesic time variable $s$ identified with regular time $t$. This amounts to a geometric approach to regular quantum mechanics on $M$, to which the theory above applies. This is of interest, but here we focus our attention on the more novel case in which $M$ is spacetime with wave functions  $\psi\in L^2(M)$ over spacetime. 

\subsection{Spacetime quantum mechanics}

We recall from Section~\ref{sec:geo} that we represented the algebra $\CD(M)$ and its differential calculus as an extended Schr\"odinger representation $\rho$ on $L^2(M)$.  We interpreted $\theta'=\extd s$ as a `Heisenberg picture' where  $\lambda{\extd a\over\extd s}=[\rh,a]$ for operators $a$ and $s$ proper time and for some choice of Hamiltonian. We used this for $a=\rho(f)$ in Section~\ref{buio} to define $\extd f$ in $\CD(M)$ and for $a=\rho(X)$ in Proposition~\ref{hdkk} to define $\extd X$. It then follows on products at least to order $\lambda$ working in $\CD(M)$, and hence at least at the order for quantum mechanics in the image under $\rho$. The corresponding Schr\"odinger picture (\ref{KGQM}) moreover matched up to an expected quantum geodesic flow with $s$ the geodesic parameter and this provides the physical meaning of the  external time $s$. If we imagine a density of dust where each particle evolves along a geodesic in spacetime, we can start with an initial configuration of $\rho$ and evolve it by proper time $s$ for each dust particle. If we proceed by analogy with ordinary quantum mechanics, it would then be natural set $\rho=|\psi|^2$ for $\psi$ some kind of `wave function' on spacetime and evolving with $s$ is then the quantum geodesic at hand. This is not quite what we do as we work on $A=\CD(M)$, but the formalism works in principle for any suitable $A-C^\infty(\R)$-bimodule, in the present case consisting of $s$-dependent  wave functions  in $L^2(M)$. This remains the conceptual setting even though our case is also more complicated due to nonassociativity at higher order. We will also refer to momentum operators acting on wave functions and defined with respect to a coordinate basis as 
\begin{equation}\label{pmu}  p_\mu:=\rho(\del_\mu)=\lambda {\del\over\del x^\mu}.\end{equation}

Next, we consider the choice of Hamiltonian. When $M$ is space, we  take $\ch$ so that $\rh=-{\hbar^2\over 2 m } \Delta+ V$ for some external potential function. In the spacetime case we will use $\square$ for the spacetime Laplacian to avoid confusion, and we will focus on the simplest case where the spacetime external potential $V=0$. With spacetime signature $-+++$, we accordingly take 
\begin{equation}\label{minkh} \rh=-{\hbar^2\over 2 m }\square\end{equation}
to define our more novel spacetime or `Klein-Gordon' quantum mechanics, which we will solve in the next section. For the theory to be unitary, we need that $\rh$ is self-adjoint, which depends on the limiting behaviour of fields when the manifold is not compact. The following observation is a step in this direction.

\begin{proposition} For a Schwarzschild background, at least on 2-differentiable radial-only dependent wave functions $\psi(r)$, $\rh$ is symmetric if we impose von Neumann conditions in the sense of $r^2 \psi' \to 0$ as $r\to \infty$ and $(r-r_s)\psi'\to 0$ as  $r\to r_s^+$. 
\end{proposition}
\proof We focus for simplicity on the radial sector of the model, so $\psi=\psi(r)$.  We use the measure $\sqrt{-\det(g)}= r^2\sin(\theta)$ so that the $L^2$-norm for radial functions is effectively
\begin{equation}\label{radmeasure} \<\psi|\psi\>=\int_{r_s}^\infty  |\psi(r)|^2 r^2\extd r\end{equation}
(the $\sin(\theta)$ cancels in expectation values for purely radial calculations, so we ignore this.) Then  $\rh={\lambda^2\over 2 m }\square$ where $\square$ acts on radial functions as $(1-{r_s\over r}){\del^2\over\del r^2} + {1\over r}(2-{r_s\over r}){\del\over\del r}$. Then
\[ \<\phi|\square\psi\>=[\bar\phi\psi'  r(r-r_s)]^\infty_{r_s^+}-\lambda^2\int_{r_s}^\infty \phi' \psi' r(r-r_s)\extd r\]
where prime denotes $\del/\del r$ and where the term from differentiating $r(r-r_s)$ in the integration by parts cancels with the second term of $\square$. If we choose Neumann conditions as stated then we do not pick up anything from the endpoints.  Doing the same for $\<\square\phi|\psi\>$ proceeds in the same way and gives the same answer for the second term.  \endproof

Next, in both space and spacetime cases without external potential,  the images in the Schr\"odinger representation of (\ref{kerrho1})-(\ref{kerrho2}) are set to zero and hence  we automatically have an Ehrenfest theorem,  
\begin{equation}\label{ehren1}  m {\extd\over\extd s}\<\psi|x^\mu|\psi\>=\<\psi|g^{\mu\nu}p_\nu-{\lambda\over 2}\Gamma^\mu|\psi\>,\end{equation}
\begin{equation}\label{ehren2}  m {\extd\over\extd s}\<\psi|p_\mu| \psi\>=\<\psi| \Gamma^\nu{}_{\mu\sigma}g^{\sigma\rho}(p_\nu p_\rho-\lambda \Gamma^\tau{}_{\nu s}p_\tau)+{\lambda\over 2}R_{\nu\mu}g^{\nu\rho}p_\rho |\psi\>.\end{equation}
This differs from  classical geodesic flow for the expectation values of the coordinates because of quantum uncertainties, i.e. since the expectation of a product is not the product of the corresponding expectations. Similarly if we add an external potential $V$. Note that in the Heisenberg picture, the state $\psi$ is fixed and does not evolve in time. However, the same result applies in the Schr\"odinger picture, where $\psi$ now evolves with $s$ according to (\ref{KGQM}) and operators are considered as fixed questions about the system and not evolving (in the basic version of the theorem).  Then
\[  m {\extd\over\extd s}\<\psi|a\psi\> = { m \over\lambda}\<\psi |\rh a- a\rh |\psi\>\]
and for $[\rh,a]$ we use the expressions previously computed as the coefficient of $\theta'$ in the calculation of $\extd a$ and its representation. As $\<\psi|\psi\>$ is a constant,  this also tells us the rate of change of the expectation value $\<a\>:=\<\psi|a|\psi
\>/\<
\psi|\psi\>$.  The Ehrenfest theorem (\ref{ehren1})-(\ref{ehren2}) in the Schr\"odinger picture thus looks the same but now with the $s$ time dependence on the left coming from the state. 

Proceeding in the spacetime Schr\"odinger picture, if we have an eigenvector for the Hamiltonian with eigenvalue $E_{KG}$, say, then each of these evolves by $-\lambda\dot\psi = E_{KG}\psi$ and hence
\[  m {\extd\over\extd s}\<\psi| a| \psi\>=0\]
just as in regular quantum mechanics. In the case of a black hole background and radial wave functions $\psi(r)$, we note the following consequence of the Ehrenfest theorem.

\begin{proposition}\label{propradial} In a Schwarzschild background, if the wave function is differentiable and has only radial dependence $\psi(r)=\psi_1+\imath\psi_2$ for real $\psi_i$  then 
\[   m {\extd  \<\psi| r|\psi\>\over\extd s}= \lambda\imath\int_{r_s}^\infty r(r-r_s)(\psi_1\psi_2'-\psi_2\psi_1')\extd r+{\lambda\over 2}\left[|\psi(r)|^2 r(r-r_s)\right]^\infty_{r_s^+}\]
if the endpoint limits exist. Hence, unitary evolution of $\<\psi|r|\psi\>$ requires that  the second term vanishes, for example if $|\psi|^2r^2\to 0$  as $r\to \infty$ and $|\psi|^2(r-r_s)\to 0$ as $r\to r_s^+$. 
\end{proposition}
\begin{proof} By the Ehrenfest theorem and the calculations in Section~\ref{sec:bh}, since $\del_r\in \CD(M)$ acts as $\lambda{\del\over\del r}$, we have
\[ m   {\extd \<\psi| r|\psi\>\over\extd s}=\lambda\<\psi|(1-{r_s\over r}){\del\over\del r}+{ 2 r-r_s\over 2r^2}|\psi\>.\]
 Then we compute
\begin{align*} \int_{r_s}^\infty \bar\psi r (r-r_s)\psi'\extd r&={1\over 2}\int_{r_s}^\infty r(r-r_s)\del_r |\psi|^2 +\imath\int_{r_s}^\infty r(r-r_s)(\psi_1\psi_2'-\psi_2\psi_1')\extd r\\
&=[{r(r-r_s)\over 2}|\psi|^2]^\infty_{r_s^+}- \<\psi| {2r-r_s\over 2 r^2}|\psi\> + \imath\int_{r_s}^\infty r(r-r_s)(\psi_1\psi_2'-\psi_2\psi_1')\extd r\end{align*}
where we apply integration by parts to the first term. We then insert this back into the Ehrenfest theorem. 
\end{proof}

For example, the `atomic' black hole eigenstates in  Section~\ref{sec:gravatom} (the type (iii) modes) are differentiable at any point just above the horizon, bounded there, and decay exponentially for large $r$, so the endpoints term vanishes and evolution is unitary as expected. Moreover, these modes are real and remain real (times a phase that is indendent of $r$ as they evolve),  and hence $\<r\>$ is a constant as  expected for evolution eigenstates. The endpoints limit condition also appears to be true for the horizon modes arising in the numerical calculations in Section~\ref{sec:dir}, but these are complex so $\<r\>$ does not have to be a constant, and indeed we will find that $\<r\>$ actually increases.

\subsection{Pseudo quantum mechanics in Schwarzschild background}\label{secBH}

Ordinary quantum mechanics arises as an approximation to solutions of the KG equation for a fixed mass $m$ and wave functions which, after factoring out a rest mass mode $e^{-{m\over\lambda}t}$, are slowly varying with respect to some local laboratory time $t$. In this section, we consider something rather different but which nevertheless quite resembles quantum mechanics. To avoid confusion, we will call it `pseudo quantum mechanics'. Namely, we look at the above spacetime Schr\"odinger picture with $\rh$ the spacetime Laplacian (and no external potential), but reduced in the presence of a time-like Killing vector. This extends ideas in \cite{BegMa:geo} to the curved static case. The big difference is that in pseudo-quantum mechanics the `quantum mechanics time' is the geodesic parameter time $s$ as explained above and not the spacetime coordinate $t$. We work in geometric units where the speed of light is $c=1$. We continue to focus on a black hole as representative of our methods.

 The required reduction at the noncommutative geometry level is to restrict to functions independent of $t$ and quotient by the time coordinate case of (\ref{kerrho1}),(\ref{kerrho2}), i.e.
\begin{equation}\label{reduction}  m  \extd t =g^{tt}\theta' \del_t, \quad \extd \del_t=0\end{equation}
to order $\lambda$. As explained in Section~\ref{sec:bh}, the vector field $\del_t\in \CD(M)$ appears in the classical limit as the `energy' $p_t$ and the first of (\ref{reduction}) with $ m $ interpreted as the mass of a particle imposes that $\theta'=\extd  s$ the proper time in the classical approximation and for the chosen signature. This is part of classically imposing (\ref{kerrho1}) whereby
\[ -\extd s^2=g_{\mu\nu}\extd x^\mu\extd x^\nu={\theta'^2\over  m^2}g_{\mu\nu}g^{\mu\alpha}g^{\nu\beta}p_\alpha p_\beta=-\theta'^2\]
for a particle of mass $ m $, but we need only impose it for one of the time coordinates to identify $\theta'$ as proper time, i.e. with the right time dilation factor. We still have a quantum geodesic flow on this reduced algebra that still lands on the Schr\"odinger equation and is now closer to the conventional one for quantum mechanics. This reduced algebra can be elaborated along the lines of  \cite{BegMa:geo},  although we do not do so here as we do not need it explicitly. 

The second of (\ref{reduction}) means we can represent the reduced algebra (and hence the original algebra) on fixed  frequency or more precisely fixed `energy' $p_t$ elements of the form
\begin{equation}\label{eptpsi}  {e^{p_t t\over\lambda}\psi(r,\theta,\phi)},\end{equation}
now with such wave functions also varying in $s$. Here $\del_t\in \CD(M)$ in the noncommutative geometry acts as $\lambda{\del\over\del t}$ and hence has value $p_t$ on the above modes.  Such modes are not normalisable and hence not precisely in $L^2(M)$ but it is useful to include them in the discusison (there are standard ways to deal with this issue more precisely). The associated quantum geodesic flow/spacetime Schr\"odinger equation on these modes then looks like 
\begin{equation}\label{KGVeff} -\lambda{\del \psi\over\del s}:=-{\hbar^2\over 2 m }\square\psi=(-{\hbar^2\over 2 m }\Delta+V_{eff})\psi,\end{equation}
where 
\[  V_{eff}(r):= -(1-{r_s\over r})^{-1}{p_t^2 \over 2 m },\]
\[ \Delta:=(1-{r_s\over r}){\del^2\over \del r^2}+{1\over r}(2-{r_s\over r}){\del\over \del r}+ {1\over r^2}\left({\del^2\over \del \theta^2} + {1\over \sin^2(\theta)}{\del^2\over \del \phi^2}+\cot(\theta){\del\over\del \theta}\right).\]
 This has been set up to  resemble some kind of quantum mechanics for a particle of mass $ m $ on a 3-manifold with the spatial part of the metric  plus an induced radial force potential. Although this only looks like (and isn't) what is normally meant by quantum mechanics, it has the merit that the KG flow/spacetime Schr\"odinger equation is coordinate invariant; we are only choosing to look at it a certain way with respect to a chosen coordinate time (and then extending to allow plane waves in this direction). With this discussion, we are nevertheless led to a precise setting were we can ask for $\psi$ to be $L^2$ on position space and use standard quantum mechanical methods and language. As in quantum mechanics, one can either solve this directly by integrating the first order PDE (\ref{KGVeff}) or one can look for eigenvalues $E_{KG}$ and eigenfunctions of the evolution operator, i.e. such that
 \[ (-{\hbar^2\over 2 m }\Delta+V_{eff})\psi=E_{KG}\psi,\]
 which amounts to (\ref{eptpsi}) solving the KG equation $\square = {m^2_{KG}\over\hbar^2}$ with `square mass' 
 \[ m^2_{KG}=-E_{KG}2 m .\]
 We may potentially be interested in all the eigenvalues $E_{KG}$, not only negative ones, since there is no specific massive  KG field in the picture and we are just using the KG wave operator $\square$ to define the flow.  We illustrate both the direct numerical PDE method and the eigenfunction method, and can consider the latter as 
 \[\left( -{\hbar^2\over2 m }\Delta+ V_{eff}+{p_t^2\over 2 m }\right)\psi=E\psi;\quad  \hbar\omega=-p_t =\sqrt{m_{KG}^2+2 m  E}, \]
for $\psi(r,\theta,\phi)$, where  we subtracted the rest energy to match conventions of the ordinary time-independent Schr\"odinger equations. We will find solutions $\psi_E$ for $E<0$ that are much like those of a hydrogen atom. Also, to align with ordinary quantum mechanics, we will be interested in $p_t\le 0$ or equivalently $\omega\ge 0$ as explained in the introduction, see \cite{BegMa:cos}.  The asymptotic form of solutions of the KG equation in a Schwarzschild background is known analytically  in terms of Whittaker functions \cite{RowSte}, and exact solutions for more general Kerr black holes were noted in \cite{Bez} in terms of Heun functions. These can also be solved for exactly using MATHEMATICA, which is the approach we take. In both cases, graphs  are presented in units with $\hbar=1$.
 
 \subsubsection{Direct integration in the radial case } \label{sec:dir}
 
 The simplest solutions are for $\psi=\psi(r)$ constant in $\theta,\phi$. Then we can solve this numerically see Figures~\ref{figradial1} and~\ref{figradial2}. Calculations are for $r_s=-p_t= m =1$ and are done  numerically for $r\in (1.000001 r_s, 50 r_s)$ with Neumann boundary conditions of zero radial derivative along the horizon edge. Figure~\ref{figradial1} (a) and (b) study the case of 
 an initial Gaussian centred at $10 r_s$ showing complex oscillations in $\psi$ and a gradual diffusion of the probability density $|\psi|^2$. Part (b) the same model in close up nearer the horizon and extending a little further in geodesic time $s$. We see the emergence of further probability density waves when the region of disturbance reaches the horizon, at around $s=0.65$. Whereas parts (a)-(b) have the initial Gaussian centred far from the horizon, part (c) shows evolution of 
an initial Gaussian at $1.4 r_s$, i.e. near to the horizon. It is significant that this is not particularly singular in our set up where our region terminates just above horizon.  

\begin{figure}
 \includegraphics[scale=0.7]
 {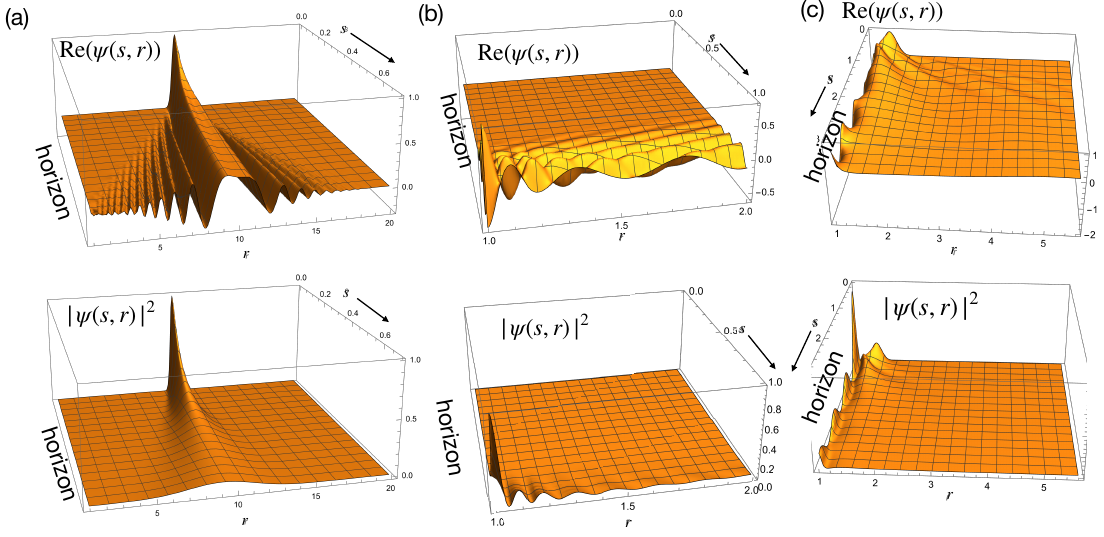}
 \caption{\label{figradial1} (a) Evolution of $\psi(r)$ initially a Gaussian centred far from the horizon at 10 $r_s$ showing complex waves and diffusion of the Gaussian probability density with motion of the peak towards the horizon (b) the same model but in close up near the horizon showing appearance of horizon modes at time $s=0.65$. (c) Evolution of an initial Gaussian centred at $1.4 r_s$ close to the horizon. Units of $r_s=1$.} \end{figure}

Figure~\ref{figradial2} (a)-(b) looks in cross section and in close up at these emergent  `horizon mode' probability density waves (their actual wave function is complex oscillatory). The density waves start very small on the tail of the Gaussian where it interacts with the horizon as shown at $s=0.65$, but by $s=0.9$ they are already twice as high as the peak of the Gaussian, even though most of the probability still resides in the Gaussian off stage at larger $r$. But by $s=3$, there is almost no trace of the original Gaussian as the horizon modes have grown and also increased their wavelength considerably. Part (b) steps back and shows what happens to the Gaussian bump. By $s=1.4$ the oscillations have passed the centre of the Gaussian. 

Note that the peak of the Gaussian bump throughout this process has an apparent motion increasingly rapidly towards the horizon so that by $s=1.1$ it is at $r=7.2 r_s$ and by $s=1.4$ it appears at about $r=4r_s$ in Figure~\ref{figradial2}(b) underneath the probability density oscillations (albeit no longer a Gaussian by this point). The picture is thus of a Gaussian bump `particle' falling into the black  by a process of absorption by waves created at the horizon. This apparent movement of the Gaussian peak towards the horizon is, however, quite a bit faster than a classical geodesic for the same initial velocity $p_t/ m $, as governed by (\ref{unitspeed}) in the form
\begin{equation}\label{classgeo} {\extd r\over\extd s}=\pm\sqrt{{p_t^2\over  m^2}-(1-{r_s\over r})}\end{equation}
Solving this with the same initial point as the initial location of the Gaussian bump, the point particle is only at $9.65 r_s$ at $s=1.1$ and $9.55$ at $s=1.4$ compared to the above. Yet in spite of the inward motion of what used to be the Gaussian peak, the expected value of $\<r\>$  all the while {\em increases} as shown in Figure~\ref{figradial2}(c). This is a somewhat unexpected effect, but what happens is that the horizon modes, while they increase with time in height near the horizon, also have increasingly larger wavelength, which pushes up the expected value of $r$.  

It is tempting to think of the disappearance of the initial Gaussian and its eventual replacement by the horizon modes as a kind of information loss. To this end, we plotted the continuous  entropy $-\<\ln(\rho)\>$  of the associated classical probability density $\rho=|\psi|^2/\<\psi|\psi\>$, which on radial functions amounts to
\begin{equation}\label{entropy} S(\psi)=- \< \ln({ |\psi|^2\over\<\psi|\psi\>})\>=- \int_{r_s}^\infty{ |\psi|^2\over\<\psi|\psi\>}\ln({ |\psi|^2\over\<\psi|\psi\>}) r^2\extd r. \end{equation}
We find in part (d) that this also increases throughput the above process. We similarly looked at the entropy starting with several other  $\R_{\ge 0}$-valued initial wave functions with support away from the horizon and $r_{max}$ (or any fixed phase times such functions) and entropy increasing appears to be a general feature for at least this narrow class, but not for all initial wave functions. There is also a natural relative entropy $S(\rho|\rho')=-\<\ln(\rho/\rho')\>$ where $\rho$ is used to compute the expected value (this is called the Kullback-Leibler divergence \cite{KL} in information geometry). However, this quantity relative to the initial state is too noisy to compute numerically due to the essentially zero probability densities of both parts of the ratio approaching $r_{max}$.

 \begin{figure}
 \includegraphics[scale=0.75]
 {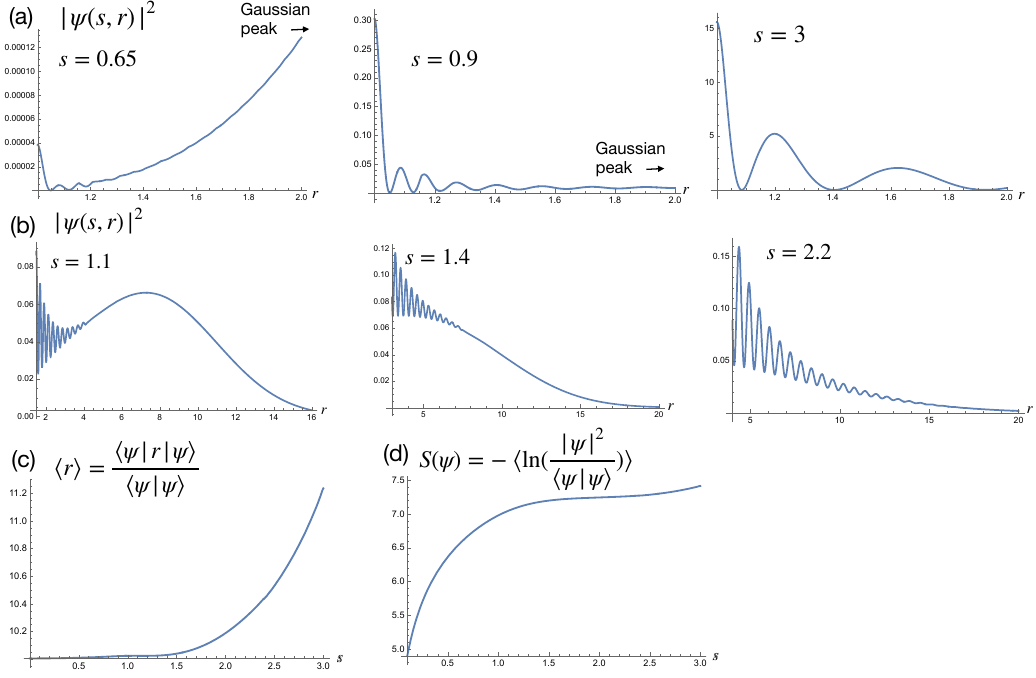}
 \caption{\label{figradial2} (a) Cross-sections of the model in Figure~\ref{figradial1}(a)-(b) showing close-ups of the emergence of probability density waves when the Gaussian tail starts to interact with the horizon, at around $s=0.65$. Note the different scales in the plots. By $s=3$ these horizon modes are all that remain. (b) The same model in larger view showing the Gaussian bump absorbed at $s=1.4$ into the horizon modes.  (c) The expected value $\<r\>$ and (d) the probability density entropy both increase throughout the process.}\end{figure}

 All of our plots are for $s$ before the point where the region of disturbance reaches $r_{max}$ (otherwise one gets a reflection there and interference from this). Integrity of the numerics before that point was assessed by computing $\<\psi|\psi\>$ which indeed remains constant up to numerical noise or systemic errors (of less than around $\pm 1\%$ over the range of $s$ plotted). Moreover, changing $r_{min}$ to ten times closer to the horizon does not visibly change any of the graphs (except for  the highly magnified $s=0.65$ case in Figure~\ref{figradial2}(a) which does not significantly change on making $r_{min}$ twice as close). In particular, the horizon modes do not appear to diverge at the horizon. It should be stressed, however, that the assumption of an initial Gaussian wave function in $r$ is entirely hypothetical and not a physical choice. For example, a particle `Gaussian bump' coming in from $r=\infty$ might be expected to have already have evolved to a complex wave function by the time its region of disturbance reaches radius  $r_{min}< r< r_{max}$ so as to be an initial state for the numerical model.

\subsubsection{Black hole atom case}\label{sec:gravatom}

For large $r$, the potential looks like a $1/r$ potential (shifted by 1) and we can solve for something which for large $r$ is like a hydrogen atom or `gravatom'.   This mirrors the hydrogen-like atom in \cite{BegMa:geo}. The term gravatom has been used in physics for the loose context of gravitationally bound states and these are of potential empirical interest \cite{gravatom}, but we are not aware of any theoretical framework to make this precise in the GR setting needed for a black-hole atom. Although in our case the physical significance of evolution with respect to $s$ is not yet established, when we focus on stationary states with respect to the $s$-evolution, as we do in this section, we do obtain something that looks like  quantum mechanics and is in the same spirit as a quantum-mechanics interpretation of solutions of the Klein Gordon equation proposed in\cite{Don}. We are then able to study gravatom modes in our framework.

We proceed similarly to a hydrogen atom, namely by separation of variables in the eigenvalue equation. Separating out  and solving for the $\phi$ coordinate dependence fixes $p_\phi$ as well as $p_t$ as parameters and we need only consider  eigenstates of the Klein-Gordon wave operator of the form
\[ e^{ p_t t\over\lambda}e^{p_\phi\phi\over\lambda} R(r) F(\theta).\]
The radial equation then separates to 
\begin{equation}\label{Reqn} \left({2 m  E_{KG}\over\hbar^2}  + {p_t^2\over\hbar^2(1-{r_s\over r})}\right) R + \left((1-{r_s\over r}){\del^2\over\del r^2}+{1\over r}(2-{r_s\over r}){\del\over \del r})\right)R={l(l+1)\over r^2}R\end{equation} 
for some constant $l$, and the remaining $\theta$ equation is then
\begin{equation}\label{Feqn} \left({l(l+1)} -{p_\phi^2\over \hbar^2\sin^2(\theta) }\right)F+  \left({\del^2\over \del \theta^2}+\cot(\theta){\del\over \del \theta}\right)F=0.\end{equation}
The latter is the same as for the angular part of the Laplace equation on $\R^3$ and solved as usual for integral $l$ by Legendre polyomials $P_l^{m}(\cos(\theta))$ where $p_\phi=m\hbar$ and $m=0,1,\cdots,l$ (these functions combine with the $e^{p_\phi\phi\over\lambda}$ to  spherical harmonics as usual). So the only difference for us is the radial equation (\ref{Reqn}). 

Note that if we take $r_s=0$ in the 2nd term on the left  of (\ref{Reqn}) and work to order $r_s$ in the 1st term then we obtain the usual equation for an energy $E$ eigenstate of a hydrogen atom, with the correspondence
\[ E_{KG}+{p_t^2\over 2 m }=E,\quad r_s p_t^2={ m  e^2\over 2\pi\hbar^2\eps_0}\]
where $ m $ is the reduced electron mass (that takes into account the mass of the nucleus). $e$ is the electron charge and $\eps_0$ the vacuum permitivity constant. Recall that the ground state of the hydrogen atom is spherical (the wave function is purely radial) and up to normalisation is of the form, with energy
\begin{equation}\label{hatom} \psi(r)= e^{-{r\over a_0}},\quad  a_0=   {4\pi\hbar^2 \eps_0\over  m  e^2},\quad E=-{\hbar^2\over 2 m  a_0^2}.\end{equation}
 Our case is more complicated, but  we can  expect some similarity in view of the above. 

Some solutions are shown in Figure~\ref{figatom} for $r_s=p_t= m =1$ and $p_\phi=0$ (the higher spin modes follow a similar pattern). In this case $\psi$ again depends only on the radius. The radial equation depends critically on $E_{KG}$ and we find that:

(i) For $E_{KG}\ge -{p_t^2\over 2 m }$, there are real oscillatory modes which are well approximated for  large $r$ by 
\[  \psi(r)\sim  {\sin(\alpha r+\beta)\over r}.\]
Up to normalisation, there is a free boundary condition resulting in a phase shift of the form $\beta$ as stated for large $r$.  

(ii) For $E_{KG}< -{p_t^2\over 2 m }$, solutions typically diverge exponentially to $\pm\infty$ at large $r$,
\[ \psi(r)\sim {e^{\alpha r}\over r}.\]

(iii) For $E_{KG}<  -{p_t^2\over 2 m }$ and carefully chosen initial conditions, the mode in (ii) can be suppressed leaving solutions approximately of the form for large $r$,
\[ \psi(r)\sim {e^{-\alpha r}}.\]

The case (iii) has a finite norm  but the other two are not normalisable with respect to the same $r^2\extd r$ measure as above, due to large $r$ contributions. Remarkably, all solutions are non-singular as $r\to r_s$. The large $r$ frequency/exponential factor is
\[ \alpha={\sqrt{|2 m  E_{KG}+p_t^2|}\over\hbar}.\]

\begin{figure}
 \includegraphics[scale=0.65]
 {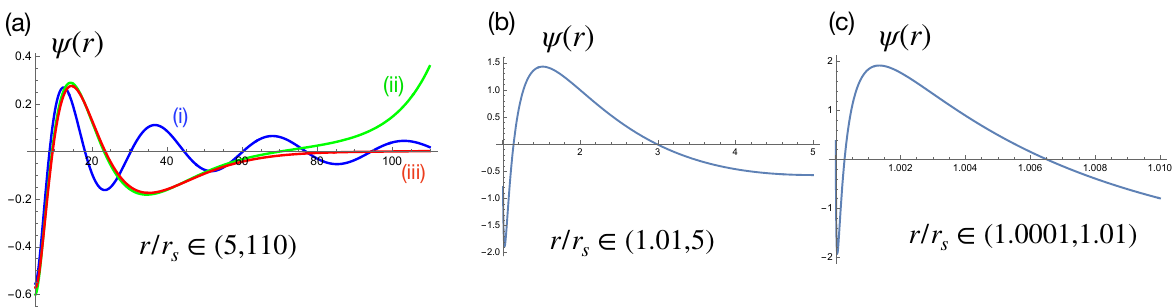}
 \caption{\label{figatom} Spherically symmetric $l=0$ evolution eigenfunctions for the pseudo gravatom with $p_t^2/2 m =0.5$. (a) shows oscilliatory mode (i) with eigenvalue $E_{KG}=-0.49$, exponentially divergent mode (ii) with $E_{KG}=-0.51$ and exponentially decaying `atomic' mode (iii) with $E_{KG}=-0.51$. (b) and (c) shows the fractal nature of all three modes approaching the horizon, where successive close ups look the same. } \end{figure}

We see that the pseudo-gravatom wave functions for $l=0$ match the usual hydrogen atom at large $r$, based on the `atom-like' type (iii) modes, but their behaviour near the horizon is completely different. Namely in the figure parts (b),(c),  the modes are shown again in close up (the type (iii) `atomic' mode is plotted but the other two increasingly coincide as we near the horizon). We see that the even more close-ups look the same, a phenomenon that persists on iterating more close-ups all the way down to machine precision. Thus the solutions, while bounded and not divergent, oscillate infinitely quickly as $r\to r_s$ and acquire a fractal nature. This is to be expected due to the time dilation approaching the horizon. Note that the horizon modes in Section~\ref{sec:dir}, while they look superficially like type (i) here, are not eigenstates. In fact, they are complex and decay much faster (namely what appears to be more like $1/r^2$ at large $r$).

Thus, the probability density $|\psi(r)|^2$ of the $l=0$ modes, unlike the case of a hydrogen atom, does not simply decay but rather, approaching the horizon, forms bands of increasingly small separation. In the example, we see these density peaks for the `atomic' type (iii) mode at:
\[ 35,\quad  15,\quad  5,\quad 1.5,\quad 1.03,\quad 1.0013,\quad 1.00006,\quad \cdots \]
in units of $r_s$.  This banding will be present in any coordinate system. Banding, i.e. the wave function crossing zero, is a feature of some higher $l$ modes in the case of the hydrogen atom, but we see it already and in a fractal form here. Indeed, the radial structure of the modes for small $l>0$ in our case appear to qualitatively identical to the $l=0$ case, while the  angular structure of the higher $l$ modes is the same as for a hydrogen atom.  Note, however, that neither $E$ nor $E_{KG}$ are  forced to be quantised. For the hydrogen atom, the exponential form $\psi(r)=e^{-{r\over a_0}}$ for the ground state in (\ref{hatom}) implies the stated value of $a_0$ to avoid a divergence in the eigenvalue equation at $r=0$, while the stated relation between $E$ and $a_0$ comes from the eigenvalue equation  at large $r$ (and together they fix the discrete value of $E$ for this type of mode). In our case we have an analogue of the large $r$ restriction, but not of the small $r$ restriction.

\section{Concluding remarks}\label{sec:con} 

We introduced a new technique for doing some form of quantum mechanics on curved spacetimes based on the algebra $\CD(M)$ of differential operators, as generated by complex smooth functions and vector fields on the manifold with the rule that they do not commute, so $[X,f]=\lambda X(f)$ where $\lambda=-\imath\hbar$. The key feature is that apart from being careful about the ordering, we can work with the usual objects of tensor calculus familiar to physicists working in GR, in contrast to other approaches to quantum theory on curved spacetimes where the focus is on operator algebras and functional analysis\cite{HolWal,Kay}. By separating these aspects, we were able to write down $\CD(M)$-level first order versions of the geodesic equations as setting to zero the kernel expressions in (\ref{keromega1}) and  Proposition~\ref{hdkk}. These equations are generally covariant, being defined by covariant geometric objects or where not  apparently covariant, behaving correctly after one accounts for the above noncommutativity in patching coordinate charts or changing bases (Remark~\ref{rembasis}).  Moreover, our local versions  (\ref{kerrho1})-(\ref{kerrho2}) land on the classical geodesic equation when $\lambda\to 0$, while at order $\lambda$ they contain information which, when $\CD(M)$ is represented as operators, encodes the Klein-Gordon equation with potential. At order $\lambda^2$, we also saw the Ricci tensor in the commutation relations in Propositions~\ref{geid7} and~\ref{geid}. 

We also gained a novel way of thinking about the need for these equations, i.e. why classical and now quantum particles move along geodesics in the first place, as due  in our context to the presence of gravity (curvature) and a resulting nonassociativity obstruction that geodesic motion kills. This is the content of Corollary~\ref{corbimod}. The full explanation here lies in the noncommutative differential geometry of $\CD(M)$ as an algebra, which in the curved case (in this paper) we were not able to fully understand but where the flat space ordinary Heisenberg algebra version is \cite{BegMa:geo}. However, by moving $\theta'=\extd s$ to the denominator and interpreting $\extd x^\mu, \extd p_\nu$ as rates of change, the quantum geodesic equations and associated commutation relations can be used without knowing their origin in noncommutative geometry and with a possibility of new effects at order $\lambda^2$ which could be investigated.

On the mathematical side, although we did not aim to develop the higher order theory other than to compute the breakdown of the Jacobi identity at order $\lambda^2$, there is a precedent in the use of $L_\infty$ and homotopy algebra methods to describe field theory in the presence of interactions, see e.g. \cite{Jur} for a review. Possibly the higher orders could be treated order by order motivated by such methods. In special cases, it might also be possible to cast the exterior algebra as quasiassociative, i.e. associative in a nontrivial monoidal category. This is another direction for further work and could connect to cochain twist methods \cite{BegMa:twi}. Finally, while the focus of the paper has been on $\CD(M)$ constructed at a smooth level, there are many interesting issues as operators in the Schr\"odinger representation. The image of $\CD(M)$ would appear to qualify as an $O^*$-algebra in the sense of \cite{Sch}, and we also note an extensive literature around metrics on phase space and the Weyl-H\"ormander calculus for the quantisation, see for example \cite{Pro}. A  Poisson-level setting behind the flat spacetime/Heisenberg algebra case is in Appendix A of \cite{BegMa:geo} and a version of it may apply in the $\CD(M)$ case as a point of contact.

Next, the parameter $s$ that we used in the spacetime case has its origins as  geodesic proper time and by extension makes sense also at the `Heisenberg picture' level of evolution in $\CD(M)$. If we are happy with this when $\hbar=0$ to describe the motion of  a classical particle, we {\em should} be somewhat happy with it in some kind of collective sense for motion of a quantum particle even if what exactly we mean by `proper time' is at this point a little fuzzy. We believe that  this interpretational conundrum does not need to be answered immediately but could emerge from applications and experience over time. Until then, it is fair to say that quantum evolution with respect to $s$ is of mathematical origin and only quantum mechanics-like, but with hints at potential physics. In particular, when $M$ is spacetime, we are not doing quantum mechanics on curved spacetime in the sense pioneered in \cite{Witt}.  We also saw that our Heisenberg picture was equivalent to a Schr\"odinger picture where wave-functions are now on spacetime and $s$ is now the time of a  Schr\"odinger-like evolution (\ref{KGQM}). We looked at this in Section~\ref{sec:app} for `Hamiltonian'  (\ref{minkh}), where the wave operator $\square$ replaces the role of spatial Laplacian and we took external potential $V=0$. We also referred to this as a  `Klein-Gordon flow'. We are led to this point of view even though wave functions on spacetime are not usually considered in Physics (but see for example  \cite{Fan}).  Moreover, even if one discounts any physical role of $s$, one can still view (\ref{KGQM})  as a tool to study geometry on a (pseudo)-Riemannian manifold on a par with heat-kernel expansions or Ricci flow methods.  As we are not solving for a fixed mass, this Klein-Gordon flow  captures off-shell information as a self-contained step in the direction of quantum field theory (which also deals with off-shell modes in computing Feynman integrals) without being quantum field theory, and could be explored further even if it is just a tool.  We exhibited a class of such flows in Section~\ref{sec:dir} around a Schwarzschild background and saw from numerical work that their probability density entropy (\ref{entropy}) increases with $s$ in the models that we looked at, as in Figure~\ref{figradial2}. The process here suggests a comparison with information being lost on conversion to Hawking radiation when matter falls into a black hole; such ideas could be explored further in our parallel setting.  Another question is that the formalism of quantum geodesics which we applied to $\CD(M)$ can also be applied directly to $M$ itself as a quantum geodesic flow on this. If these latter $\psi$ are real and positive then this is not really different from a density flow (where each particle moves on its own geodesic), but when $\psi$ is complex then new phenomena are potentially possible. Such quantum geodesic flows directly on $M$ are different from the Klein-Gordon flows in the present paper, but both could be looked for other spacetimes of interest. 

In the final Section~\ref{sec:gravatom} of the paper, we connected with another point of view on what could be called quantum mechanics on curved spacetime, namely solutions of the wave or Dirac equation that are interpreted as quantum-mechanics-like with respect to a preferred time direction (then there is no external time $s$ and `wave functions' are on space). This is not new and for black hole backgrounds appears first to have been considered in \cite{Don}. General solutions for this and more generally for Kerr backgrounds are also known \cite{Bez}. In our case such Klein-Gordon solutions appear as stationary states for the Klein-Gordon flow and we used methods motivated by analogy with time independent quantum mechanics, finding `gravatom' states with banding in the probability density even for orbital angular moment $l=0$  and of a fractal nature as we saw in Figure~\ref{figatom}. In our companion paper\cite{BegMa:cos} we do a similar analysis for FLRW cosmologies, this time using similar methods to those  used for quantum mechanical tunneling now  applied to solve the Klein-Gordon equation through a period of rapid inflation. It would be interesting to look at other spacetimes in the same vein. 

Last but not least, while the entire paper as well as \cite{BegMa:geo, BegMa:cos} are about applying noncommutative geometry to  quantum mechanics with classical space or spacetime,  the same formalism can be applied to the case where the space or spacetime is noncommutative or discrete (the latter also falls within noncmmutative geometry with finite-difference 1-forms not commuting with functions). Direct quantum geodesic applications are in \cite{LiuMa, BegMa:cur, BegMa:geogra} respectively, while quantum mechanics on noncommutative spaces via the methods of the present paper remains to be considered.

\section*{Declarations}

\noindent{\bf Funding:}  SM was supported by a Leverhulme Trust project grant RPG-2024-177. 

\medskip
\noindent{\bf Data availability:} Data sharing is not applicable as no data sets were generated or analysed during the current 
study.

\medskip
\noindent{\bf Conflict of Interest:} The authors have no competing  interests to declare that are relevant to the content of this article.

\end{document}